\documentclass[reqno,a4paper,11pt]{amsart}
\usepackage[utf8]{inputenc}

\calclayout

\usepackage[utf8]{inputenc}
\usepackage{verbatim}
\usepackage[english]{babel}
\usepackage{amsmath}
\usepackage{cite}
\makeatletter
\def\thm@space@setup{%
  \thm@preskip=12pt plus 3pt minus 3pt
  \thm@postskip=\thm@preskip 
}
\makeatother

\usepackage{amsfonts,amssymb}
\usepackage{mathtools}
\mathtoolsset{showonlyrefs}
\usepackage{upgreek}
\usepackage{amsthm}
\usepackage{bm}
\usepackage{float}
\usepackage{color}
\usepackage{tikz}
\usetikzlibrary{decorations.pathreplacing,decorations.markings,calc}
\usetikzlibrary{arrows}
\tikzset{
  on each segment/.style={
    decorate,
    decoration={
      show path construction,
      moveto code={},
      lineto code={
        \path [#1]
        (\tikzinputsegmentfirst) -- (\tikzinputsegmentlast);
      },
      curveto code={
        \path [#1] (\tikzinputsegmentfirst)
        .. controls
        (\tikzinputsegmentsupporta) and (\tikzinputsegmentsupportb)
        ..
        (\tikzinputsegmentlast);
      },
      closepath code={
        \path [#1]
        (\tikzinputsegmentfirst) -- (\tikzinputsegmentlast);
      },
    },
  },
  mid arrow/.style={postaction={decorate,decoration={
        markings,
        mark=at position .5 with {\arrow[#1]{stealth}}
      }}},
}

\tikzset{
	on each segment/.style={
		decorate,
		decoration={
			show path construction,
			moveto code={},
			lineto code={
				\path [#1]
				(\tikzinputsegmentfirst) -- (\tikzinputsegmentlast);
			},
			curveto code={
				\path [#1] (\tikzinputsegmentfirst)
				.. controls
				(\tikzinputsegmentsupporta) and (\tikzinputsegmentsupportb)
				..
				(\tikzinputsegmentlast);
			},
			closepath code={
				\path [#1]
				(\tikzinputsegmentfirst) -- (\tikzinputsegmentlast);
			},
		},
	},
	mid arrow/.style={postaction={decorate,decoration={
				markings,
				mark=at position .7 with {\arrow[#1]{stealth}}
	}}},
	rmid arrow/.style={postaction={decorate,decoration={
				markings,
				mark=at position .3 with {\arrowreversed[#1]{stealth}}
	}}},
	vmid arrow/.style 2 args={postaction={decorate,decoration={
				markings,
				mark=at position #2 with {\arrow[#1]{stealth}}
	}}},
	vrmid arrow/.style 2 args={postaction={decorate,decoration={
				markings,
				mark=at position #2 with {\arrowreversed[#1]{stealth}}
	}}},
}

\makeatletter
\def\grd@save@target#1{%
  \def\grd@target{#1}}
\def\grd@save@start#1{%
  \def\grd@start{#1}}
\tikzset{
  grid with coordinates/.style={
    to path={%
      \pgfextra{%
        \edef\grd@@target{(\tikztotarget)}%
        \tikz@scan@one@point\grd@save@target\grd@@target\relax
        \edef\grd@@start{(\tikztostart)}%
        \tikz@scan@one@point\grd@save@start\grd@@start\relax
        \draw[minor help lines] (\tikztostart) grid (\tikztotarget);
        \draw[major help lines] (\tikztostart) grid (\tikztotarget);
        \grd@start
        \pgfmathsetmacro{\grd@xa}{\the\pgf@x/1cm}
        \pgfmathsetmacro{\grd@ya}{\the\pgf@y/1cm}
        \grd@target
        \pgfmathsetmacro{\grd@xb}{\the\pgf@x/1cm}
        \pgfmathsetmacro{\grd@yb}{\the\pgf@y/1cm}
        \pgfmathsetmacro{\grd@xc}{\grd@xa + \pgfkeysvalueof{/tikz/grid with coordinates/major step}}
        \pgfmathsetmacro{\grd@yc}{\grd@ya + \pgfkeysvalueof{/tikz/grid with coordinates/major step}}
        \foreach \x in {\grd@xa,\grd@xc,...,\grd@xb}
        \node[anchor=north] at (\x,\grd@ya) {\pgfmathprintnumber{\x}};
        \foreach \y in {\grd@ya,\grd@yc,...,\grd@yb}
        \node[anchor=east] at (\grd@xa,\y) {\pgfmathprintnumber{\y}};
      }
    }
  },
  minor help lines/.style={
    help lines,
    step=\pgfkeysvalueof{/tikz/grid with coordinates/minor step}
  },
  major help lines/.style={
    help lines,
    line width=\pgfkeysvalueof{/tikz/grid with coordinates/major line width},
    step=\pgfkeysvalueof{/tikz/grid with coordinates/major step}
  },
  grid with coordinates/.cd,
  minor step/.initial=.2,
  major step/.initial=1,
  major line width/.initial=2pt,
}

\usepackage{graphicx}
\usepackage{subcaption}
\usepackage{enumerate}
\usepackage{bm}
\usepackage{seqsplit}

\usepackage[pdftex,colorlinks]{hyperref}
\hypersetup{
	colorlinks=true,
	linkcolor=blue,
	citecolor=red
}

\makeatletter
\def\l@subsection{\@tocline{2}{0pt}{2.5pc}{5pc}{}}

\DeclareMathOperator{\re}{Re}
\DeclareMathOperator{\im}{Im}
\DeclareMathOperator{\ee}{\rm e}
\DeclareMathOperator{\supp}{supp}
\DeclareMathOperator{\tr}{Tr}

\newcommand{\C}{\mathbb{C}}
\newcommand{\R}{\mathbb{R}}
\newcommand{\Z}{\mathbb{Z}}

\newcommand{\E}{\mathbb{E}}
\newcommand{\boh}{\mathit{o}}
\newcommand{\Boh}{\mathcal{O}}

\newcommand{\ii}{\mathrm{i}}
\newcommand{\dd}{\mathrm{d}}
\newcommand*{\deff}{\mathrel{\vcenter{\baselineskip0.5ex \lineskiplimit0pt
                     \hbox{\scriptsize.}\hbox{\scriptsize.}}}%
                     =}
\newcommand*{\revdeff}{=\mathrel{\vcenter{\baselineskip0.5ex \lineskiplimit0pt
                     \hbox{\scriptsize.}\hbox{\scriptsize.}}}%
                     }
\DeclareMathOperator{\Li}{Li}

\renewcommand{\bm}{\mathbf}
\newcommand{\mcal}{\mathcal}

\newcommand{\msf}{\mathsf}

\newcommand{\wh}{\widehat}
\newcommand{\wt}{\widetilde}

\renewcommand{\sp}{\boldsymbol \sigma}

\newcommand{\aad}{\msf a}
\newcommand{\bad}{\msf b}
\newcommand{\cad}{\msf c}

\newcommand{\tad}{\msf t}
\newcommand{\sad}{\msf s}
\newcommand{\uad}{\msf u}
\newcommand{\xad}{\msf x}

\newcommand{\gt}{\tau}

\newcommand{\bes}{\mathrm{(Bes)}}

\usepackage{xcolor}
\definecolor{lightgreen}{RGB}{100,200,100}
 \usepackage[backgroundcolor=lightgreen]{todonotes}

\newtheorem{theorem}{Theorem}[section]
\newtheorem{prop}[theorem]{Proposition}
\newtheorem{lemma}[theorem]{Lemma}
\newtheorem{corollary}[theorem]{Corollary}

\theoremstyle{definition}
\newtheorem{definition}[theorem]{Definition}

\newtheorem{rhp}[theorem]{RHP}
\newtheorem{assumption}[theorem]{Assumption}

\newtheoremstyle{remarkstyle}
  {3pt}{3pt}{\normalfont}{}{\itshape}{.}{.5em}{}

\theoremstyle{remarkstyle}
\newtheorem{remarkinner}[theorem]{Remark}
\numberwithin{equation}{section}

\newenvironment{remark}
  {\begin{remarkinner}}
  {
  \hspace*{\fill}$\triangleright$\end{remarkinner}
  \vspace{5pt}
  }

\newcommand\restr[2]{{
		\left.\kern-\nulldelimiterspace 
		#1 
		\vphantom{\big|} 
		\right|_{#2} 
}}

\begin{document}

\title[Conditional thinning and multiplicative stats of LU-type ensembles]{Conditional thinning and multiplicative statistics of Laguerre-type orthogonal polynomial ensembles}

\author[L.~Molag]{Leslie Molag}
\address[LM]{Carlos III University of Madrid, Spain.}
\email{lmolag@math.uc3m.es}

\author[G.~Silva]{Guilherme L.~F.~Silva}
\address[GS]{Universidade de S\~ao Paulo, Brazil.}
\email{silvag@icmc.usp.br}

\author[L.~Zhang]{Lun Zhang}
\address[LZ]{Fudan University, China}
\email{lunzhang@fudan.edu.cn}

\date{}


\begin{abstract}
We study the local statistics of orthogonal polynomial ensembles near a hard edge, subject to a multiplicative deformation of the measure. Probabilistically, this deformation corresponds to a position-dependent conditional thinning of the particles. We prove that, under critical hard edge  scaling and for a large class of potentials and deformation symbols, the correlation kernel of the conditional ensemble converges to a universal limit, which we identify as the conditional thinned Bessel point process.

We derive an explicit expression for this limiting kernel in terms of the solution to a nonlocal integrable system depending on a parameter. For a special choice of the parameter, this system was recently identified in the study of multiplicative statistics of the Bessel point process. Our results establish that this system governs the full correlation structure of the conditional Bessel point process, extending the classical connection between the standard Bessel kernel and the Painlevé V equation. 
\end{abstract}


\vspace*{-1.6cm}

\maketitle

\setcounter{tocdepth}{2} \tableofcontents 

\section{Introduction}

The universality of spectral statistics is a cornerstone of Random Matrix Theory (RMT). From its origins in nuclear physics to modern applications in number theory, statistical mechanics, and high-dimensional data analysis, the local correlation functions of eigenvalues have proven to be robust invariants across a vast array of models, and a rich ground for the prediction of new critical phenomena.

Among the many striking connections between random particle systems and other fields, the link to integrable systems has proven exceptionally fruitful. This relationship is most celebrated at the ``soft edge" of the spectrum, that is, close to the largest eigenvalue. The fluctuations of the largest eigenvalues in Hermitian random matrix models are described by the Airy$_2$ point process, and the distribution of the largest particle is the Tracy-Widom distribution $F_2$. In turn, $F_2$ is governed by the Hastings-McLeod solution to the Painlevé II equation, a second order nonlinear integrable ODE. Such a connection between extremal eigenvalues and nonlinear differential equations was pioneered by Tracy and Widom \cite{tracy_widom_level_spacing} and Forrester \cite{Forrester_1993} for the Gaussian Unitary Ensemble (GUE). These results were subsequently extended to a wide variety of random matrix models using diverse techniques \cite{pastur_shcherbina_universality, deift_book, deift_gioev_universality_edge, deift_its_zhou_riemann_hilbert_random_matrices_integrable_statistical_mechanics, soshnikov_1999}.

But the connection between strongly correlated particle systems and integrable systems extends well beyond RMT. The distribution $F_2$ describes the fluctuations of the longest increasing subsequences in random permutations \cite{baik_deift_johansson, Romik2015} and, more generally, characterizes the universal scaling limit of growth models in the Kardar-Parisi-Zhang (KPZ) universality class \cite{AmirCorwinQuastel2011, Sasamoto2010a}.

In recent years, attention has expanded from standard gap probabilities to the study of multiplicative statistics of these point processes. It was discovered that the generating functions for multiplicative statistics of the Airy$_2$ process relate directly to the solution of the KPZ equation \cite{AmirCorwinQuastel2011, Sasamoto2010a, BorodinGorin2016, Calabrese2010a, Dotsenko2010}, and that the underlying point process, when weighted by these statistics, forms a new ``conditional" ensemble with a rich integrable structure \cite{ClaeysGlesner2021, CafassoClaeys2022, CafassoClaeysRuzza2021}. Furthermore, this framework has become a powerful tool for extracting fine-grained probabilistic data, such as large deviation estimates and tail asymptotics \cite{CorwinGhosal2020a, CafassoClaeys2022, CafassoClaeysRuzza2021, Tsai2018}, while also motivating recent studies on higher-order soft edge scalings \cite{CafassoTarricone2023, BothnerCafassoTarricone2022, CafassoPinheiro2025, BothnerSheperd2025, ClaeysRuzzaTarricone2024, DeanLeDoussalMajumdarSchehr2019, LeDoussalMajumdarSchehr2018}.

The developments mentioned so far concern RMT fluctuations around a so-called ``soft edge", and a parallel and equally rich structure exists at the ``hard edge" of the spectrum. A hard edge arises when eigenvalues are bounded by a strict physical constraint, most commonly the origin for positive definite matrices. The canonical model for this regime is the Wishart ensemble (or Laguerre Unitary Ensemble - LUE), consisting of matrices of the form $XX^*$, where $X$ is a rectangular matrix with complex Gaussian entries \cite{Wishart28}. As the matrix size grows, the eigenvalues near the origin, when critically scaled, converge to the Bessel point process, a canonical point process whose correlations are described by Bessel functions \cite{Forrester_1993, KonigOConnell2001}.

Fluctuations of eigenvalues near a hard edge also display a rich structure in connection with integrable systems.
Historically, the gap probabilities for the Bessel process were connected to the Painlevé V equation by Tracy and Widom \cite{TracyWidom1994Bessel}, see also \cite{NagaoForrester1995}. This connection was further extended to generating functions and multi-interval statistics over the years \cite{DaiXuZhang2018, AtkinCharlierZohren2018, CharlierDoeraene2019}. However, in light of the recent developments at the soft edge, it is natural to ask: To what extent does the rich integrable structure of multiplicative statistics carry over at a hard edge?

In the present work, we provide an affirmative answer to this question by explicitly constructing the kernel of a limiting point process associated with these multiplicative statistics, and describing it in terms of integrable systems.

We begin by considering the eigenvalues of positive definite matrices of size $n$. We subject these eigenvalues to a conditional thinning procedure, where each eigenvalue is retained or removed independently with a position-dependent probability, and the final process is the resulting ensemble conditioned that all particles have been retained. This procedure defines a finite-$n$ conditional ensemble. Our main result describes the asymptotic behavior, as $n\to \infty$, of the conditional ensemble in the critical regime where eigenvalues accumulate near the hard edge.

We prove that, under the hard-edge scaling, the correlation kernel of the finite conditional ensemble converges to a universal limit. Our results are universal, in the sense that we establish them for a broad class of random matrix models with a hard edge, as well as for a wide variety of thinning profiles. We identify this limit as the conditional thinned Bessel point process, which corresponds to applying the conditional thinning operation directly to the classical Bessel point process. Furthermore, our asymptotic analysis yields an explicit formula for the correlation kernel of this process in terms of the solution to a nonlocal integrable system.

The integrable structure of multiplicative statistics for the Bessel point process was recently investigated by Ruzza \cite{Ruzza2024}, who showed that they are governed by a specific nonlocal integrable equation, which recovers the classical Painlevé V equation in the degenerate limit of gap probabilities. In our analysis, we explicitly identify the nonlocal system characterizing our limiting kernel with the equation derived by Ruzza. This identification is crucial: it shows that the solution to this nonlocal equation describes not just the multiplicative statistics, but the full correlation kernel of the conditional ensemble itself. Thus, the nonlocal Painlevé V transcendent emerges as the fundamental structural invariant of the conditionally thinned Bessel process.

This result completes the description of conditional thinning statistics across the standard spectral regimes. The kernel for the conditional thinning ensemble was related to nonlocal integrable systems at a regular soft edge in \cite{GhosalSilva22}, and at a regular bulk point in \cite{CandidoEtal2026} (see also \cite{ClaeysSilva25} for further discussion, and \cite{CafassoClaeysRuzza2021, ClaeysTarricone2024} for related multiplicative statistics). Our work establishes the analogous theory for regular hard edges.

Our analysis relies on the integrable structure of the ensembles via orthogonal polynomials (OPs) and the associated Riemann-Hilbert Problem (RHP) approach. The conditional thinned point process of an orthogonal polynomial ensemble is itself an orthogonal polynomial ensemble, defined by a multiplicative deformation of the original weight \cite{ClaeysGlesner2021}. This observation serves as the starting point for our RHP formulation of the finite $n$ correlation kernel.

While we apply the Deift-Zhou steepest descent method, our analysis deviates from the standard route in the construction of the local parametrix near the hard edge. Unlike the classical case, the jump matrices for this local parametrix cannot be reduced to piecewise constant matrices. Consequently, the resulting model problem involves nontrivial position-dependent jumps that encode the symbol of the deformation. This model problem requires a separate asymptotic analysis, and ultimately it establishes the connection with the integrable systems that emerge.

Finally, armed with the asymptotics for the conditional correlation kernel, we recover the multiplicative statistics using a recently established general deformation formula \cite{GhosalSilva25}. This approach bypasses the traditional route using differential identities for Hankel determinants.

\section{Statement of results}

The general family of models we consider consists of deformations of the Laguerre-type OP ensembles. More precisely, we consider random points $x_1,\hdots, x_n\in (0,\infty)$ 
with joint distribution 
of the form
\begin{equation}\label{eq:Lagdef}
\frac{1}{\msf Z_n(\sad)}\prod_{1\leq i<j\leq n}|x_i-x_j|^2 \prod_{j=1}^n \omega_n(x_j\mid \sad)\, \dd x_1\cdots \dd x_n,
\end{equation}
where 
\begin{equation}\label{eq:deffweight}
\omega_n(x)=\omega_n(x\mid \sad)\deff \sigma_n(x) x^\alpha \ee^{-nV(x)}, \quad \alpha>-1, \quad x>0,
\end{equation}
with
\begin{equation}\label{def:sigman}
\sigma_n(x)=\sigma_n^Q(x\mid \sad)\deff \frac{1}{1+\ee^{-\sad-n^{2m}Q(x)}},\quad \sad \in \R, \quad m\in \Z_{>0} \text{ fixed}, 
\end{equation}
and the normalization constant
$$
\msf Z_n(\sad)\deff \int_0^\infty\cdots \int_0^\infty \prod_{j<k}(x_k-x_j)^2 \prod_{j=1}^n \omega_n(x_j\mid \sad) \dd x_1\cdots \dd x_n,
$$
also known as the partition function. In the formula above, the function $\sigma_n$ should be seen as a deformation factor, and $\omega_n$ is thus viewed as a deformation of the weight $x^\alpha\ee^{-nV(x)}$. We refer to $\sigma_n$ as the symbol of the deformation, depending parametrically on $\sad \in \R$. Further conditions on $\sigma_n$ will be placed, and it is in fact more convenient to place these conditions on the function $Q$ itself. We comment more on the meaning of these quantities in a moment.

Still about \eqref{eq:deffweight}, we assume that $V$ satisfies the following conditions.
\begin{assumption}\label{deff:condQ} 
The potential $V$ in \eqref{eq:deffweight} is a polynomial with positive leading coefficient, and is such that the equilibrium measure $\dd\mu_V(x)=\mu'_V(x)\dd x$ over the interval $[0,+\infty)$ in the external field $V$ is supported on a single interval $[0,a]$. Moreover, we assume that $\mu_V$ is regular with a (regular) hard edge at $x=0$ and a (regular) soft edge at $x=a$.
\end{assumption}

The measure $\mu_V$ is the limiting global distribution of particles of the model \eqref{eq:Lagdef} as $n\to \infty$. When $\sigma_n\equiv 1$, this is a known result, and for $\sigma_n$ as considered here this result is a consequence of our analysis. 

For the reader unfamiliar with potential theory, we introduce and discuss the equilibrium measure $\mu_V$ in a more detailed manner in Section~\ref{sec:eqmeasure}. The regularity conditions on $\mu_V$ that we work with are the analogue of the traditional one-cut regular assumptions for equilibrium measures on $\R$, and are standard in random matrix theory. We recall them in detail in Assumption~\ref{assump:onecutregular} below. This regularity holds for a large family of potentials $V$, for instance when $V$ is convex, the case $V(x)=x$ (the Laguerre weight) being the prototypical example. The most relevant part for the coming discussion is the regularity assumption at the hard edge $x=0$: this condition means that the density $\mu'_V$ of $\mu_V$ blows up as a square-root at the origin,
\begin{equation}\label{eq:localbehdensityeqmeasintro}
\mu'_V(x)=\kappa_0 |x|^{-1/2}(1+\boh(1)),\quad x\searrow 0,\quad \text{ for some constant }\kappa_0>0.
\end{equation}

Our focus in the current paper is in describing the effect of the deformation $\sigma_n$ on eigenvalues near $x=0$, that is, near the (regular) hard edge. As such, many of our local scaling limits will still hold if we drop the regularity condition on the soft edge $x=a$, or also if we drop the assumption that $\supp\mu_V$ is connected. Nevertheless, we work under such conditions to simplify our analysis, as our main goal is not to work under the most general assumptions, but instead to compute the novel scaling limits that arise at the hard edge, and showcase their universality. 


By \eqref{def:sigman}, it is readily seen that  $\sigma_n(x)\to 1$ as $\sad\to +\infty$, and \eqref{eq:Lagdef} reduces to the classical Laguerre-type OP ensemble. From now on, whenever we refer to the classical Laguerre-type OP ensemble, we always mean \eqref{eq:Lagdef} with the choice $\sad=+\infty$, that is, $\sigma_n\equiv 1$. With this degeneration in mind, the quotient
\begin{equation}\label{eq:MultStats}
\msf L_n^Q(\sad)\deff \frac{\msf Z_n(\sad)}{\msf Z_n(+\infty)}=\E\left[ \prod_{j=1}^n\sigma_n(x_j\mid \sad) \right]
\end{equation}
is a multiplicative statistic of the classical Laguerre-type OP ensemble, associated to the symbol ${\sigma_n(x\mid \sad)}$.

Alongside the interest in the multiplicative statistics $\msf L_n^Q(\sad)$, the model \eqref{eq:Lagdef} may also be interpreted as a deformed ensemble obtained when performing a conditional thinning procedure. Starting from the classical Laguerre-type OP ensemble with $n$ particles $x_1,\hdots,x_n$, each particle receives a mark $1$ with probability $\sigma_n(x_j)$, and a mark $0$ with complementary probability $1-\sigma_n(x_j)$. The mark $1$ is interpreted as if the marked particle is ``seen" in the system, and the mark $0$ as if the particle is ``lost" by noise in the sample. This marking interpretation corresponds to the thinning step. Now construct a new particle system, the conditional thinned ensemble, obtained by conditioning that all the particles received the mark $1$. A general framework developed by Claeys and Glesner \cite{ClaeysGlesner2021} shows that \eqref{eq:Lagdef} is precisely the density of particles in this conditional thinned ensemble starting from the corresponding Laguerre-type OP ensemble.

With the conditional thinned interpretation we just explained in mind, it is natural to expect that statistics of \eqref{eq:Lagdef} may depend on properties of the function $Q$ labelling $\sigma_n$.
We assume the function $Q:[0,\infty)\to \R$ satisfies the following conditions.
\begin{assumption}\label{asu:Q}
The function $Q$ extends to an analytic function in a complex neighborhood of $[0,\infty)$ such that
\begin{equation}\label{eq:behQorigin}
Q(x)>0 , \qquad x>0,
\end{equation}
and 
\begin{equation}\label{eq:asyQ}
    Q(x) =
\begin{cases}
\tad x^m(1+\Boh(x)),& x\to 0, \\
\Boh(x^{\epsilon}), & x\to \infty,
\end{cases}
\end{equation}
for the same integer value $ m\in \Z_{>0}$ appearing in \eqref{def:sigman}, where $\tad$ and $\epsilon$ are some positive numbers. 
\end{assumption}




The scaling $n^{2m}$ in $\sigma_n$ is motivated by the following reasoning. Under our assumptions on $V$, the local statistics of the classical Laguerre-type OP ensemble near the hard edge $z=0$ fluctuate on the scale $\Boh(n^{-2})$. Thus, in a local coordinate $\zeta\approx n^{2} x$ the statistics in the $\zeta$-plane should behave as $\Boh(1)$, that is, the variable $\zeta$ now detects the nontrivial fluctuations. Still with this same scaling, it follows from \eqref{eq:asyQ} that in the variable $\zeta$ we have $Q(x)=\Boh(\zeta^m n^{-2m})$ and therefore
$$
\frac{1}{1+\ee^{-\sad -n^{2m}Q(x)}}\approx \frac{1}{1+\ee^{-\sad -\tad \zeta^m}}.
$$
That is, with the scaling $n^{2m}$ in $\sigma_n$ and the vanishing condition \eqref{eq:asyQ} in place, the factor $\sigma_n$ in the deformed ensemble \eqref{eq:Lagdef} produces a nontrivial change of behavior in the fluctuations at the local scale $\zeta$, when compared to the classical Laguerre-type OP ensemble. 

The positivity condition on $Q$ is motivated by the following reasoning. Pointwise, the fact that $Q>0$ implies that
$$
\sigma_n(x)\to 1,\quad n\to \infty,\quad x>0.
$$
This means that $\sigma_n$ does not affect the global asymptotics of the random matrix model in intervals of the form $[\alpha,\infty)$ with $\alpha>0$, and only in a scale $x=\Boh(n^{-2})$ the statistics of the model should become affected. In principle, this global positivity condition on $(0,\infty)$ could be replaced by positivity enforced only near the hard edge $x=0$, at the cost of more complicated analysis near points outside the hard edge. As we mentioned earlier, since our goal is to showcase novel phenomena near the hard edge, we opt for keeping this condition for the sake of simplicity.

Our main results explain these heuristic interpretations rigorously, computing all the limiting quantities in terms of integrable systems. To state our results, we introduce certain quantities that allow us to decode statistics of \eqref{eq:Lagdef}.

The monic OPs $\msf P_j(x)=\msf P_j^{(n)}(x \mid \sad)=x^j+\cdots$, $j=0,1,2,\ldots$, for the weight  $\omega_n(\cdot\mid \sad)$ are uniquely determined by
$$
\int_0^\infty \msf P_j(x)x^k\omega_n(x\mid \sad)\dd x=0,\quad k=0,\hdots, j-1.
$$
The associated norming constants $\gamma_j(\sad)=\gamma_j^{(n)}(\sad)>0$, determined from
\begin{equation}\label{eq:normingcttdeff}
\frac{1}{\gamma_j(\sad)^2}=\int_0^\infty \msf P_j(x)^2\omega_n(x\mid \sad)\dd x,
\end{equation}
are such that the sequence $\{\gamma_j(\sad) \msf P_j,j=0,1,2,\ldots\}$ is orthonormal in $L^2(\omega_n(x \mid \sad)\dd x, [0,\infty))$. With these quantities in mind, we introduce the Christoffel-Darboux kernel
\begin{equation}\label{deff:CDkernel}
\msf K_n(x,y)=\msf K_n(x,y\mid \sad)\deff
\sum_{j=0}^{n-1}\gamma_j(\sad)^2\msf P_j(x)\msf P_j(y),
\end{equation}
and correlation kernel
\begin{equation}\label{deff:CorrKernel}
\wh{\msf K}_n(x,y)=\wh{\msf K}_n(x,y\mid \sad)\deff \omega_n(x\mid \sad)^{1/2}\msf K_n(x,y\mid \sad )\omega_n(y\mid \sad)^{1/2}.
\end{equation}
All the quantities we just introduced depend smoothly on the parameter $\sad$, although sometimes we do not make explicit mention to it in our notations. Also, as we will show, the limiting quantities 
$$
\gamma_j^{(n)}(\infty)=\lim_{\sad\to +\infty} \gamma_j^{(n)}(\sad)\quad \text{and}\quad \wh{\msf K}_n(\cdot \mid \infty)=\lim_{\sad\to +\infty}\wh{\msf K}_n(\cdot \mid \sad)
$$ 
make sense, and are the norming constants and correlation kernel for the undeformed weight $x^\alpha\ee^{-nV}$ corresponding to the classical Laguerre-type OP ensemble.

The relevance of these quantities is the following. The norming constants relate to the partition function via the identity (e.g., see \cite{mehta_book})
\begin{equation}\label{eq:PartNorming}
\msf Z_n(\sad)=n!\prod_{j=0}^{n-1}\gamma_{j}^{(n)}(\sad)^{-2}.
\end{equation}
The ensemble \eqref{eq:Lagdef} itself forms a determinantal point process: the joint distribution \eqref{eq:Lagdef} may be written in the determinantal form
$$
\det\left( \wh{\msf K}_n(x_i,x_j\mid \sad) \right)_{i,j=1}^n.
$$
From this determinantal structure, statistics of the ensemble \eqref{eq:Lagdef} may be computed using the theory of determinantal point processes \cite{Soshnikov2000a, BaikDeiftSuidan2016Book}. In this regard, the multiplicative statistics $\msf L_n^Q(\sad)$ from \eqref{eq:Lagdef} can be computed in terms of a Fredholm determinant involving the kernel $\msf K_n(\cdot\mid \infty)$ of the original Laguerre-type ensemble. However, such an expression is not appropriate for asymptotics, as one would still need to analyse a full Fredholm determinant. Our starting point will be instead a trace-type formula that allows to compute $\msf L_n^Q(\sad)$ directly in terms of the diagonal of the deformed kernel $\msf K_n(\cdot\mid \sad)$. With $\msf L_n^Q, \msf K_n$ and $\omega_n$ defined in \eqref{eq:MultStats}, \eqref{deff:CDkernel} and \eqref{eq:deffweight}, respectively, 
this formula is
\begin{equation}\label{eq:deffformula}
\log \msf L_n^Q(\sad)=-
\int_{\sad}^{\infty} \int_0^\infty \msf K_n(x,x\mid u) \partial_u\omega_n(x\mid u)\dd x \dd u.
\end{equation}

A variant of \eqref{eq:deffformula} appeared in \cite[Proposition~9.1]{GhosalSilva22}, where it was proven in the particular context of OPs over the full real line. Other variants also appeared in \cite[End of Section~5.2]{ClaeysGlesner2021} and \cite[Theorem~1.10]{CafassoClaeys2024} (see also \cite{CafassoClaeys2022, CafassoClaeysRuzza2021}), which obtain their analogues exploring the Jacobi differentiation identity for Fredholm determinants of so-called IIKS integrable operators. The formula as it appears here is a particular case of \cite[Proposition~4.1]{GhosalSilva25}, which provides a deformation formula valid for general OP ensembles (including discrete ones), deformations of weights beyond the case \eqref{eq:deffweight}, and which relies only on the orthogonality of the polynomials $\msf P_j$ used to construct the kernel $\msf K_n(\cdot\mid \sad)$.

Thanks to \eqref{eq:deffformula}, the first step in understanding $\msf L_n^Q(\sad)$ as $n\to\infty$ relies in obtaining asymptotics for $\msf K_n(\cdot\mid \sad)$. It turns out that in the large $n$ limit, the major contribution to the $x$-integration in \eqref{eq:deffformula} comes from a neighborhood of the origin, precisely where $\msf K_n$ has a novel scaling limit. To describe this scaling limit, let us introduce the function
\begin{equation}\label{deff:sigmaPhi}
\sigma_{\bm \Phi}(\zeta)=\sigma_{\bm \Phi}(\zeta\mid \sad)\deff \frac{1}{1+\ee^{-\sad -(-1)^m \zeta^m}}.
\end{equation}
In this definition, the index $\bm\Phi$ is introduced just for notational convenience, to distinguish between a generic function $\sigma$ which will be used later on.

\begin{theorem}\label{thm:intsysPhi}
    Fix $\alpha>-1$. For $(\zeta,\sad,\xad)\in \R\times \R \times (0,+\infty)$, the nonlocal differential equation
    \begin{equation}\label{eq:nonlocalPDEintro}
    \partial^2_\xad \Phi(\zeta\mid \sad,\xad)=\left( 
     \zeta+\frac{4\alpha^2-1}{4\xad^2}+\frac{2\ee^{-\pi\ii\alpha}}{\pi}\int_\sad^\infty \int_{-\infty}^0 \Phi(\xi\mid u,\xad)\partial_\xad \Phi(\xi\mid u,\xad)\partial_\sad \sigma_{\bm \Phi}(\xi\mid u)\dd\xi\dd u 
     \right)\Phi(\zeta\mid \sad,\xad)
    \end{equation}
    admits a solution $ \Phi=\Phi(\zeta\mid \sad,\xad)$ satisfying for each $\xad>0$ fixed,
    \begin{align*}
& \Phi(\zeta)=\left(1+\Boh(\zeta^{-1})\right)\frac{\zeta^{-1/4}}{\sqrt{2}}\ee^{\xad \zeta^{1/2}},\quad \zeta\to +\infty,
\\
& 
\Phi(\zeta)=\left(1+\Boh(\zeta^{-1})\right)\sqrt{2}\ee^{\pi\ii\alpha/2}|\zeta|^{-1/4}\cos\left( \xad|\zeta|^{1/2}-\frac{\pi}{4}-\frac{\pi\alpha}{2} \right),\quad \zeta\to -\infty.
\end{align*}
    Furthermore, as $\xad\to 0^+$,
    %
%
\begin{equation}\label{eq:asympPhixto0}
\Phi(\zeta\mid \sad,\xad) = \sqrt{\pi}\xad^{1/2} I_\alpha(\xad\zeta^{1/2} ) \left(1 + \Boh(\xad^{2+2\alpha}(1+\xad|\zeta|^{1/2})) \right),
\end{equation}
    uniformly for $\zeta\in \R\setminus \{0\}$ and $s\in\mathbb{R}$,
    where $I_\alpha$ is the modified Bessel function of the first kind with index $\alpha$, and for $\zeta<0$ we must consider the $+$-boundary value of $I_\alpha$.
\end{theorem}

With the solution $\Phi$ just discussed, we construct the kernel
\begin{multline}\label{deff:Kalpha}
\msf K_\alpha(u,v)\deff \frac{\ee^{-\pi\ii \alpha}\sqrt{\sigma_{\bm \Phi}\left(-4u/\xad^2 \right)}\sqrt{\sigma_{\bm\Phi}\left(-4v/\xad^2\right)}}{2\pi (u-v)}
\\
\times \left[ \Phi\left(-\frac{4u}{\xad^2}\right)(\partial_\xad \Phi )\left(-\frac{4v}{\xad^2}\right) -\Phi\left(-\frac{4v}{\xad^2}\right)(\partial_\xad \Phi )\left(-\frac{4u}{\xad^2}\right) \right],
\end{multline}
valid for $u\neq v$, and extended by continuity to $u=v$.
Observe that $\Phi=\Phi(\cdot\mid \sad,\xad)$, and therefore $\msf K_\alpha=\msf K_\alpha(\cdot,\cdot\mid \sad,\xad)$. 

For the next result, set
\begin{equation}\label{eq:msfcVintro}
\msf c_V\deff \pi^2\kappa_0^2,\qquad \text{and identify}\qquad \xad = \left(\frac{4\msf c_V}{\tad^{1/m}}\right)^{1/2},
\end{equation}
where $\kappa_0>0$ is as in \eqref{eq:localbehdensityeqmeasintro} and we recall that $\tad>0$ is as in \eqref{eq:asyQ}.

\begin{theorem}\label{thm:kernelintro}
Under Assumptions~\ref{deff:condQ} and \ref{asu:Q} and the identification \eqref{eq:msfcVintro}, the estimate
    $$
    \frac{1}{n^2\msf c_V }\wh{\msf K}_n\left(\frac{u}{\msf c_{V}n^2},\frac{v}{\msf c_V n^2}\right)=\msf K_\alpha(u,v\mid \sad,\xad)+\Boh(\ee^{-\sad}n^{-\kappa} + n^{-2}),\quad n\to\infty,
    $$
    holds uniformly for $u,v$ in compact subsets of $(0,+\infty)$, and uniformly for $\sad\geq \sad_0$ with $\sad_0\in \R$ fixed, and for any fixed $\kappa\in (0,2)$.
\end{theorem}

The kernel $\wh{\msf K}_n$ is the correlation kernel of the deformed matrix model \eqref{eq:Lagdef}, and in the limit $\sad\to +\infty$ it degenerates to the regular (undeformed) matrix model with weight $x^\alpha \ee^{-nV(x)}$. It was shown in \cite{Vanlessen2007} that the latter converges to the Bessel kernel
\begin{equation}\label{def:BesKer}
    \msf J_\alpha(u,v):=\frac{J_\alpha(\sqrt{u})\sqrt{v}J_\alpha'(\sqrt{v})-\sqrt{u}J_\alpha'(\sqrt{u})J_\alpha(\sqrt{v})}{2(u-v)}, \quad u,v>0,
\end{equation}
where $J_\alpha$ is the Bessel function of the first kind of order $\alpha$. The uniformity of the bound \eqref{thm:kernelintro} immediately yields the following corollary.

\begin{corollary}\label{cor:Besselkerneldeg}
    The limit
    $$
    \lim_{\sad\to +\infty}\msf K_\alpha(u,v\mid \sad)= \msf J_\alpha(u,v)
    $$
    holds pointwise for $u,v>0$, where $\msf J_\alpha$ is the Bessel kernel of order $\alpha>-1$ defined in \eqref{def:BesKer}.
\end{corollary}

When scaled around the origin, the smallest eigenvalues of the model \eqref{eq:Lagdef} converge to the Bessel point process (in the weak sense), which is the determinantal point process defined by the kernel $\msf J_\alpha$. Corollary~\ref{cor:Besselkerneldeg} is the kernel convergence reflecting this property.

Next, we turn to the asymptotics for $\msf L_n^Q$, for which we introduce
\begin{equation}\label{deff:Halpha}
\msf L_\alpha^\bes(\sad,\xad)\deff \exp\left(-\int_\sad^\infty \int_{0}^\infty \msf K_\alpha(\zeta,\zeta\mid \sad=u)(\partial_\sad \log \sigma_{\bm\Phi})\left(-\frac{4\zeta}{\xad^2}\mid \sad=u\right) \dd \zeta \dd u\right).
\end{equation}
The meaning of the upper index $\bes$ in the definition of $\msf L_\alpha^\bes$ will be explained in a moment.

\begin{theorem}\label{thm:asympLn}
Under Assumptions~\ref{deff:condQ} and \ref{asu:Q} and the identification \eqref{eq:msfcVintro}, the estimate
$$
\log \msf L_n^Q(\sad)=\log\msf L_\alpha^\bes(\sad,\xad)+\Boh\left(\frac{\ee^{-\sad}}{n}\right),\quad n\to\infty,
$$
holds uniformly for $\sad\geq \sad_0$ with any fixed $\sad_0\in \R$.
\end{theorem}

 As mentioned earlier, the kernel $\wh{\msf K}_n$ is the correlation kernel for the conditional thinning of the eigenvalues of \eqref{eq:Lagdef}. In the recent work \cite{ClaeysSilva25}, Claeys and the second-named author give conditions on the symbol $\sigma_n$ in \eqref{eq:MultStats} under which (1) the conditional thinning of a finite $n$ random particle system converges to the conditional thinning of the large $n$ limit of the particle system, and (2) multiplicative statistics of the finite $n$ random particle system converge to the multiplicative statistics of their large $n$ limit. 

Modulo checking the technical conditions required to apply the results from \cite{ClaeysSilva25}, regarding (1),
Theorem~\ref{thm:kernelintro} establishes that the kernel for the conditional thinning process of the Bessel point process associated to the symbol $\sigma_{\bm\Phi}(-4\zeta\xad^{-2}\mid \sad,\xad)$ is precisely the kernel $\msf K_\alpha$. Regarding (2), Theorem~\ref{thm:asympLn} is saying that $\msf L_\alpha^\bes(\sad)$ coincides with the multiplicative statistics of the Bessel point process for the same symbol, which also explains why we used the upper index $\bes$.

This way, we refer to $\msf K_\alpha$ as the conditional thinning kernel of the Bessel point process, and $\msf L_\alpha^\bes$ as a multiplicative statistic of the Bessel point process.
Checking the requirements for the application of the results in \cite{ClaeysSilva25} is a technical feat that we will not delve into. 

With $\Phi$ as in Theorem~\ref{thm:intsysPhi}, introduce
\begin{equation}\label{eq:deffpintro}
\msf p(\sad,\xad)\deff-\frac{4\alpha^2-1}{8\xad}+\frac{\ee^{-\pi\ii\alpha}}{2\pi}\int_{\sad}^\infty \int_{-\infty}^0 
\Phi(\xi\mid \sad=u,\xad)^2\partial_\sad \sigma_{\bm\Phi}(\xi\mid \sad=u)\dd\xi \dd u.
\end{equation}
We may view \eqref{eq:nonlocalPDEintro} as a linear Schrödinger equation with potential $2\partial_\xad\msf p$, namely
$$
\partial_\xad^2\Phi(\zeta\mid \sad,\xad)=\left( \zeta+2\partial_\xad\msf p(\sad,\xad) \right)\Phi(\zeta\mid \sad,\xad).
$$

It turns out that we can relate the multiplicative statistic $\msf L_\alpha^\bes$ directly to the potential $\msf p$.

\begin{theorem}\label{thm:TWtyperepres}
Fix $\alpha>-1$. The estimate
\begin{equation}\label{eq:asymppTWthm}
\msf p(\sad,\xad)=-\frac{4\alpha^2-1}{8\xad}+\frac{1}{2^{2\alpha+1}\Gamma(\alpha+1)^2} \xad^{2\alpha+1}\int_0^\infty \frac{u^\alpha \ee^{-\sad-u^m}}{1+\ee^{-\sad-u^m}}\dd u+\Boh(\xad^{2\alpha+3}), \quad \xad \to 0^+,
\end{equation}
holds. In addition, the identity
    \begin{equation}\label{eq:TWrepr}
    \partial_{\xad}\left( \log\msf L_\alpha^\bes(\sad,\xad)\right)=-\msf p(\sad,\xad)-\frac{4\alpha^2-1}{8\xad}
    \end{equation}
    is also true.
\end{theorem}

Using \eqref{eq:asympPhixto0}, it is straightforward to show that $\log \msf L_\alpha^\bes$ goes to $0$ as $\xad \to 0^+$.  In particular, Theorem~\ref{thm:TWtyperepres} then provides the representation
$$
\log\msf L_\alpha^\bes(\sad,\xad)=-\int_0^\xad\left(\msf p(\sad,u)+\frac{4\alpha^2-1}{8u}\right)\dd u.
$$

The celebrated $\beta=2$ Tracy-Widom distribution had its name coined thanks to the work of Tracy and Widom \cite{tracy_widom_level_spacing}, who showed that the second log derivative of a multiplicative statistic of random matrices (a gap probability) at the soft edge could be expressed in terms of the Hastings-McLeod solution to the Painlevé~II equation. The analogue representation for gap probabilities of the Bessel kernel in terms of the Painlevé~V equation was obtained shortly afterwards by Tracy and Widom as well \cite{TracyWidom1994Bessel}. Several different multivariate extensions have been obtained in the past 20 years \cite{TracyWidom2003AiryProcess, adler_vanmoerbeke_2005, Wang2009, BertolaCafasso2012, AmirCorwinQuastel2011, BaikProkhorovSilva2023}.

Theorem~\ref{thm:TWtyperepres} is the analogue representation, showing that the second log derivative of $\msf L_\alpha^\bes$ may be expressed in terms of $\msf p$ which, in turn, is given in terms of the solution to the nonlocal nonlinear equation \eqref{eq:nonlocalPDEintro}. In the case $m=1$, it is possible to obtain a PDE satisfied by $\msf p$ directly.

\begin{theorem}\label{thm:m1PDE}
    Suppose that $m=1$ in \eqref{def:sigman}. Then the function $\msf p$ from \eqref{eq:deffpintro} satisfies the PDE
    $$
   \partial^4_{\sad \xad \xad \xad}\msf p-8\partial_{\xad}\msf p\left(\partial^2_{\sad\xad}\msf p+\frac{1}{2}\right)-4\partial^2_{\xad\xad}\msf p\left(\partial_{\sad}\msf p+\frac{\xad}{2}\right)=0. 
   $$
\end{theorem}

As said earlier, some of our results are in direct comparison with the recent work \cite{Ruzza2024} by Ruzza. Therein, and as mentioned earlier, Ruzza considers a class of multiplicative statistics for the Bessel point process. His multiplicative statistics are coming from linear scalings of fixed functions which are not necessarily continuous, and our family of statistics $\msf L_\alpha^\bes$ reduces to a particular instance of his setup\footnote{With the identification of variables $(\xad,\sad)=(x,-t)$, our function $\msf p(\xad,\sad)$ is comparable with the function $v(x,t)$ in \cite{Ruzza2024} through the relation%
$$
\msf p(\xad,\sad)=-\ii\left(v(x,t)+\frac{xt}{2}\right).
$$
} only when $m=1$. This way, Theorem~\ref{thm:m1PDE} is a rediscovery of \cite[Theorem~1.1--(ii)]{Ruzza2024}, the case $m=1$ of \eqref{eq:nonlocalPDEintro} with boundary condition \eqref{eq:asympPhixto0} is equivalent to \cite[Theorem~1.1--(iii)]{Ruzza2024}, and the case $m=1$ of \eqref{eq:asymppTWthm} is the equivalent of \cite[Theorem~1.6]{Ruzza2024}.

\subsection*{Organization of the paper}\hfill 

The remainder of this paper is organized as follows.

Section \ref{sec:modelproblemfull} is devoted to the construction and analysis of a model RHP, which plays a central role in the asymptotic analysis of OPs carried out later. This model problem governs the hard edge statistics. We establish the solvability of this model problem and derive the nonlocal differential equations satisfied by its solution.

In Section \ref{sec:modelproblemasymptotics}, we perform a detailed asymptotic analysis of the model RHP with respect to a large parameter (which corresponds to the system size $n$), and analyze its limiting behavior in relevant regimes.

Section \ref{sec:conseqmodelproblem} wraps up the asymptotic consequences from Section \ref{sec:modelproblemasymptotics}, in particular establishing certain asymptotic properties of the correlation kernel, and related integrals that will later appear in the analysis of multiplicative statistics. We also prove main theorems on integrable equations in this section. 

Finally, Sections \ref{sec:RHPOPsanalysis} and \ref{sec:conclusionmainresults} contain the asymptotic analysis of the orthogonal polynomials via the Deift-Zhou nonlinear steepest descent method, where the model problem and its properties established in the previous sections play a fundamental role.

\subsection*{About the notation}\label{sec:notation}  
\hfill

We use $D_r(z_0)$ to denote the disk on the complex plane centered at $z_0$ with radius $r>0$, and $D_r=D_r(0)$ for the particular case when $z_0=0$ is the origin. In general, we use bold capital letters $\bm Y,\bm \Psi$ etc. to denote matrix-valued functions. The letters $\varepsilon,\delta,\eta$ always denote positive constants that can be made arbitrarily small but are kept fixed, and we always emphasize when they may depend on external parameters. These small constants may have different values for different occurrences in the text. The set of strictly positive integers is $\Z_{>0}$, and we also use the notation $\Z_{\geq 0}\deff \Z_{>0}\cup \{0\}$.

When we write that $x\to \infty$ for some variable $x$, we mean that $x\to +\infty$ when $x$ is real, or $x\to \infty$ along any direction of the complex plane in case $x$ is allowed to assume values in $\C\setminus \R$ as well. These two distinct cases will always be clear from the context and meaning of the variables involved.

We also use the following matrix notation. We denote by $\bm I$ and $\bm 0$ the identity matrix and the null matrix, respectively, and by $\bm E_{ij}$ the $2\times 2$ matrix with $1$ in the $(i,j)$-entry and $0$ in the remaining entries. For convenience, we set 
\begin{equation}\label{def:U0}
 \bm U_0\deff \frac{1}{\sqrt{2}}
\begin{pmatrix}
    1 & \ii \\ \ii & 1
\end{pmatrix}
=
\frac{1}{\sqrt{2}}
\left(\bm I+\ii \bm E_{12}+\ii \bm E_{21}\right)
,\quad 
\sp_3\deff 
\begin{pmatrix}
    1 & 0 \\ 0 & -1
\end{pmatrix}=
\bm E_{11}-\bm E_{22}.   
\end{equation}
In the course of the Riemann-Hilbert analysis, we will use matrix norm notation. For a matrix-valued function $M:U\subset \C\to \C^{2\times 2}$, we denote
$$
|M(z)|\deff \max_{i,j=1,2} |M_{ij}(z)|,
$$
where  $M_{ij}$ stands for the $(i,j)$-th entry of $M$. 
For $p\in [1,\infty]$ and a curve $\Gamma\subset U$, we also use the corresponding $L^p(\Gamma)=L^p(\Gamma,|\dd s|)$ norm with respect to the arc length measure $|\dd s|$,
$$
\|M\|_{L^p(\Gamma)}\deff \max_{i,j=1,2}\|M_{ij}\|_{L^p(\Gamma)}.
$$

For a matrix-valued function $M$ depending on a variable $x\in \R$, it is convenient to introduce the notation
\begin{equation}\label{deff:Deltax}
\bm\Delta_x M(x)\deff M(x)^{-1}M'(x).
\end{equation}
Under a change of variables $x\mapsto \zeta=\zeta(x)$, this operator transforms as
\begin{equation}\label{deff:Deltaxzeta}
\bm\Delta_x(M(\zeta(x))=\zeta'(x) M(\zeta(x))^{-1}\frac{\dd M}{\dd \zeta}(\zeta(x))=\zeta ' \bm\Delta_\zeta M(\zeta).
\end{equation}

\subsection*{Acknowledgments}\hfill 

G.S. thanks Mattia Cafasso, Tom Claeys, Thomas Chouteau, Alfredo Deaño, and Giulio Ruzza, for various discussions related to this work. Parts of this project were carried out during academic visits of G.S. to Fudan University. He acknowledges the hospitality during these visits.

This project was finalized while the three authors were taking part in the workshop {\it Integrable Systems and Random Matrix Theory}, at Great Bay University, Dongguan, China, in January 2026. They appreciate the hospitality and encouraging environment during the workshop.

L.M. is supported by the UC3M grant 2024/00002/007/001/023 ``Local and global limits of complex-dimensional DPPs" and the grant ID2024-155133NB-I00,
``Orthogonality, Approximation, and Integrability:
Applications in Classical and Quantum Stochastic Processes (ORTH-CQ)" by the Agencia Estatal de Investigación.

G.S. acknowledges support by the São Paulo Research Foundation (FAPESP),
Brazil, Process Number \# 2019/16062-1, and by the Brazilian National Council for Scientific and Technological Development (CNPq) under Grant \# 306183/2023-4. 

L.Z. acknowledges support by National Natural
Science Foundation of China under Grant \# 12271105 and “Shuguang Program”
supported by Shanghai Education Development Foundation and Shanghai Municipal Education Commission.

\section{The model RHP}\label{sec:modelproblemfull}

As mentioned earlier, the main technical tool we utilize is the RHP formulation of OPs and associated quantities. When we perform the asymptotic analysis of the OPs via the RHP method, at one of the core steps we need to construct an approximation to them near the origin, the so-called model RHP for the local parametrix. For usual Laguerre-type OPs, such approximation is constructed using Bessel functions as originally developed in \cite{KuijlaarsMcLaughlinVanAsscheVanlessen2004}; see  \cite{Vanlessen2007} for details. Due to the nontrivial scaling of $\sigma_n$ near the origin, we need to use a different model problem, which we introduce and study in this section.

This section is structured as follows. In Section~\ref{sec:modelproblem} we introduce the model problem itself. In Section~\ref{sec:solvability} we establish the solvability of the model RHP, by means of vanishing-lemma-like arguments. In Section~\ref{sec:modelIntSys}, we connect the model problem with integrable systems; the content of the latter section ultimately leads to the differential equations from our main results.

\subsection{The model problem}\label{sec:modelproblem}\hfill 

Fix $m\in \Z_{>0}$, which will have the same meaning as in \eqref{eq:asyQ}, and set
\begin{equation}\label{deff:thetambasiccontours}
\theta_m\deff \pi\left(1-\frac{1}{3m}\right),\quad \Gamma_\pm \deff (\ee^{\pm \ii \theta_m}\infty, 0], \quad \Gamma_0\deff (-\infty,0],\quad \Gamma\deff \Gamma_+\cup\Gamma_0\cup\Gamma_-,
\end{equation}
where the orientations of the arcs $\Gamma_\pm,\Gamma_0$ are taken from $\infty$ to the origin. Also, introduce the sectors
\begin{equation}\label{def:mcalSpm}
\begin{aligned}
& \mcal S_0\deff \{\zeta\in \C\mid -\theta_m<\arg\zeta<\theta_m\}, \\
& \mcal S_+ \deff \{\zeta\in \C\mid \theta_m<\arg\zeta<\pi\}, \\
& \mcal S_- \deff \{\zeta\in \C\mid -\pi<\arg\zeta<-\theta_m\},
\end{aligned}
\end{equation}
with principal value of the argument, that is, $\arg\zeta\in (-\pi,\pi)$. We refer to Figure \ref{fig:Gammapm} for an illustration of these arcs and sectors.

\begin{figure}[t]
\begin{center}
   \setlength{\unitlength}{1truemm}
   \begin{picture}(100,70)(-5,2)       
       \put(40,40){\line(-2,-3){16}}
       \put(40,40){\line(-2,3){16}}
       \put(40,40){\line(-1,0){30}}

       \put(30,55){\thicklines\vector(2,-3){1}}
       \put(30,40){\thicklines\vector(1,0){1}}
       \put(30,25){\thicklines\vector(2,3){1}}

       \put(40,36.3){$0$}
       \put(20,11){$\Gamma_-$}
       \put(20,67){$\Gamma_+$}
       \put(3,40){$\Gamma_0$}
       \put(20,49){$\mcal S_+$}
       \put(20,29){$\mcal S_-$}
       \put(55,40){$\mcal S_0$}

       \put(40,40){\thicklines\circle*{1}}
   \end{picture}
   \caption{The contours $\Gamma_{\pm}$, $\Gamma_0$, and the sectors $\mcal S_{\pm}$, $\mcal S_{0}$.}
   \label{fig:Gammapm}
\end{center}
\end{figure}
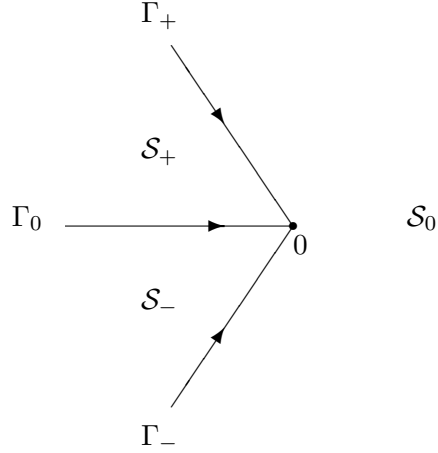

Observe that
$$
\partial\mcal S_0=\Gamma_+\cup\Gamma_-,\quad \partial\mcal S_+=\Gamma_+\cup\Gamma_0,\quad \partial\mcal S_-=\Gamma_-\cup\Gamma_0.
$$
The value $\theta_m$ is chosen so that
\begin{equation}\label{eq:ineqzetam}
\frac{1}{2}|\zeta|^m\leq (-1)^m\re(\zeta^m)\leq |\zeta|^m \quad \text{for}\quad \zeta\in \mcal S_+\cup \mcal S_-,
\end{equation}
with the upper bound being attained only along $\partial\mcal S_+\cap\partial\mcal S_-=\Gamma_0$, and the lower bound being attained only along $(\partial\mcal S_+\cup \partial\mcal S_-)\setminus \Gamma_0=\Gamma_+\cup\Gamma_-$. 

Fix $\alpha>-1$ (which has the same meaning as in \eqref{eq:deffweight}) and a function $\msf h:\Gamma\to \C$. We are interested in the following RHP.
\begin{rhp}\label{RHP:model}
Find a $2\times 2$ matrix-valued function $\bm\Psi:\C \setminus \Gamma\to \C^{2\times 2}$ with the following properties.
\begin{enumerate}[(i)]
\item $\bm \Psi$ is analytic on $\C\setminus \Gamma$.
\item The entries of $\bm \Psi_\pm$ are continuous except possibly at the origin, and for $\zeta\in \Gamma\setminus \{0\}$ they satisfy the jump relation 
$$
\bm \Psi_+(\zeta)=\bm\Psi_-(\zeta)\bm J_{\bm \Psi}(\zeta),
$$
with jump matrix 
\begin{equation}\label{eq:jumpPsi}
\bm J_{\bm \Psi}(\zeta)\deff
\begin{dcases}
\bm I+(1+\ee^{-\msf h(\zeta)})\ee^{\pm \pi \ii \alpha}\bm E_{21}, & \zeta\in \Gamma_\pm, \\
\frac{1}{1+\ee^{-\msf h(\zeta)}}\bm E_{12}-(1+\ee^{-\msf h(\zeta)})\bm E_{21}, & \zeta\in \Gamma_0.
\end{dcases}
\end{equation}
\item For some matrix-valued function $\bm \Psi_{1}$ which is independent of $\zeta$, the following expansion is valid,
\begin{equation}\label{eq:asymptRHPpsimodel}
\bm \Psi(\zeta)=\left(\bm I+\frac{\bm \Psi_{1}}{\zeta}+\Boh(\zeta^{-2})\right)\zeta^{-\sp_3/4}\bm U_0\ee^{2\zeta^{1/2}\sp_3},\quad \zeta\to \infty,
\end{equation}
  with the principal branches of the fractional powers, and where $\bm U_0$ and $\sp_3$ are defined in \eqref{def:U0}. 
\item The entries of $\bm \Psi_\pm$ belong to $L^2_{\mathrm{loc}}(\Gamma,|\dd s|)$. 
\end{enumerate}
\end{rhp}

Condition (ii) asks for continuity of the boundary values $\bm \Psi_\pm$ except possibly at the origin, so Condition (iv) is an additional requirement only about the behavior of $\bm \Psi$ near the origin.

In the statement of RHP~\ref{RHP:model} above, we have used the function $\msf h$, but in fact the jump may be written simply in terms of the function $\sigma$ given by
\begin{equation}\label{eq:sigmafactor}
\sigma:\Gamma\to \C,\quad \sigma(\zeta)\deff \frac{1}{1+\ee^{-\msf h(\zeta)}},
\end{equation}
namely as
$$
\bm J_{\bm\Psi}(\zeta)=
\begin{dcases}
    \bm I+\frac{\ee^{\pm\pi\ii\alpha}}{\sigma(\zeta)}\bm E_{21},& \zeta\in \Gamma_\pm, \\ 
    \sigma(\zeta)\bm E_{12}-\frac{1}{\sigma(\zeta)}\bm E_{21}, & \zeta\in \Gamma_0.
\end{dcases}
$$
Historically, factors of the form $(1+\ee^{-\msf h})^{-1}$ have appeared in the so-called finite temperature deformations of random matrix models, and partly for that reason oftentimes it is convenient to state conditions/results in terms of $\msf h$. However, in what follows we may use either $\sigma$ or $\msf h$, depending on the convenience of the situation.

Obviously, $\bm\Psi$ depends on $\msf h,\alpha$, and when we need to make this dependence explicit we write $\bm \Psi=\bm \Psi(\cdot\mid \msf h),\bm\Psi(\cdot\mid \alpha),\bm \Psi(\cdot\mid \msf h,\alpha)$, etc. For the RHP above to be well-posed, we should impose conditions on $\msf h$. Global conditions on $\msf h$ will be imposed in a moment, but it is convenient to write the condition (iv) in a more explicit manner when $\msf h$ is analytic at the origin, as given by the next result.  To proceed, we introduce
\begin{equation}\label{eq:hatUpsilonasyfactor}
\msf a(\zeta)=\msf a_\alpha(\zeta)\deff 
\begin{dcases}
\frac{1}{2\ii \sin(\alpha \pi)}, & \alpha \notin \mathbb{Z}, \alpha>-1, \\
\frac{\ee^{ \pi \ii \alpha}}{2 \pi  \ii}\log \zeta, & \alpha\in \Z_{\geq 0},
\end{dcases}
\end{equation}
for $\zeta\in \C \setminus (-\infty,0]$. It is easily seen that $\msf a$ satisfies the jump condition
\begin{equation}\label{eq:jumpafactor}
\msf a_+(\zeta)\ee^{\pi\ii\alpha}-\msf a_-(\zeta)\ee^{-\pi \ii\alpha}=1,\qquad \zeta<0,
\end{equation}
for any $\alpha>-1$. This function will be used several times throughout the text.

\begin{prop}\label{Prop:behmodelproblemorigin}
Suppose that $\msf h$ is analytic in a neighborhood of the origin and $1+\ee^{-\msf h}$ never vanishes near the origin. Then Condition (iv) in RHP \ref{RHP:model} implies that for some $\varepsilon>0$, the identities
\begin{equation}\label{eq:behPsiorigin}
\begin{aligned}[b]
\bm \Psi(\zeta)
& =\bm \Psi_0(\zeta)\zeta^{\alpha\sp_3/2}(\bm I+\msf a(\zeta)\bm E_{12})\sigma(\zeta)^{-\sp_3/2}
(\bm I-\sigma(\zeta)^{-1}(\chi_{\mcal S_+}(\zeta)\ee^{\pi \ii \alpha}-\chi_{\mcal S_-}(\zeta)\ee^{-\pi \ii \alpha})\bm E_{21}) \\
& =\bm \Psi_0(\zeta)\zeta^{\alpha\sp_3/2}(\bm I+\msf a(\zeta)\bm E_{12})(\bm I-(\chi_{\mcal S_+}(\zeta)\ee^{\pi \ii \alpha}-\chi_{\mcal S_-}(\zeta)\ee^{-\pi \ii \alpha})\bm E_{21}) \sigma(\zeta)^{-\sp_3/2}
\end{aligned}
\end{equation}
hold for $\zeta\in D_\varepsilon$, where $\bm \Psi_0(\zeta)$ is analytic on $D_\varepsilon$, all the roots are principal branches, and $\chi_{\mcal S_\pm}$ is the characteristic function of the set $\mcal S_\pm$.
\end{prop}
\begin{proof}
The second equality in \eqref{eq:behPsiorigin} follows from a simple algebraic manipulation, which we now use without further justification.

We use the first identity in \eqref{eq:behPsiorigin} as a defining relation for $\bm \Psi_0$ and prove that it is an analytic function near the origin. Note that the factors $\zeta^{\alpha/2}, \sigma^{-1/2}$ and $\msf a$ are analytic and non-vanishing on a set of the form $D_\varepsilon\setminus (-\varepsilon,0]$, from which it follows that $\bm \Psi_0$ is piecewise analytic, with jumps along $\Gamma$. A direct calculation, also using the analyticity of $\zeta^{\alpha/2}, \sigma^{-1/2},\msf a$ away from the negative axis, shows that $\bm \Psi_0$ is analytic across $\Gamma_+$ and $\Gamma_-$ as well. For $\zeta\in \Gamma_0\setminus \{0\}$, we compute
$$
\bm \Psi_{0,-}(\zeta)^{-1}\bm \Psi_{0,+}(\zeta)=\bm I+|\zeta|^{\alpha}\left(1-\msf a_+(\xi)\ee^{\pi \ii \alpha}+\msf a_-(\zeta)\ee^{-\pi \ii\alpha}\right)\bm E_{12}.
$$
In virtue of \eqref{eq:jumpafactor}, the right-hand side above is equal to $\bm I$, and we conclude that $\bm \Psi_0$ is analytic across $\Gamma_0$ as well. 

The arguments above show that $\zeta=0$ is an isolated singularity of $\bm\Psi_0$. Because we are assuming that the singularities of the entries of $\bm \Psi$ at the origin are square-integrable and $\alpha>-1$, we see that the entries of $\bm \Psi_{0,\pm}$ belong to $L^1$, implying that $\bm \Psi_0$ has a removable singularity at $\zeta=0$.
\end{proof}

\begin{remark}
The previous proposition requires that $\sigma^{-1}=1+\ee^{-\msf h}$ is analytic and non-vanishing near the origin, but not necessarily $\msf h$ itself. It is thus immediate to see that its conclusion also holds for the choice $\msf h=+\infty$, which leads to $1+\ee^{-\msf h}\equiv 1$; for this choice the RHP for $\bm \Psi$ reduces to a RHP for Bessel functions (see Appendix \ref{sec:BesselParametrix} below), which is known and will be used later.
\end{remark}

For the remainder of this section we are interested in two independent aspects. First of all, under some additional conditions on $\msf h$ we will prove that the model problem is solvable. Second, for certain particular choices of $\msf h$, we explain how our model problem is connected to a family of integrable systems.

\subsection{Solvability of the model problem}\label{sec:solvability}\hfill 

Our goal in this section is to prove that under some conditions on $\msf h$, the RHP~\ref{RHP:model} for $\bm \Psi$ always admits a solution. Recall the contours $\Gamma_\pm,\Gamma_0$ and the sectors $\mcal S_0,\mcal S_\pm$ defined in \eqref{deff:thetambasiccontours} and \eqref{def:mcalSpm}, respectively, we 
assume the following.
\begin{assumption}\label{assumpt:hsolutionrhp}
The function $\msf h$ in \eqref{eq:jumpPsi} satisfies the following conditions.
\begin{enumerate}[(i)]
    \item $\msf h:G\to \C$ is well-defined and analytic in some open set $G$ that contains the closure of the union of sectors ${\mcal S}_+\cup {\mcal S}_-$.
    \item $\msf h$ is real-valued along $(-\infty,0)$.
    \item The associated function $\sigma$ from \eqref{eq:sigmafactor} extends to $G$ with a bounded control. More precisely,
    $$
    \sigma:G\to \C,\quad \sigma(\zeta)\deff \frac{1}{1+\ee^{-\msf h(\zeta)}},
    $$
    remains bounded as $\zeta \to \infty$ along $G$, and furthermore
    $$
    \sigma(\zeta)=1+\Boh(\zeta^{-2}),
    $$
    as $\zeta\to \infty$  in a neighborhood of  $(-\infty,0)$.
\end{enumerate}
\end{assumption}

The final goal of this section is to prove the following result.

\begin{theorem}\label{thm:solvability}
    Under Assumption~\ref{assumpt:hsolutionrhp},  RHP \ref{RHP:model} for $\bm \Psi$ admits a unique solution.
\end{theorem}

The uniqueness claim in Theorem~\ref{thm:solvability} follows in a simple standard way in RHP theory \cite{deift_book}. We will establish the existence of the solution $\bm \Psi$ by means of the technique of a vanishing lemma. Following the general scheme laid out in \cite{DKMVZ1,FZ92,Zhou89}, a key step in the argument is to show that the “homogeneous” version of the RHP has only the trivial solution. Over here, the relevant homogeneous version of the RHP~\ref{RHP:model} is the following.

\begin{rhp}\label{rhp:homogeneous}
Find a $2\times 2$ matrix-valued function $\bm\Psi^{\rm H}:\C\setminus \Gamma\to \C^{2\times 2}$ with the following properties.
\begin{enumerate}[(i)]
\item $\bm \Psi^{\rm H}$ is analytic on $\C\setminus \Gamma$.
\item The entries of $\bm \Psi_\pm^{\rm H}$ are continuous except possibly at $\zeta=0$, and for $\zeta\in \Gamma\setminus \{0\}$ they satisfy the jump relation 
$$
\bm \Psi^{\rm H}_+(\zeta)=\bm\Psi^{\rm H}_-(\zeta)\bm J_{\bm \Psi}(\zeta),
$$
with the same jump matrix $\bm J_{\bm \Psi}$ as in \eqref{eq:jumpPsi}.

\item As $\zeta\to \infty$,
\begin{equation}
\bm \Psi^{\rm H}(\zeta)=\Boh(\zeta^{-1})\zeta^{-\sp_3/4}\bm U_0\ee^{2\zeta^{1/2}\sp_3}.
\end{equation}
\item As $\zeta\to 0$, 
it behaves as
\begin{equation}\label{eq:PsiHzero}
 \bm \Psi^{\rm H}(\zeta)=\bm \Psi_0^{\rm H}(\zeta)\zeta^{\alpha\sp_3/2}(\bm I+\msf a(\zeta)\bm E_{12})(\bm I-( \chi_{{\mcal S}_+}(\zeta)\ee^{\pi\ii\alpha}- \chi_{{\mcal S}_-}(\zeta)\ee^{-\pi\ii\alpha})\bm E_{21})\sigma(\zeta)^{-\sp_3/2},   
\end{equation}
where $\bm \Psi_0^{\rm H}$ is analytic in a neighborhood of the origin, $\msf a$ and $\sigma$ are defined in \eqref{eq:hatUpsilonasyfactor} and \eqref{eq:sigmafactor}, respectively. 
\end{enumerate}
\end{rhp}

As mentioned, from standard RHP theory it follows that Theorem~\ref{thm:solvability} is proved once we establish the next lemma, which is usually referred to as a vanishing lemma.

\begin{lemma}\label{lem:vanishing}
    Under Assumption~\ref{assumpt:hsolutionrhp}, the unique solution to RHP \ref{rhp:homogeneous} is the trivial solution $\bm \Psi^{\rm H}\equiv 0$.
\end{lemma}

This way, for the remainder of this section our work will be focused towards proving Lemma~\ref{lem:vanishing}, which will be established following a standard route of making certain transformations of the homogeneous RHP, and then applying the theory of analytic functions to prove that the entries of a transformed version of $\bm \Psi^{\rm H}$ are all identically zero.

As a first step, we explore that the function $\msf h$ is analytic on a neighborhood of ${\mcal S}_+\cup {\mcal S}_-$, and transform
$$
\bm \Phi^{\rm H}(\zeta)\deff
\bm\Psi^{\rm H}(\zeta)\left( \bm I+\sigma(\zeta)^{-1} \left( \ee^{\pi \ii\alpha}\chi_{{\mcal S}_+}(\zeta)-\ee^{-\pi \ii \alpha}\chi_{{\mcal S}_-}(\zeta)\right)\bm E_{21}\right)\ee^{-2\zeta^{1/2}\sp_3}.
$$
This transformation has the effect of removing the jumps outside the real axis, and also to remove the exponential behavior from the asymptotics at $\infty$. More precisely, a direct calculation shows that the function $\bm \Phi^{\rm H}$ solves the following RHP.
\begin{rhp}Find a $2\times 2$ matrix-valued function $\bm\Phi^{\rm H}:\C\setminus (-\infty,0]\to \C^{2\times 2}$ with the following properties.
\begin{enumerate}[(i)] 
\item $\bm \Phi^{\rm H}$ is analytic on $\C\setminus (-\infty,0]$.
\item The entries of $\bm \Phi_\pm^{\rm H}$ are continuous except possibly at the origin, and for $\zeta\in (-\infty,0)$ they satisfy the jump relation 
%
\begin{equation}\label{eq:PhiHjump}
 \bm \Phi^{\rm H}_+(\zeta)=\bm\Phi^{\rm H}_-(\zeta) \left( \ee^{(\pi \ii \alpha-4\ii|\zeta|^{1/2})\sp_3}+\sigma(\zeta)\bm E_{12}\right ).    
\end{equation}

\item As $\zeta\to \infty$,
\begin{equation}
\bm \Phi^{\rm H}(\zeta)=\Boh(\zeta^{-1})\zeta^{-\sp_3/4}\bm U_0\Boh(1)=\Boh(\zeta^{-3/4}).
\end{equation}
\item As $\zeta\to 0$, it behaves as
$$
\bm \Phi^{\rm H}(\zeta)=\bm \Psi_0^{\rm H}(\zeta)\zeta^{\alpha\sp_3/2}(\bm I+\msf a(\zeta)\bm E_{12})\Boh(1),
$$
where $\bm \Psi_0^{\rm H}$ is analytic in a neighborhood of the origin as given in \eqref{eq:PsiHzero}.
\end{enumerate}
\end{rhp}

Next, set
$$
\bm X^{\rm H}(\zeta)\deff 
\bm \Phi^{\rm H}(\zeta)  
\begin{cases}
\begin{pmatrix}
    0 & -1 \\ 1 & 0
\end{pmatrix}, & \im \zeta>0, \\
\bm I, & \im \zeta <0.
\end{cases}
$$

With this transformation, we move the diagonal entries in \eqref{eq:PhiHjump} to off-diagonal entries, and the matrix $\bm X^{\rm H}$ is a solution to the following RHP.

\begin{rhp}Find a $2\times 2$ matrix-valued function $\bm X^{\rm H}:\C\setminus \R\to \C^{2\times 2}$ with the following properties.
\begin{enumerate}[(i)]
\item $\bm X^{\rm H}$ is analytic on $\C\setminus \R$.
\item The entries of $\bm X_\pm^{\rm H}$ are continuous except possibly at the origin, and for $\zeta \in \R$ they satisfy the jump relation 
\begin{equation}\label{eq:XHjump}
\bm X^{\rm H}_+(\zeta)=\bm X^{\rm H}_-(\zeta)\bm J_{\bm X^{\rm H}}(\zeta),
\end{equation}
with jump matrix 
\begin{equation}\label{def:JXH}
    \bm J_{\bm X^{\rm H}}(\zeta)=
\begin{cases}
    \begin{pmatrix}
        0 & -1 \\ 1 & 0 
    \end{pmatrix}, & \zeta>0, \\
    \begin{pmatrix}
        \sigma(\zeta) & -\ee^{\pi \ii\alpha - 4\ii |\zeta|^{1/2}} \\ \ee^{-\pi \ii\alpha +4\ii |\zeta|^{1/2}} & 0
    \end{pmatrix}, & \zeta<0.
\end{cases}
\end{equation}

\item As $\zeta\to \infty$,
\begin{equation}
\bm X^{\rm H}(\zeta)=\Boh(\zeta^{-3/4}). 
\end{equation}
\item As $\zeta\to 0$, it behaves as
$$
\bm X^{\rm H}(\zeta)=\bm \Psi_0^{\rm H}(\zeta)\zeta^{\alpha\sp_3/2}(\bm I+\msf a(\zeta)\bm E_{12})
\Boh(1),
$$
where $\bm \Psi_0^{\rm H}$ is analytic in a neighborhood of the origin as given in \eqref{eq:PsiHzero}.
\end{enumerate}
\end{rhp}

We are now ready to prove Lemma~\ref{lem:vanishing} and thus conclude Theorem~\ref{thm:solvability}.

\begin{proof}[Proof of Lemma~\ref{lem:vanishing}]

Denote by $M^*$ the adjoint (complex transpose) of a matrix $M$. The jump matrix $\bm J_{\bm X^{\rm H}}$ enjoys the symmetry properties
\begin{equation}\label{eq:JXH}
\bm J_{\bm X^{\rm H}}(\zeta)+\bm J_{\bm X^{\rm H}}(\zeta)^*=\begin{cases}
 \bm 0, & \zeta>0,
 \\
2\sigma(\zeta)\bm E_{11}, & \zeta<0.
\end{cases}
\end{equation}

To explore these symmetries, let us set
$$
\bm H(\zeta)\deff \bm X^{\rm H}(\zeta)\bm X^{\rm H}(\overline{\zeta})^*,\quad \zeta\in \C\setminus \R.
$$
Then $\bm H$ is analytic on $\C\setminus \R$ and has the behavior
$$
\bm H(\zeta)=\Boh(\zeta^{-3/2}),\quad \zeta\to \infty.
$$
Using \eqref{eq:hatUpsilonasyfactor} and the definitions of $\bm H$ and $\bm X^{\rm H}$, it also follows that 
$$
\bm H(\zeta)=\Boh(\zeta^\alpha (1+\msf c_\alpha \log \zeta) ),\quad \zeta\to 0,
$$
for some constant  $\msf c_\alpha$ depending on $\alpha$ with $\msf c_\alpha=0$ if $\alpha\notin \Z$. In particular, this transformation regularizes the behavior at the origin and at $\infty$, to an integrable behavior. Using Cauchy's Theorem, we then obtain
$$
\bm 0=\int_\R \bm H_+(\zeta)\, \dd\zeta=\int_\R \bm X_-^{\rm H}(\zeta)\bm J_{\bm X^{\rm H}}(\zeta)(\bm X^{\rm H}(\overline{\zeta})^*)_+ \, \dd \zeta
=\int_\R \bm X_-^{\rm H}(\zeta)\bm J_{\bm X^{\rm H}}(\zeta)\bm X_-^{\rm H}(\zeta)^* \, \dd \zeta. 
$$
Adding the integral on the right-hand side with its adjoint and then using \eqref{eq:JXH} we obtain the identity
$$
\int_{-\infty}^0 \bm X_-^{\rm H}(\zeta)\bm E_{11}\bm X_-^{\rm H}(\zeta)^* \sigma(\zeta) \, \dd \zeta= \bm 0.
$$
Writing entry-wise, and using that $\sigma>0$ along the real axis, we have
$$
\bm X^{\rm H}_{11,-}(\zeta)=\bm X^{\rm H}_{21,-}(\zeta)=0,\quad \zeta<0.
$$
With the aid of the jump relation \eqref{eq:XHjump} and \eqref{def:JXH} satisfied by $\bm X^{\rm H}$, we further conclude that
$$
\bm X^{\rm H}_{12,+}(\zeta)=\bm X^{\rm H}_{22,+}(\zeta)=0,\quad \zeta<0.
$$
Moreover, with the principal branch of $\zeta^{1/2}$,
the jump relation of $\bm X^{\rm H}$ also implies that for $j=1,2$, the function
$$
\zeta\mapsto 
\begin{cases}
    \bm X_{j2}^{\rm H}(\zeta),& \zeta \in \mcal S_+, \\ 
    -\ee^{\pi \ii\alpha +4\zeta^{1/2}}\bm X_{j1}^{\rm H}(\zeta),& \zeta \in \mcal S_-,
\end{cases}
$$
is analytic on a neighborhood of the interval $(-\infty,0)$, and by the identities we just obtained it vanishes on this interval. Hence, these functions are identically zero on $S_\pm$, and by analytic continuation we conclude that
\begin{equation}\label{eq:vanishingX1}
\bm X_{11}^{\rm H}(\zeta)=\bm X_{21}^{\rm H}(\zeta)=0, \, \im \zeta<0,\quad \text{and}\quad \bm X_{12}^{\rm H}(\zeta)=\bm X_{22}^{\rm H}(\zeta)=0, \, \im \zeta>0.
\end{equation}

As a final step, we now improve these identities to the vanishing of $\bm X^{\rm H}$ in the whole complex plane. For that, we account for \eqref{eq:vanishingX1} and write the jump condition \eqref{eq:XHjump} and \eqref{def:JXH}  for $\bm X^{\rm H}$ as
\begin{alignat*}{2}
& \bm X^{\rm H}_{j1,+}(\zeta)=\left(\bm X^{\rm H}_{j2}(\zeta) \ee^{-\pi \ii \alpha-4\zeta^{1/2}}\right)_-, \quad && \zeta<0, \\
& \bm X^{\rm H}_{j1,+}(\zeta)=\bm X^{\rm H}_{j2,-}(\zeta), && \zeta>0,
\end{alignat*}
where in the first line we use the principal branch of the square-root and $j=1,2$. Such jump conditions motivate us to define the scalar functions
\begin{equation}\label{def:gjH}
\msf g_j^{\rm H}(\zeta)\deff
\begin{cases}
\bm X_{j1}^{\rm H}(\zeta), & \im \zeta>0, \\
\bm X_{j2}^{\rm H}(\zeta), & \im \zeta<0.
\end{cases} 
\end{equation}

The properties of $\bm X^{\rm H}$ we just discussed then imply that $\msf g_j^{\rm H}$ is a solution to the following scalar RHP.

\begin{rhp} For $j=1,2$, the scalar function $\msf g_j^{\rm H}$  defined in \eqref{def:gjH} satisfies the following conditions.
\begin{enumerate}[(i)]
\item $\msf g_j^{\rm H}$ is analytic on $\C\setminus (-\infty,0]$.
\item For $\zeta<0$, it satisfies the jump relation
\begin{equation}\label{eq:jumpgjH}
    \msf g_{j,+}^{\rm H}(\zeta)=\msf g_{j,-}^{\rm H}(\zeta)\ee^{-\pi \ii \alpha +4\ii |\zeta|^{1/2}}.
\end{equation}
\item As $\zeta \to \infty$,
$$
\msf g_j^{\rm H}(\zeta)=\Boh(\zeta^{-3/4}).
$$
\item As $\zeta\to 0$,
$$
\msf g_1^{\rm H}(\zeta)=\Boh(\zeta^{\alpha/2}(1+\msf c_\alpha\log\zeta)), \quad \msf g_2^{\rm H}(\zeta)=\Boh(\zeta^{-\alpha/2}(1+\msf c_\alpha\log\zeta)),
$$
where $\msf c_\alpha\neq 0$ only for $\alpha\in \Z_{\geq 0}$.
\end{enumerate}
\end{rhp}
Let us now redefine
$$
\wh{\msf g}_j^{\rm H}(\zeta)=\msf g_j^{\rm H}(\zeta^2),\quad \re \zeta>0.
$$
Then $\wh{\msf g}_j^{\rm H}$ is analytic on $\re\zeta>0$ and it admits an analytic continuation to a sector containing the imaginary axis. In fact, the jump condition \eqref{eq:jumpgjH} for $\msf g_j^{\rm H}$ along $(-\infty,0)$ implies that for any $\theta \in (0,\pi/2)$, the function
$$
\wh{\msf g}_j^{\rm H}(\zeta)=
\begin{cases}
\msf g_j^{\rm H}(\zeta^2),&  \re \zeta>0, \\
\msf g_j^{\rm H}(\zeta^2)\ee^{-\pi \ii \alpha + 4\zeta },&  \frac{\pi}{2}< \arg\zeta <\frac{\pi}{2}+\theta, \\
\msf g_j^{\rm H}(\zeta^2)\ee^{\pi \ii \alpha + 4\zeta },&  -\frac{\pi}{2}-\theta < \arg\zeta <-\frac{\pi}{2}, 
\end{cases}
$$
is analytic across $\ii\R\setminus \{0\}$, and satisfies
$$
\wh{\msf g}_j^{\rm H}(\zeta)=\Boh(\zeta^{-3/2}), \; \zeta\to \infty \text{ with } -\frac{\pi}{2}-\theta<\arg\zeta<\frac{\pi}{2}+\theta,
$$
as well as
$$
\wh{\msf g}_1^{\rm H}(\zeta)=\Boh(\zeta^{\alpha}(1+\msf c_\alpha\log \zeta)),\quad \wh{\msf g}_2^{\rm H}(\zeta)=\Boh(\zeta^{-\alpha}(1+\msf c_\alpha\log \zeta)),\quad \zeta\to 0.
$$
Finally, let us further introduce
$$
\wt{\msf g}_j^{\rm H}(\zeta)=\left(\frac{\zeta}{1+\zeta}\right)^{|\alpha|+1}\wh{\msf g}_j^{\rm H}(\zeta).
$$
It follows that $\wt {\msf g}_j^{\rm H}$ is analytic on $-\frac{\pi}{2}-\theta<\arg\zeta<\frac{\pi}{2}+\theta$, it goes to $0$ as $\zeta\to \infty$ along the same sector, and it is also bounded as $\zeta\to 0$. Applying Phragmén-Lindel\"of principle \cite[Section~5.1, in particular Exercise~5.1.2]{Simon2015}, it follows that $\wt{\msf g}_j^{\rm H}\equiv 0$.

Tracing back the transformations, this shows that the entries of $\bm X^{\rm H}$ are all identically zero, as we wanted.
\end{proof}

\subsection{Differential equations for the model problem with some particular data}\label{sec:modelIntSys}\hfill 

Consider the model problem with the particular choice
\begin{equation}\label{def:hinfty}
\msf h(\zeta)=\msf h_\infty(\zeta\mid \sad,\uad)\deff \sad +\msf u (-1)^m \zeta^m,\quad \sad\in \R,\quad \msf u>0.
\end{equation}
The index $\infty$ is used because later this choice $\msf h_\infty$ will appear as a limiting data.

Let
\begin{equation}\label{def:Psiinfty}
    \bm \Psi_\infty(\zeta)\deff \bm \Psi(\zeta\mid \msf h=\msf h_\infty,\alpha)
\end{equation}
be the associated solution to the RHP~\ref{RHP:model}. Thanks to Theorem~\ref{thm:solvability}, the solution $\bm \Psi_\infty$  exists uniquely. The goal of this subsection is to connect certain quantities associated to $\bm \Psi_\infty$ to integrable systems.
Recall that $\sigma_{\bm \Phi}$ was introduced in \eqref{deff:sigmaPhi}. It may be alternatively expressed as
\begin{equation}\label{eq:sigmaphiinfity}
    \sigma_{\bm\Phi}(\zeta)=\frac{1}{1+\ee^{-\msf h_\infty(\zeta/\msf u^{1/m})}}.
\end{equation}

For the following calculations, it turns out to be convenient to collapse the jumps outside the real axis down to the negative axis. This process is performed through the modification of $\bm \Psi_\infty$ given by
\begin{equation}\label{eq:PsiinftytoPhiinfty}
{\bm \Phi}_\infty(\zeta)\deff \msf u^{-\sp_3/(4m)}\bm \Psi_\infty(\zeta/\msf u^{1/m})\left( \bm I+\sigma_{\bm \Phi}(\zeta)^{-1}\left(\ee^{\pi \ii \alpha}\chi_{\mcal S_+}(\zeta)-\ee^{-\pi \ii \alpha}\chi_{\mcal S_-}(\zeta)\right)\bm E_{21} \right),\quad \zeta\in \C\setminus \Gamma,
\end{equation}
where we recall that the sectors $\mcal S_\pm$ are displayed in Figure~\ref{fig:Gammapm}. It then follows from RHP \ref{RHP:model} for $\bm \Psi_\infty$ and \eqref{def:Psiinfty} that $\bm \Phi_\infty$ satisfies the following RHP.

\begin{rhp}\label{rhp:Phiinfty}
Find a $2\times 2$ matrix-valued function $\bm \Phi_\infty:\C\setminus (-\infty,0]\to \C^{2\times 2}$ with the following properties.
\begin{enumerate}[(i)]
\item $\bm \Phi_\infty$ is analytic on $\C\setminus (-\infty,0]$.
\item The entries of $\bm \Phi_{\infty,\pm}$ are continuous along $(-\infty,0)$ and they are related by
\begin{equation}\label{eq:jumpPhiInfty}
\bm \Phi_{\infty,+}(\zeta)=\bm \Phi_{\infty,-}(\zeta)\bm J_{\bm \Phi_\infty}(\zeta), \quad \textrm{with} \quad \bm J_{\bm \Phi_\infty}(\zeta)\deff \ee^{\pi\ii\alpha\sp_3}+\sigma_{\bm \Phi}(\zeta)\bm E_{12}.
\end{equation}
\item There exists a matrix $\bm \Phi_{\infty,1}$, independent of $\zeta$, such that as $\zeta\to \infty$,
\begin{equation}\label{eq:Phiinftyinfty}
    \bm \Phi_\infty(\zeta)=\left( \bm I+\frac{\bm \Phi_{\infty,1}}{\zeta}+\Boh(\zeta^{-2}) \right)\zeta^{-\sp_3/4}\bm U_0\ee^{\xad\zeta^{1/2}\sp_3}\left( \bm I+(\ee^{\pi \ii \alpha}\chi_{\mcal S_+}(\zeta)-\ee^{-\pi \ii \alpha}\chi_{\mcal S_-}(\zeta))\bm E_{21} \right),
\end{equation}
where we have set
\begin{equation}\label{eq:uxrelation}
\xad\deff 2\msf u^{-\frac{1}{2m}}>0.
\end{equation}

\item There exists a matrix-valued function $\bm \Phi_{\infty,0}$ which is analytic near $\zeta=0$ for which
\begin{equation}\label{eq:Phiinfty0}
    \bm \Phi_{\infty}(\zeta)=\bm \Phi_{\infty,0}(\zeta)\zeta^{\alpha\sp_3/2}\left(\bm I+\msf a(\zeta)\bm E_{12}\right)\sigma_{\bm\Phi}(\zeta)^{-\sp_3/2},\quad \zeta\to 0,
\end{equation}
where $\msf a$ is as in \eqref{eq:hatUpsilonasyfactor}.
\end{enumerate}
\end{rhp}

For later use, we record that the coefficient $\bm \Phi_{\infty,1}$ is related with the coefficient $\bm\Psi_{1}=\bm\Psi_{\infty,1}(\msf h=\msf h_\infty)$ appearing in \eqref{eq:asymptRHPpsimodel} through the identity
\begin{equation}\label{eq:relationPhi1inftyPsi1infty}
\bm\Phi_{\infty,1}=\msf u^{1/m} \msf u^{-\sp_3/(4m)}\bm\Psi_{\infty,1}\msf u^{\sp_3/(4m)}=\frac{4}{\xad^2} 2^{-\sp_3/2}\xad^{\sp_3/2}\bm\Psi_{\infty,1}\xad^{-\sp_3/2}2^{\sp_3/2}.
\end{equation}

There are two major reasons for the change $\bm \Psi_\infty\mapsto \bm \Phi_\infty$. First, with the introduction of the characteristic functions in the right-most side of \eqref{eq:PsiinftytoPhiinfty} the matrix $\bm \Phi_\infty$ has jumps only along $(-\infty,0)$. This structure will be useful later when deriving the mentioned nonlocal equation. Second, as a consequence of the change of variables $\zeta\mapsto \zeta/\msf u^{1/m}$ in \eqref{eq:PsiinftytoPhiinfty}, we obtain that the jump matrix $\bm J_{\bm\Phi_\infty}$ for $\bm \Phi_\infty$ is independent of the external parameter $\xad$. 

As a consequence of the latter property, we obtain a differential equation for $\bm \Phi_\infty$ in the variable $\xad$. The argument is standard and goes as follows. Both $\bm \Phi_\infty$ and its derivative $\partial_\xad\bm\Phi_\infty$ have the same jump relations, and therefore the matrix $\bm A(\zeta)\deff \partial_\xad \bm \Phi_\infty(\zeta)\bm \Phi_\infty(\zeta)^{-1}$ is analytic across $(-\infty,0)$. In other words, $\bm A$ has an isolated singularity at $\zeta=0$, and it is analytic elsewhere in $\C$. From the explicit behavior of $\bm \Phi_\infty$ at $\zeta=0$ in \eqref{eq:Phiinfty0} we see that the singularity of $\bm A$ at $\zeta=0$ is removable, that is, all the entries of $\bm A$ are entire functions. From the behavior of $\bm \Phi_\infty$ at $\zeta=\infty$ in \eqref{eq:Phiinftyinfty} we obtain an explicit expression for $\bm A$ in terms of $\bm \Phi_{\infty,1}$. Writing
\begin{equation}\label{deff:pqr}
\bm \Phi_{\infty,1}=
\begin{pmatrix}
    \msf q & -\ii \msf p \\ \ii \msf r & -\msf q
\end{pmatrix}, \quad \text{with} \quad  \msf q=\msf q(\sad,\xad), \msf r=\msf r(\sad,\xad) \quad \text{and} \quad \msf p=\msf p(\sad,\xad)=\frac{2\ii }{\xad }(\bm \Psi_{\infty,1})_{12},
\end{equation}
we obtain
\begin{equation}\label{eq:LaxeqtionPsi}
\partial_\xad \bm \Phi_{\infty}(\zeta)=\bm A(\zeta)\bm \Phi_{\infty}(\zeta),\quad \text{with}\quad 
\bm A(\zeta)\deff 
\begin{pmatrix}
    \msf p & -\ii \\ \ii (\zeta -2\msf q) & -\msf p
\end{pmatrix}.
\end{equation}
Moreover, the $1/\zeta$ term in the large $\zeta$-expansion of $\partial_\xad \bm \Phi_{\infty}(\zeta) \bm \Phi_{\infty}(\zeta)^{-1}$
must vanish, and the $(1,2)$-entry therein gives us 
\begin{equation}\label{eq:psquare}
    \msf p^2=2\msf q+\partial_\xad \msf p.
\end{equation}

Although the matrix $\bm A$ is rather simple-looking, it is convenient to transform $\bm \Phi_{\infty}$ a little further, in order to obtain an ODE with a matrix coefficient with vanishing diagonal. To that end, we transform
\begin{equation}\label{eq:UpsilonPhi}
\bm \Upsilon(\zeta)\deff (\bm I+\ii \msf p \bm E_{21})\bm \Phi_\infty(\zeta).
\end{equation}
A direct calculation using \eqref{eq:LaxeqtionPsi} and \eqref{eq:psquare} then shows that $\bm \Upsilon$ satisfies the ODE
\begin{equation}\label{eq:LaxeqtionUpsilon}
\partial_\xad \bm \Upsilon(\zeta)=
\bm B(\zeta)\bm\Upsilon(\zeta),\qquad \bm B(\zeta)\deff \bm B_1\zeta+\bm B_0, \quad \bm B_1\deff \ii \bm E_{21},\quad \bm B_0\deff 2\ii \partial_\xad\msf p \bm E_{21}-\ii\bm E_{12}.
\end{equation}

This new ODE implies in particular the relation
\begin{equation}\label{eq:Upsilonrowrelation}
\bm \Upsilon_{2k}(\zeta)=\ii \partial_\xad\bm \Upsilon_{1k}(\zeta),\quad k=1,2.
\end{equation}
Furthermore, taking one more $\partial_\xad$-derivative in \eqref{eq:LaxeqtionUpsilon}, looking entrywise, and setting
\begin{equation}\label{eq:PhiPhikUpsilonk}
\bm \Upsilon_{1k}(\zeta) =\Phi_k(\zeta),\quad \Phi\deff \Phi_1=(\bm \Phi_\infty)_{11},
\end{equation}
we obtain
\begin{equation}\label{eq:PhikpartialPhik}
\partial^2_\xad \Phi_{k}(\zeta)=\left(\zeta+ 2\partial_\xad\msf p\right)\Phi_k(\zeta),\quad k=1,2.
\end{equation}

The function $\Phi$ plays a major role in our results. For all the calculations that follows, we understand
$$
\Phi(\zeta)=\Phi_+(\zeta)\quad \text{for }\zeta<0,
$$
having in mind also that by the jump condition \eqref{eq:jumpPhiInfty},
\begin{equation}\label{eq:jumpphiffction}
\Phi(\zeta)=\Phi_+(\zeta)=\ee^{\pi\ii\alpha}\Phi_-(\zeta),\quad \zeta<0.
\end{equation}
Observe that the very definition of $\Phi$ in \eqref{eq:PhiPhikUpsilonk} in terms of $\bm\Phi_\infty$, when combined with \eqref{eq:Phiinftyinfty}, yields the asymptotic behaviors
\begin{equation}\label{eq:asymptPhizetainftyplus}
\Phi(\zeta)=\left(1+\Boh(\zeta^{-1})\right)\frac{\zeta^{-1/4}}{\sqrt{2}}\ee^{\xad \zeta^{1/2}},\quad \zeta\to +\infty,
\end{equation}
and
\begin{equation}\label{eq:asymptPhizetainftyminus}
\Phi(\zeta)=\left(1+\Boh(\zeta^{-1})\right)\sqrt{2}\ee^{\pi\ii\alpha/2+\pi\ii/4}\zeta_+^{-1/4}\cos\left( \xad|\zeta|^{1/2}-\frac{\pi}{4}-\frac{\pi\alpha}{2} \right),\quad \zeta\to -\infty.
\end{equation}

In summary, equation~\eqref{eq:LaxeqtionUpsilon} is a first order ODE for $\bm \Upsilon$, which yields a second order ODE for its entries, given by \eqref{eq:PhikpartialPhik} and the relation \eqref{eq:Upsilonrowrelation}. Following the usual dressing method for RHPs, one would be tempted in obtaining a similar ODE for $\bm \Upsilon$ in the $\zeta$ variable, and from which a Lax pair whose compatibility condition should yield scalar ODEs/PDEs of interest. However, the key feature for obtaining \eqref{eq:LaxeqtionUpsilon} was the fact that the jump matrix $\bm J_{\bm \Upsilon}=\bm J_{\bm \Phi_\infty}$ for $\bm \Upsilon$ is independent of $\xad$, and because it is not constant as a function of $\zeta$ or $\sad$ it is not possible to obtain the ODE in the $\zeta$ or $\sad$ variables.

To overcome this issue, we follow a modification of the dressing method\footnote{We thank Thomas Chouteau for introducing this technique to us.}.
Start by defining
$$
\bm W_0=\ii \partial_\sad \msf p \bm E_{21},
$$
which is introduced so that
\begin{equation}
\partial_\sad \bm \Upsilon(\zeta)\bm \Upsilon(\zeta)^{-1}-\bm W_0=\Boh(\zeta^{-1}),\quad \zeta\to \infty.
\end{equation}
For later reference, we state the expansion
\begin{equation}\label{eq:behpreVinfty}
\partial_\sad \bm \Upsilon(\zeta)\bm \Upsilon(\zeta)^{-1}=\bm W_0+ \frac{1}{\zeta}\bm W_{-1}+ \Boh(\zeta^{-2}),\quad \zeta\to \infty,
\end{equation}
with
$$
\bm W_{-1}\deff 
\begin{pmatrix}
    \partial_\sad \msf q-\msf p\partial_\sad \msf p & -\ii \partial_\sad \msf p \\ 
    \ii \left( \msf p^2\partial_\sad \msf p-2\msf p\partial_\sad \msf q-\partial_\sad \msf r \right) & -\partial_\sad \msf q+\msf p\partial_\sad \msf p
\end{pmatrix},
$$
as well as
\begin{multline}\label{eq:behpreVorigin}
\partial_\sad \bm \Upsilon(\zeta)\bm \Upsilon(\zeta)^{-1}-\bm W_0
\\
=\msf a(\zeta)\zeta^\alpha\partial_\sad \log\sigma_{\bm \Phi}(\zeta) (\bm I+\ii\msf p\bm E_{21})\bm \Phi_{\infty,0}(0)\bm E_{12}\bm \Phi_{\infty,0}(0)^{-1}(\bm I-\ii\msf p\bm E_{21})+\Boh(1),\quad \zeta\to 0.
\end{multline}

Moving further, introduce
\begin{equation}\label{eq:deffVlaseqtion}
\bm V(\zeta)\deff \bm \Upsilon(\zeta)^{-1}\partial_\sad\bm \Upsilon(\zeta)-\bm \Upsilon(\zeta)^{-1}\bm W_0\bm \Upsilon (\zeta),\quad \zeta\in \C\setminus (-\infty,0],
\end{equation}
so that
\begin{equation}\label{eq:LaxeqtionUpsilonVs}
\partial_\sad \bm\Upsilon(\zeta)=\bm W_0\bm \Upsilon(\zeta)+\bm\Upsilon(\zeta)\bm V(\zeta),\quad \zeta\in \C\setminus (-\infty,0].
\end{equation}
Then, the product $\bm \Upsilon\bm V\bm\Upsilon^{-1}$ is analytic on $\C\setminus (-\infty,0]$, and thanks to \eqref{eq:behpreVinfty}--\eqref{eq:behpreVorigin},
\begin{equation}\label{eq:behUpsVUps}
\bm \Upsilon(\zeta)\bm V(\zeta)\bm\Upsilon(\zeta)^{-1}=\partial_\sad \bm \Upsilon(\zeta)\bm \Upsilon(\zeta)^{-1}-\bm W_0=\begin{cases}
\Boh(\zeta^{-1}),&  \zeta\to \infty, \\
\Boh(1+\msf a(\zeta)\zeta^\alpha),&  \zeta\to 0.
\end{cases}
\end{equation}
We now compute additive jump relations for $\bm \Upsilon\bm V\bm\Upsilon^{-1}$ along $(-\infty,0)$. The matrix $\bm W_0$ is constant and therefore it has no jumps, so that
\begin{align*}
(\bm \Upsilon(\zeta)\bm V(\zeta)\bm\Upsilon(\zeta)^{-1})_+-(\bm \Upsilon(\zeta)\bm V(\zeta)\bm\Upsilon(\zeta)^{-1})_- & =(\partial_\sad\bm\Upsilon(\zeta)\bm\Upsilon(\zeta)^{-1})_+-(\partial_\sad\bm\Upsilon(\zeta)\bm\Upsilon(\zeta)^{-1})_- \\ 
& = \bm\Upsilon_+(\zeta)\bm J_{\bm \Upsilon}(\zeta)^{-1}\partial_\sad \bm J_{\bm \Upsilon}(\zeta)\bm \Upsilon_+(\zeta)^{-1}\\
& = \ee^{-\pi\ii\alpha}\partial_\sad \sigma_{\bm \Phi}(\zeta)\bm \Upsilon_{+}(\zeta)\bm E_{12}\bm \Upsilon_+(\zeta)^{-1}.
\end{align*}

From \eqref{eq:Upsilonrowrelation}, \eqref{eq:PhiPhikUpsilonk} and the identity $\det\bm\Upsilon\equiv 1$ we compute $\bm \Upsilon\bm E_{12}\bm \Upsilon^{-1}$, finally obtaining
\begin{equation}\label{eq:jumpsUpsVUps}
(\bm \Upsilon(\zeta)\bm V(\zeta)\bm\Upsilon(\zeta)^{-1})_+-(\bm \Upsilon(\zeta)\bm V(\zeta)\bm\Upsilon(\zeta)^{-1})_-= \ee^{-\pi \ii \alpha}\partial_\sad \sigma_{\bm \Phi}(\zeta)
\begin{pmatrix}
    -\ii \Phi(\zeta)\partial_\xad \Phi(\zeta) & \Phi(\zeta)^2 \\ (\partial_\xad \Phi(\zeta))^2 & \ii \Phi(\zeta)\partial_\xad \Phi(\zeta)
\end{pmatrix},
\end{equation}
where in the above and in what follows we recall that we understand $\Phi(\zeta)=\Phi_+(\zeta)$ for $\zeta<0$.

The jump relation \eqref{eq:jumpsUpsVUps} yields a scalar additive RHP for $\bm \Upsilon\bm V\bm\Upsilon^{-1}$ along the jump contour $(-\infty,0)$. When coupled with the boundary conditions \eqref{eq:behUpsVUps} on the endpoints $\infty,0$, it can be uniquely solved using Plemelj's formula. In order to verify such solution, we still need to understand the behavior of its jump, given by the right-hand side of \eqref{eq:jumpsUpsVUps}, near the endpoints $\infty,0$, which we do next. 

With $\bm \Phi_{\infty,0}$ being the analytic function on the right-hand side of \eqref{eq:Phiinfty0}, denote
$$
\Phi_0(\zeta)\deff (\bm \Phi_{\infty,0})_{11}.
$$
A direct calculation using \eqref{eq:Phiinfty0} and \eqref{eq:UpsilonPhi} shows that
$$
\Phi(\zeta)\sim \Phi_0(\zeta)\zeta^{\alpha/2}\sigma_{\bm \Phi}(\zeta)^{-1/2},\quad \zeta\to 0.
$$
Therefore,
$$
\partial_\sad \sigma_{\bm \Phi}(\zeta)
\begin{pmatrix}
    -\ii \Phi(\zeta)\partial_\xad \Phi(\zeta) & \Phi(\zeta)^2 \\ (\partial_\xad \Phi(\zeta))^2 & \ii \Phi(\zeta)\partial_\xad \Phi(\zeta)
\end{pmatrix}\sim
\zeta^\alpha \frac{\partial_\sad \sigma_{\bm \Phi}(\zeta)}{\sigma_{\bm \Phi}(\zeta)}
\begin{pmatrix}
    -\ii \Phi_0(\zeta)\partial_\xad \Phi_0(\zeta) & \Phi_0(\zeta)^2 \\ (\partial_\xad \Phi_0(\zeta))^2 & \ii \Phi_0(\zeta)\partial_\xad \Phi_0(\zeta)
\end{pmatrix},
$$
and together with the fact that $\zeta\mapsto \partial_\sad\log\sigma_{\bm\Phi}(\zeta)$ is bounded along the real axis (see \eqref{deff:sigmaPhi}), we learn
\begin{equation}\label{eq:JumpYVYbehorigin}
\partial_\sad \sigma_{\bm \Phi}(\zeta)
\begin{pmatrix}
    -\ii \Phi(\zeta)\partial_\xad \Phi(\zeta) & \Phi(\zeta)^2 \\ (\partial_\xad \Phi(\zeta))^2 & \ii \Phi(\zeta)\partial_\xad \Phi(\zeta)
\end{pmatrix}
=\Boh(1+\zeta^\alpha),\quad \zeta\to 0.
\end{equation}

Likewise, from the behavior \eqref{eq:Phiinftyinfty} we get the rough estimate $\Phi(\zeta)=\Boh(\zeta^{-1/4})$ as $\zeta\to -\infty$, and from the explicit expression for $\sigma_{\bm\Phi}$ in \eqref{deff:sigmaPhi} we learn that
$$
\partial_\sad \sigma_{\bm \Phi}(\zeta)
\begin{pmatrix}
    -\ii \Phi(\zeta)\partial_\xad \Phi(\zeta) & \Phi(\zeta)^2 \\ (\partial_\xad \Phi(\zeta))^2 & \ii \Phi(\zeta)\partial_\xad \Phi(\zeta)
\end{pmatrix}
=\Boh(\ee^{-|\zeta|^m}),\quad \zeta\to -\infty.
$$
This estimate is not necessarily uniform in $\sad$ but it is sufficient for our purposes.

These last two estimates show that the right-hand side of \eqref{eq:jumpsUpsVUps} is integrable over the negative axis, so its Cauchy transform is well-defined. They also ensure that this Cauchy transform satisfies both local behaviors \eqref{eq:behUpsVUps}. As a consequence, we finally learn that
$$
\bm \Upsilon(\zeta)\bm V(\zeta)\bm\Upsilon^{-1}(\zeta)=\frac{\ee^{-\pi \ii \alpha}}{2\pi \ii}\int_{-\infty}^0 
\begin{pmatrix}
    -\ii \Phi(\xi)\partial_\xad \Phi(\xi) & \Phi(\xi)^2 \\ (\partial_\xad \Phi(\xi))^2 & \ii \Phi(\xi)\partial_\xad \Phi(\xi) 
\end{pmatrix}
\frac{\partial_\sad \sigma_{\bm \Phi}(\xi)\dd \xi}{\xi-\zeta},\quad \zeta\in \C\setminus (-\infty,0].
$$

Let us combine what we have so far. Equations~\eqref{eq:LaxeqtionUpsilon} and \eqref{eq:LaxeqtionUpsilonVs} form a Lax pair,
$$
\partial_\xad \bm \Upsilon (\zeta)=\bm B(\zeta)\bm \Upsilon(\zeta)\quad \text{and}\quad \partial_\sad \bm \Upsilon (\zeta)=\bm \Upsilon(\zeta)\bm V(\zeta) + \bm W_0\bm \Upsilon(\zeta).
$$
The compatibility between these two equations reads
\begin{equation}\label{eq:comptBWV}
\partial_\sad \bm B(\zeta)-\partial_\xad \bm W_0+[\bm B(\zeta),\bm W_0]=[\bm \Upsilon(\zeta)\bm V(\zeta)\bm \Upsilon(\zeta)^{-1},\bm B(\zeta)]+\partial_\xad (\bm \Upsilon(\zeta)\bm V(\zeta)\bm \Upsilon(\zeta)^{-1}).
\end{equation}
The terms $\bm B$ and $\bm W_0$ are matrix-valued rational functions in $\zeta$, whereas $\bm \Upsilon\bm V\bm \Upsilon^{-1}$ - albeit more complicated - admits an asymptotic expansion as $\zeta\to \infty$ of the form
\begin{equation}\label{eq:asympexpUpsV}
\bm \Upsilon(\zeta)\bm V(\zeta)\bm\Upsilon^{-1}(\zeta)\sim -\sum_{k=1}^\infty \frac{1}{\zeta^k}\bm V_{-k},
\end{equation}
with coefficients
$$
\bm V_{-k}\deff \frac{\ee^{-\pi \ii \alpha}}{2\pi \ii} \int_{-\infty}^0 
\begin{pmatrix}
    -\ii \Phi(\xi)\partial_\xad \Phi(\xi) & \Phi(\xi)^2 \\ (\partial_\xad \Phi(\xi))^2 & \ii \Phi(\xi)\partial_\xad \Phi(\xi) 
\end{pmatrix}
\xi^{k-1} \partial_\sad\sigma_{\bm \Phi}(\xi)\dd\xi,\qquad k\in\mathbb{Z}_{>0}. 
$$
These coefficients clearly satisfy $\tr\bm V_{-k}=0$ for every $k$.

Expanding \eqref{eq:comptBWV} at $\zeta=\infty$ in powers of $\zeta$, we obtain a system of equations between the coefficients of $\bm B,\bm W_0$ and the coefficients $\bm V_{-k}$. The first two of such nontrivial identities read
\begin{equation}\label{eq:LPVBW}
\left\{
\begin{aligned}
& [\bm V_{-1},\bm B_1]+[\bm B_0,\bm W_0]+\partial_\sad \bm B_0-\partial_\xad \bm W_0=0,\\
& [\bm V_{-1},\bm B_0]+[\bm V_{-2},\bm B_1]+\partial_\xad \bm V_{-1}=0.
\end{aligned}
\right.
\end{equation}

Looking at the $(1,1)$-entry of the first relation in \eqref{eq:LPVBW}, we obtain the identity
\begin{equation}\label{eq:identitypartialspphi}
\partial_\sad\msf p=-\ii (\bm V_{-1})_{12}=-\frac{\ee^{-\pi \ii \alpha}}{2\pi}\int_{-\infty}^0\Phi(\xi)^2\partial_\sad \sigma_{\bm \Phi}(\xi)\dd\xi.
\end{equation}
The next step is to integrate this identity over the interval $(\sad,\infty)$, and then take an $x$-derivative of the result. As a consequence of Proposition~\ref{prop:estimatesBessellimitmodelpphi} that will be proved later,
$$
\msf p\to -\frac{4\alpha^2-1}{8\xad } \quad \text{as}\quad \sad\to +\infty.
$$
Also thanks to Proposition~\ref{prop:estimatesBessellimitmodelpphi}, we know that for some function $f$ which is integrable in compacts of $(-\infty,0]$, grows with $\Boh(|\zeta|^{1/4})$ as $\zeta\to -\infty$, and is independent of $\sad$ and $\xad$, the bound
$$
|\Phi(\zeta)|+|\partial_\xad \Phi(\zeta)|\leq f(\zeta)
$$
is valid. From the explicit form of $\sigma_\Phi$ in \eqref{deff:sigmaPhi},
$$
|\partial_\sad\sigma_\Phi(\zeta)|= \frac{\ee^{-\sad -(-1)^m\zeta^m}}{\left(1+\ee^{-\sad -(-1)^m\zeta^m}\right)^2}\leq \ee^{-\sad -(-1)^m\zeta^m},\qquad \zeta<0.
$$
Hence, $(|\Phi|+|\partial_\xad\Phi|)\partial_\sad\sigma_{\bm\Phi}$ is bounded by an $(\zeta,\sad)$-integrable function independent of $\xad$, and we can indeed integrate \eqref{eq:identitypartialspphi} in the variable $\sad$, obtaining the identity
\begin{equation}\label{deff:pcoreint}
\msf p(\sad,\xad)=-\frac{4\alpha^2-1}{8\xad}+\frac{\ee^{-\pi\ii\alpha}}{2\pi}\int_{\sad}^\infty \int_{-\infty}^0 
\Phi(\xi\mid \sad=u,\xad)^2\partial_\sad \sigma_{\bm\Phi}(\xi\mid \sad=u)\dd\xi \dd u.
\end{equation}

Differentiating now in $\xad$, we obtain
$$
\partial_\xad \msf p(\sad,\xad)=\frac{4\alpha^2-1}{8\xad^2}+\frac{\ee^{-\pi\ii\alpha}}{\pi}\int_\sad^\infty \int_{-\infty}^0 \Phi(\xi\mid \sad=u,\xad)\partial_\xad \Phi(\xi\mid \sad=u,\xad)\partial_\sad \sigma_{\bm \Phi}(\xi\mid \sad=u)\dd\xi\dd u.
$$

Thus, combining this relation with \eqref{eq:PhikpartialPhik}, we finally conclude that
\begin{theorem}\label{thm:Phinonlocalcore}
    For $(\zeta,\sad,\xad)\in \R\times \R\times (0,+\infty)$, the function $\Phi=\Phi(\zeta\mid \sad,\xad)$ satisfies the nonlocal PDE
    \begin{equation}\label{eq:PDEPhi}
    \partial_\xad^2\Phi(\zeta\mid \sad,\xad)=
     \left( 
     \zeta+\frac{4\alpha^2-1}{4\xad^2}+\frac{2\ee^{-\pi\ii\alpha}}{\pi}\int_\sad^\infty \int_{-\infty}^0 \Phi(\xi\mid u,\xad)\partial_\xad \Phi(\xi\mid u,\xad)\partial_\sad \sigma_{\bm \Phi}(\xi\mid u)\dd\xi\dd u 
     \right)\Phi(\zeta\mid \sad,\xad).
    \end{equation}
\end{theorem}

For reference, we now obtain a formula that will later be used to relate the limiting kernel for the OP ensemble \eqref{eq:Lagdef} with the function $\Phi$ just constructed. Using the relation \eqref{eq:PsiinftytoPhiinfty}, we obtain the identity
\begin{multline*}
\left[
\left(\bm I+\ee^{-\pi \ii \alpha}(1+\ee^{-\msf h_\infty(\xi)})\chi_{\mcal S_-}(\xi)\bm E_{21}\right)
\bm \Psi_{\infty}(\xi)^{-1}\bm\Psi_{\infty}(\zeta)
\left(\bm I-\ee^{-\pi \ii \alpha}(1+\ee^{-\msf h_\infty(\zeta)})\chi_{\mcal S_-}(\zeta)\bm E_{21}\right)
\right]_{21,-}=
\\
\left[
\bm \Phi_{\infty}(\msf u^{1/m}\xi)^{-1}\bm\Phi_{\infty}(\msf u^{1/m}\zeta)
\right]_{21,-},
\end{multline*}
valid for $\zeta,\xi\in \R\setminus\{0\}$. Using \eqref{eq:uxrelation}, \eqref{eq:UpsilonPhi}, \eqref{eq:Upsilonrowrelation} and \eqref{eq:PhiPhikUpsilonk}, we obtain
\begin{multline}\label{eq:kernelPsiinftyPhi}
\left[
\left(\bm I+\ee^{-\pi \ii \alpha}(1+\ee^{-\msf h_\infty(\xi)})\chi_{\mcal S_-}(\xi)\bm E_{21}\right)
\bm \Psi_{\infty}(\xi)^{-1}\bm\Psi_{\infty}(\zeta)
\left(\bm I-\ee^{-\pi \ii \alpha}(1+\ee^{-\msf h_\infty(\zeta)})\chi_{\mcal S_-}(\zeta)\bm E_{21}\right)
\right]_{21,-}=
\\
\ii \ee^{-2\pi\ii\alpha} \left[  \Phi(4\xi/\xad^2)(\partial_\xad \Phi )(4\zeta/\xad^2)-\Phi(4\zeta/\xad^2)(\partial_\xad \Phi )(4\xi/\xad^2)  \right],
\end{multline}
where $\Phi$ is the same function appearing in \eqref{eq:PDEPhi}, and we used \eqref{eq:jumpphiffction}.

In particular, comparing with \eqref{deff:Kalpha} we obtain the RHP representation of the kernel $\msf K_\alpha$,
\begin{multline}\label{eq:KalphaRHPrepr}
\msf K_\alpha(\zeta,\xi)=\frac{\ee^{\pi\ii \alpha}\sqrt{\sigma_{\bm \Phi}\left(-4\xi/\xad^2 \right)}\sqrt{\sigma_{\bm \Phi}\left(-4\zeta /\xad^2 \right)}}{2\pi\ii(\xi-\zeta)} \\
\times 
\left[
\left(\bm I+\frac{\ee^{-\pi \ii \alpha}}{\sigma_{\bm \Phi}\left(-4\xi/\xad^2 \right)}\bm E_{21}\right)
\bm \Psi_{\infty}(-\xi)^{-1}\bm\Psi_{\infty}(-\zeta)
\left(\bm I-\frac{\ee^{-\pi \ii \alpha}}{\sigma_{\bm \Phi}\left(-4\zeta /\xad^2 \right)}\bm E_{21}\right)
\right]_{21,-},
\end{multline}
which is valid for $\xi,\zeta \in (0,\infty)$.

Recall the notation \eqref{deff:Deltax}. Sending $\xi\to \zeta$, we also obtain the identity
\begin{equation}\label{eq:KalphaRHPreprdiag}
\left[
\bm\Delta_\zeta\left[
\bm\Psi_\infty(\zeta)\left(\bm I-\frac{\ee^{-\pi\ii\alpha}}{\sigma_{\bm \Phi}\left(4\zeta/\xad^2 \right)}\bm E_{21}\right)
\right]
\right]_{21,-}=
\frac{2\pi\ii \ee^{-\pi\ii\alpha} }{\sigma_{\bm \Phi}\left(4\zeta /\xad^2 \right)} \msf K_\alpha(-\zeta,-\zeta),\quad \zeta<0.
\end{equation}
These representations will be useful later. We now move on to computing relevant asymptotics for the model problem.

\section{The model RHP: admissible data and asymptotics}\label{sec:modelproblemasymptotics}

In the course of the analysis of OPs, the model problem needed will involve a function $\msf h=\msf h_\gt$ which depends on a large parameter $\gt>0$. Later on, we will set this parameter $\gt$ to be the number of particles $n$, but during this section we keep $\gt$ as a free parameter. In the limit $\gt\to\infty$, the function $\msf h_\gt$ will converge to $\msf h_\infty$ from \eqref{def:hinfty}. The goal of the current section is to perform the required asymptotic analysis of the corresponding $\gt$-dependent model problem $\bm \Psi_\gt=\bm\Psi(\cdot\mid \msf h=\msf h_\gt)$, showing rigorously that it converges to the model problem $\bm\Psi_\infty=\bm\Psi_\infty(\cdot\mid \msf h=\msf h_\infty)$ from \eqref{def:Psiinfty}.

The remainder of this section is structured as follows. In Section~\ref{sec:admmodel} we introduce the conditions on the function $\msf h=\msf h_\gt$ that we will be working with. This function $\msf h_\gt$ will also depend on the two additional parameters $\sad\in \R$ and $\xad>0$. In Section~\ref{sec:boundsPsiinfty} we analyze the model problem $\bm\Psi_\infty$ in the degenerate limit $\sad \to +\infty$. This will be needed in order to integrate $\bm\Psi_\infty$ in the variable $\sad$, to cope with the limit of the deformation formula \eqref{eq:deffformula}. This analysis in the limit $\sad\to +\infty$ will also provide the major steps for the asymptotic analysis when $\xad\to 0^+$; the latter will be completed in Section~\ref{sec:Psiinftylargeuanalysis}. Finally, in Section~\ref{sec:analysisgtinfty} we complete the third asymptotic analysis necessary for the model problem, namely the one in the limit $\gt\to +\infty$.

\subsection{The class of admissible $\msf h$}\label{sec:admmodel}  \hfill 

We now consider the model problem $\bm \Psi$ associated to a function $\msf h$ which depends on a (large) parameter $\gt$ in a certain structured manner, as we introduce next. For the next statement, recall that the contour $\Gamma$ was introduced in \eqref{deff:thetambasiccontours} and also Figure~\ref{fig:Gammapm}.

\begin{definition}\label{def:admissibledata}
A function $\msf h:\Gamma\to \C$ is {\it admissible} if it is of the form
$$
\msf h(\zeta)=\msf h_\gt(\zeta\mid \sad)\deff \sad +\gt^{2m}\msf H\left(\frac{\zeta}{\gt^2}\right),
$$
where $\sad\in \R$, $\gt>0$, $m\in \Z_{>0}$ are parameters, and $\msf H:\Gamma\to \C$ is a fixed function with the following properties.
\begin{enumerate}[(i)]
\item The function $\msf H$ is $C^\infty$ in a neighborhood of $\Gamma$ and real-valued along $\Gamma_0$.
\item The function $\msf H$ is analytic on a disk $D_\delta$ with
$$
\msf H(\zeta)=(-1)^m \msf u \zeta^m+\Boh(\zeta^{m+1}),\quad \zeta\to 0,
$$
where $\msf u>0$.
\item There exists $\Lambda>0$ for which
$$
\re \msf H(\zeta)\geq \Lambda |\zeta|^m,\quad \zeta \in \Gamma.
$$
\end{enumerate}
\end{definition}

For us we are interested in the choice of parameter $\gt=n$, with $n$ as in \eqref{eq:Lagdef}, but in Definition~\ref{def:admissibledata} and for the rest of the present section we opt for introducing the new dummy parameter $\gt$, to emphasize that the model problem also has an interest on its own.

With $\dd s$ being the complex line element in $\Gamma$, Definition~\ref{def:admissibledata}--(iii) implies, in particular, that
$$
|v|^k\ee^{-\msf H}\in L^1(\Gamma,|\dd v|)\cap L^\infty(\Gamma,|\dd v|),\quad \text{for any } k>0.
$$


Any admissible $\msf h_\gt$ is analytic on a disk centered at the origin with growing radius. Moreover, the expansion
\begin{equation}
\msf h_\gt(\zeta)=\sad +(-1)^m\msf u\zeta^m (1+\Boh(\gt^{-2}))
\end{equation}
holds uniformly for $\zeta$ in compacts of $\Gamma$, with an error term which is independent of $\sad$. With $\msf h_\infty$ as in \eqref{def:hinfty}, this shows in particular that
\begin{equation}\label{eq:httohinfty}
\msf h_\gt(\zeta)=\msf h_\infty(\zeta)+\Boh(\gt^{-2}), \quad \gt\to \infty,
\end{equation}
again with the convergence being uniform in compacts of $\Gamma$, with an error term which is independent of $\sad\in \R$. Observe, in particular, that the right-hand side coincides with the data studied in Section~\ref{sec:modelIntSys}. Since we are assuming $\msf s\in \R$, the convergence above also shows that $1+\ee^{-\msf h}$ is analytic and non-vanishing on a neighborhood of the origin, with this neighborhood being independent of $\gt$. Therefore, Proposition~\ref{Prop:behmodelproblemorigin} applies for $\bm \Psi$ with admissible $\msf h=\msf h_\gt$, for any $\gt$ sufficiently large, a fact which will be used for the rest of the paper.

From now on, we fix $\msf h=\msf h_\gt$ to always be an admissible function in the sense of Definition~\ref{def:admissibledata} and denote
\begin{equation}\label{deff:Psitauadmissibleh}
\bm \Psi_\gt(\zeta)\deff \bm \Psi(\zeta\mid \msf h_\gt,\alpha),\quad \bm \Psi_\infty(\zeta)\deff \bm \Psi(\zeta\mid \msf h_\infty,\alpha),
\end{equation}
with corresponding jump matrices
$$
\bm J_\gt=\bm J_{\bm \Psi_\gt},\quad \bm J_\infty=\bm J_{\bm \Psi_\infty}.
$$
This notation is consistent with \eqref{def:Psiinfty}. The convergence \eqref{eq:httohinfty} indicates that $\bm \Psi_\gt\to \bm \Psi_\infty$. The main goal of the remainder of this section is to perform the asymptotic analysis to justify this convergence, and which should be valid uniformly for $\sad \geq \sad_0$ with any fixed $\sad_0\in \R$. 

\subsection{Asymptotic analysis for $\bm \Psi_\infty$ as $\sad\to +\infty$}\label{sec:boundsPsiinfty}\hfill

As said, we need to compare $\bm \Psi_{\gt}$ with $\bm \Psi_{\infty}$ with bounds valid uniformly for $\sad \in [\sad_0,\infty)$, where $\sad_0\in \R$ fixed. In that analysis, we will also require bounds on appropriate norms of $\bm \Psi_\infty$ valid uniformly for $\sad\geq \sad_0$ and also with some rough control as $\zeta\to \infty$ along $\Gamma$. This section is devoted to obtain such bounds.

We start analyzing $\bm \Psi_\infty$ when $\sad\to +\infty$. To that end, set
$$
\bm J^\bes(\zeta)\deff
\begin{cases}
\bm I+\ee^{\pm\pi \ii \alpha}\bm E_{21},& \zeta\in \Gamma_\pm, \\
\bm E_{12}-\bm E_{21}, & \zeta\in \Gamma_0.
\end{cases}
$$
That is, $\bm J^\bes$ is the jump matrix for the model problem obtained with data $\ee^{-\msf h}\equiv 0$, corresponding to setting $\msf h=+\infty$. Let $\bm \Psi^\bes_\alpha$ be the solution to the corresponding RHP, such solution is standard and can be constructed explicitly using Bessel functions, see Appendix~\ref{sec:BesselParametrix} for more details.

The pointwise convergence
$$
\bm J_{\infty}\to \bm J^\bes,\quad \sad\to +\infty,
$$
is straightforward, and leads us to compare $\bm \Psi_\infty$ with $\bm \Psi^\bes_\alpha$. In fact, for $\alpha>-1/2$ one could prove directly that the jump matrix for $\bm \Psi_\infty(\bm\Psi^\bes_\alpha)^{-1}$ converges to the identity matrix in $L^2\cap L^\infty$ as $\sad\to +\infty$, which would suffice to apply small norm theory to estimate norms of $\bm \Psi_\infty$ and of its boundary values. However, we could not improve such a direct argument also to the range $-1<\alpha\leq -1/2$. The alternate route still involves applying the steepest descent method, but the construction of a local parametrix is needed, as we discuss next.

 Let $\bm \Psi_0=\bm \Psi_{\infty,0}$ be the analytic function at the origin obtained when we apply Proposition~\ref{Prop:behmodelproblemorigin} to $\bm \Psi_\infty$, and $\bm \Psi_{0}=\bm\Psi_{\alpha,0}^\bes$ the one when we apply the same proposition to $\bm \Psi^\bes_\alpha$. For a sufficiently small $\varepsilon>0$ which we keep fixed from now on, the function $\bm \Psi_{\infty,0}$ is analytic on $D_\varepsilon$. 

 The function
\begin{equation}\label{eq:lambdainfty}
\lambda(\zeta)\deff \frac{1}{2\pi \ii}\int_{-\infty}^0 \left(\frac{1}{1+\ee^{-\msf h_\infty(s)}}-1\right)\frac{|s|^{\alpha}}{s-\zeta}\dd s,\quad \zeta\in \C\setminus \Gamma_0,
\end{equation}
is analytic on $\C\setminus \Gamma_0$, and satisfies
\begin{equation}\label{eq:jumpslambda}
\lambda_+(\zeta)-\lambda_-(\zeta)=\left(\frac{1}{1+\ee^{-\msf h_\infty(\zeta)}}-1\right)|\zeta|^{\alpha},\quad \zeta<0.
\end{equation}
Also, from basic properties of Cauchy transforms, we obtain the behavior
\begin{equation}\label{eq:lambdabeh}
\lambda(\zeta)=
\begin{cases}
\Boh(\zeta^{\alpha}), & \text{if } -1<\alpha<0, \\
\Boh(\log \zeta), & \text{if } \alpha= 0, \\
\Boh(1), & \text{if } \alpha>0,
\end{cases}
\qquad \zeta\to 0.
\end{equation}

We will use $\lambda$ to construct the local parametrix, and some estimates on $\lambda$ as $\sad\to+\infty$ will be needed. Using the inequality
$$
\left| \frac{1}{1+\ee^{-\msf h_\infty(\zeta)}}-1 \right|=\left|\frac{\ee^{-\msf h_\infty(\zeta)}}{1+\ee^{-\msf h_\infty(\zeta)}}\right|\leq \ee^{-\msf h_\infty(\zeta)}=\ee^{-\sad} \ee^{-|\zeta|^m\msf u}, \quad \zeta<0,
$$
we obtain that 
\begin{equation}\label{eq:boundlambdabessanalysis}
|\lambda(\zeta)|\leq \frac{\ee^{-\sad}}{2\pi}\int_{-\infty}^0 \frac{|s|^\alpha\ee^{-\msf u |s|^m}}{|\zeta-s|}\dd s,\quad \zeta\in \C\setminus \Gamma_0.
\end{equation}
Recalling that $\msf u>0$, we learn that $\lambda \to 0$ uniformly for $\zeta$ in compacts of $\C\setminus \Gamma_0$ as $\sad\to +\infty$. Deforming contours in the integral defining $\lambda$, we see that this convergence also extends uniformly to the whole plane $\C$; more precisely
\begin{equation}\label{eq:behlambdainfty}
\|\lambda\|_{L^\infty(\C\setminus \Gamma_0)}=\Boh(\ee^{-\sad}),\quad \|\lambda_\pm\|_{L^\infty(\Gamma_0)}=\Boh(\ee^{-\sad}),\quad \sad\to +\infty.
\end{equation}
 
We use this function $\lambda$ and set
\begin{multline}\label{eq:localparPsiinfty}
\bm P_0(\zeta)\deff \bm \Psi^\bes_{\alpha,0}(\zeta)\left(\bm I+\lambda(\zeta)\bm E_{12}\right)\bm \zeta^{\alpha\sp_3/2}(\bm I+\msf a(\zeta)\bm E_{12})\\
\times (\bm I-(\chi_{\mcal S_+}(\zeta)\ee^{\pi \ii\alpha}-\chi_{\mcal S_-}(\zeta)\ee^{-\pi \ii \alpha})(1+\ee^{-\msf h_\infty(\zeta)})\bm E_{21}), \quad \zeta\in  D_\varepsilon.
\end{multline}
Next, we define
\begin{equation}\label{deff:Rbesselasymp}
\bm R(\zeta)\deff
\begin{cases}
\bm \Psi_\infty(\zeta)\bm \Psi_\alpha^\bes(\zeta)^{-1}, & \zeta\in \C\setminus (\Gamma\cup \overline{D}_\varepsilon), \\
\bm \Psi_\infty(\zeta)\bm P_0(\zeta)^{-1}, & \zeta\in D_\varepsilon\setminus \Gamma.
\end{cases}
\end{equation}
It is clear that $\bm R$ is analytic on $\C\setminus (\Gamma\cup \partial D_\varepsilon)$. Using the jump properties of $\msf a$ and $\lambda$ given in \eqref{eq:jumpafactor} and \eqref{eq:jumpslambda}, a cumbersome but straightforward calculation shows that
$$
\bm P_{0,+}(\zeta)=\bm P_{0,-}(\zeta)\bm J_\infty(\zeta),\quad \zeta\in \Gamma\cap D_\varepsilon,
$$
that is, $\bm P_0$ and $\bm \Psi_\infty$ satisfy the same jumps for $\zeta\in \Gamma\cap D_\varepsilon$, implying in particular that $\bm R$ has an isolated singularity at $\zeta= 0$.

Also, using \eqref{eq:localparPsiinfty} and Proposition~\ref{Prop:behmodelproblemorigin} with $\bm \Psi=\bm\Psi_\infty$, it follows that
$$
\bm R(\zeta)=\bm\Psi_{\infty,0}(\zeta)(1+\ee^{-\msf h_\infty(\zeta)})^{\sp_3/2}\left(\bm I+\left[\zeta^{\alpha}\msf a(\zeta)\left(\frac{1}{1+\ee^{-\msf h_\infty(\zeta)}}-1\right)-\lambda(\zeta)\right]\bm E_{12}\right) \bm \Psi_{\alpha,0}^\bes(\zeta)^{-1},\quad |\zeta|< \varepsilon.
$$
In virtue of \eqref{eq:lambdabeh}, this shows that 
\begin{equation}\label{eq:behLocalParOrigin}
\bm R(\zeta)=\boh(\zeta^{-1}), \; \zeta\to 0, \quad \text{as well as}\quad \bm R_\pm(\zeta)=\boh(\zeta^{-1}), \; \zeta\to 0 \text{ along }\Gamma,
\end{equation}
and therefore the singularity at $\zeta=0$ is removable. 

In summary, setting
$$
\Gamma_{\bm R}\deff \left(\Gamma\setminus D_\varepsilon\right)\cup \partial D_\varepsilon,
$$
where $\partial D_\varepsilon$ is oriented clockwise, we conclude that $\bm R$ satisfies the following RHP.

\begin{rhp}\label{RHP:RmatrixModelBesselanalysis}
Find a $2\times 2$ matrix-valued function $\bm R:\C\setminus \Gamma_{\bm R}\to \C^{2\times 2}$ with the following properties.
\begin{enumerate}[(i)]
\item $\bm R$ is analytic on $\C\setminus \Gamma_{\bm R}$.
\item The entries of $\bm R_\pm$ are continuous along $\Gamma_\bm R$, and satisfy the jump relation $\bm R_+(\zeta)=\bm R_-(\zeta)\bm J_{\bm R}(\zeta)$ with 
$$
\bm J_{\bm R}(\zeta)\deff
\begin{dcases}
\bm I+\ee^{-\msf h_\infty(\zeta)}\ee^{\pm\pi \ii \alpha} \bm \Psi_{\alpha,-}^\bes(\zeta)\bm E_{21}\bm \Psi_{\alpha,-}^\bes(\zeta)^{-1}, & \zeta\in \Gamma_\pm\setminus \overline{D}_\varepsilon, \\
\bm I-\frac{\ee^{-\msf h_\infty(\zeta)}}{1+\ee^{-\msf h_\infty(\zeta)}}\bm \Psi_{\alpha,-}^\bes(\zeta)\bm E_{11}\bm \Psi_{\alpha,-}^\bes(\zeta)^{-1}
+\ee^{-\msf h_\infty(\zeta)}\bm \Psi_{\alpha,-}^\bes(\zeta)\bm E_{22}\bm \Psi_{\alpha,-}^\bes(\zeta)^{-1}, & \zeta\in \Gamma_0\setminus \overline{D}_\varepsilon, \\
\bm P_0(\zeta)\bm \Psi_\alpha^\bes(\zeta)^{-1}, & \zeta\in \partial D_\varepsilon,
\end{dcases}
$$
and where we orient $\partial D_\varepsilon$ in the clockwise direction.

\item As $\zeta\to \infty$, we have
$$
\bm R(\zeta)=\bm I+\Boh(\zeta^{-1}).
$$
\item The entries of $\bm R_\pm-\bm I$ belong to $L^2_{\mathrm{loc}}(\Gamma,|\dd s|)$.
\end{enumerate}
\end{rhp}

Indeed, the properties (i)--(iii) are immediate from the corresponding properties of the RHPs satisfied by $\bm \Psi_\infty$ and $\bm \Psi^\bes_\alpha$, and property (iv) follows from \eqref{eq:behLocalParOrigin} and the continuity of $\bm P_{0}, \bm \Psi_\alpha^\bes, \bm \Psi_{\alpha,-}^\bes$ away from the origin.

To conclude the asymptotic analysis, we now analyze the behavior of $\bm J_{\bm R}$ as $\sad\to +\infty$. For the sake of clarity, we split this analysis into the next few lemmas. For their statements, we recall that the matrix norm notations were introduced in Section~\ref{sec:notation}.

\begin{lemma}\label{lem:estR1}
The estimate
$$
\|\bm J_{\bm R}-\bm I\|_{L^1\cap L^\infty(\partial D_\varepsilon)}=\Boh\left(\ee^{-\sad}\right)
$$
holds as $\sad\to +\infty$.
\end{lemma}
\begin{proof}
Using the definition of $\bm P_0$ and Proposition~\ref{Prop:behmodelproblemorigin} with $\bm \Psi=\bm \Psi^\bes_\alpha$, we have that for $\zeta\in \partial D_\varepsilon$,
\begin{align}\label{eq:jumpRbessboundary}
\bm J_{\bm R}(\zeta) & =
%
\bm \Psi^\bes_{\alpha,0}(\zeta)(\bm I+\lambda(\zeta)\bm E_{12})\zeta^{\alpha\sp_3/2}(\bm I+\msf a(\zeta)\bm E_{12}) 
\nonumber \\
& ~~~ \times \left[ \bm I-\ee^{-\msf h_\infty(\zeta)}\left(\chi_{\mcal S_+}(\zeta)\ee^{\alpha\pi \ii}-\chi_{\mcal S_-}(\zeta)\ee^{-\alpha\pi\ii}\right)\bm E_{21}\right]
\left(\bm I-\msf a(\zeta)\bm E_{12}\right)\zeta^{-\alpha\sp_3/2}\bm\Psi^\bes_{\alpha,0}(\zeta)^{-1}
\nonumber \\ 
& = 
    \bm I+\bm \Psi^\bes_{\alpha,0}(\zeta)\zeta^{\alpha\sp_3/2}\left[ \lambda(\zeta)\zeta^{-\alpha}\bm E_{12} -\ee^{-\msf h_\infty(\zeta)}\left(\chi_{\mcal S_+}(\zeta)\ee^{\alpha\pi \ii}-\chi_{\mcal S_-}(\zeta)\ee^{-\alpha\pi\ii}\right) \right.
    \nonumber \\
    &~~~ \left.  \times (\bm I+(\msf a(\zeta)+\lambda(\zeta)\zeta^{-\alpha})\bm E_{12})\bm E_{21}\left(\bm I-\msf a(\zeta)\bm E_{12}\right) \right]\zeta^{-\alpha\sp_3/2}\bm \Psi^\bes_{\alpha,0}(\zeta)^{-1}.
\end{align}
It is clear that $\ee^{-\msf h_\infty}=\Boh(\ee^{-\sad})$ uniformly in compacts of $\C$ as $\sad\to +\infty$. Also, $\msf a$ is independent of $\sad$. Next, we already observed the convergence $\lambda= \Boh(\ee^{-\sad})$ as $\sad\to +\infty$, uniformly in compacts (see \eqref{eq:behlambdainfty}). The result then follows once we notice that $\bm\Psi^\bes_{\alpha,0}$ and its inverse are both independent of $\sad$ and continuous functions of $\zeta$, and the contour $\partial D_\varepsilon$ is obviously bounded.
\end{proof}

\begin{lemma}\label{lem:estR2}
The estimate
$$
\|\bm J_{\bm R}-\bm I\|_{L^1\cap L^\infty(\Gamma_+\cup\Gamma_-\setminus \overline D_\varepsilon)}=\Boh\left(\ee^{-\sad}\right)
$$
holds as $\sad\to +\infty$.
\end{lemma}
\begin{proof}
The asymptotic behavior in \eqref{eq:asymptRHPpsimodel} for $\bm \Psi=\bm\Psi^\bes_\alpha$ gives that as $\zeta\to \infty$ along $\Gamma_+\cup\Gamma_-$,
$$
 \bm \Psi_{\alpha,-}^\bes(\zeta)\bm E_{21}\bm \Psi_{\alpha,-}^\bes(\zeta)^{-1}=\Boh(\ee^{-4\zeta^{1/2}}\zeta^{1/2}),
$$
with all the terms in the identity above being independent of $\sad$. Along $\Gamma_-\cup\Gamma_+$, we have $\re \zeta^{1/2}= |\zeta|^{1/2}\cos \frac{\theta_m}{2}>0$, with $\theta_m$ as in \eqref{deff:thetambasiccontours}, and therefore the term on the right-hand side above has finite $L^1$ and $L^\infty$ norms along $\Gamma_+\cup\Gamma_-$. Likewise, the function $\ee^{(-1)^{m+1}\msf u\zeta^m}=\ee^{-\msf h_\infty(\zeta)+\msf s}$ has finite $L^1$ and $L^\infty$ norms along $\Gamma_+\cup\Gamma_-$, which are independent of $\sad$ as well. From the inequality
$$
\|\bm J_{\bm R}-\bm I\|_{L^1\cap L^\infty(\Gamma_+\cup\Gamma_-)}\leq \ee^{-\sad} \|\ee^{-\msf h_\infty+\msf s\pm \pi \ii \alpha}\bm \Psi_{\alpha,-}^\bes\bm E_{21}(\bm \Psi_{\alpha,-}^\bes)^{-1}\|_{L^1\cap L^\infty(\Gamma_+\cup\Gamma_-)},
$$
the lemma then follows.
\end{proof}

\begin{lemma}\label{lem:estR3}
The estimate
$$
\|\bm J_{\bm R}-\bm I\|_{L^1\cap L^\infty(\Gamma_0\setminus \overline D_\varepsilon)}=\Boh\left(\ee^{-\sad}\right)
$$
holds as $\sad\to +\infty$.
\end{lemma}
\begin{proof}
Using again the asymptotic behavior in \eqref{eq:asymptRHPpsimodel} for $\bm \Psi=\bm\Psi^\bes_\alpha$, we obtain that as $\zeta\to \infty$ along $\Gamma_0$, and for $j=1,2$
\begin{equation}\label{eq:esthPsiBessGamma0sinfty}
 \ee^{-\msf h_\infty(\zeta)+\sad}\bm \Psi_{\alpha,-}^\bes(\zeta)\bm E_{jj}\bm \Psi_{\alpha,-}^\bes(\zeta)^{-1}=\Boh(\ee^{-\msf u|\zeta|^m}\zeta^{1/2}),
\end{equation}
with all the terms in the identity above being independent of $\sad$. Using that $\msf h_\infty$ is real along the real axis, we obtain the inequality $|1/(1+\ee^{-\msf h_\infty})|\leq 1$ for $\zeta<0$. Proceeding now as in the proof of Lemma~\ref{lem:estR2}, these estimates are sufficient to conclude the proof.
\end{proof}

We finally conclude
\begin{prop}\label{prop:RLpestimatesBesselasympt}
The estimates
$$
\|\bm R-\bm I\|_{L^\infty(\C\setminus \Gamma_{\bm R})}=\Boh(\ee^{-\sad}),\quad \|\bm R_\pm-\bm I\|_{L^1\cap L^2\cap L^\infty(\Gamma_{\bm R})}=\Boh(\ee^{-\sad})
$$
are valid as $\sad\to +\infty$.
\end{prop}
\begin{proof}
The estimates on $L^1\cap L^\infty$ from Lemmas~\ref{lem:estR1}--\ref{lem:estR3} extend to estimates valid on $L^2\cap L^\infty$ as well, and the current lemma then follows from standard perturbation theory for RHPs.
\end{proof}

For later, we will need some consequences of the asymptotic analysis just carried out. We start stating a rough estimate for $\bm\Psi_\infty$, valid uniformly in $s$.

\begin{prop}\label{prop:estPsiinfty1}
Fix $\sad_0\in \R$. 
There exist constants $M>0$ and $\eta>0$ such that the inequalities
$$
\left|\bm \Psi_{\infty,-}(\zeta)\bm E_{21}\bm \Psi_{\infty,-}(\zeta)^{-1} \right|\leq M \ee^{-\eta |\zeta|^{1/2}},\quad \zeta\in (\Gamma_+\cup\Gamma_-)\setminus D_\varepsilon,
$$
and
$$
\left| \bm \Psi_{\infty,-}(\zeta)\bm E_{kk}\bm \Psi_{\infty,-}(\zeta)^{-1} \right|\leq M |\zeta|^{1/2},\quad \zeta\in \Gamma_0\setminus D_\varepsilon, \quad k=1,2,
$$
are valid for any $\sad\geq \sad_0$.
\end{prop}
\begin{proof}
Since the jumps of $\bm R$ are analytic, standard arguments on the asymptotic analysis of $\bm R$ as $\sad\to +\infty$ ensure that
$$
\bm \Psi_{\infty,-}(\zeta)=(\bm I+\Boh(\ee^{-\sad}))\bm \Psi_{\alpha,-}^\bes(\zeta),
$$
uniformly along $\Gamma \setminus D_\varepsilon$ as $\sad \to +\infty$. The result now follows from the asymptotics \eqref{eq:asymptRHPpsimodel} for $\bm \Psi=\bm \Psi_{\alpha}^\bes(\zeta)$ and the fact that $\bm \Psi_{\alpha,-}^\bes$ is continuous along $\Gamma\setminus D_\varepsilon$ and hence bounded on compact subsets of $\Gamma\setminus D_\varepsilon$.
\end{proof}

Next, we obtain estimates for quantities related to the integrable systems studied in Section~\ref{sec:modelIntSys}. 
It is worth mentioning that the estimate in (ii) below is not optimal, but sufficient for our purposes.

\begin{prop}\label{prop:estimatesBessellimitmodelpphi}
With the functions $\msf p=\msf p(\sad,\xad)$ and $\Phi(\zeta)=\Phi(\zeta\mid \sad,\xad)$ introduced in \eqref{deff:pqr} and \eqref{eq:PhiPhikUpsilonk}, respectively, we have the following estimates. 
\begin{enumerate}
    \item [\rm{(i)}] The estimate
    $$
    \msf p(\sad,\xad)=-\frac{4\alpha^2-1}{8\xad }+\Boh(\ee^{-s})
    $$
    is valid as $\sad\to +\infty$, uniformly for $\xad>0$ in compacts of $(0,+\infty)$.
    \item [\rm{(ii)}]
    There exists $\delta>0$ and $M>0$ such that the estimate
    \begin{multline*}
         |\Phi(\zeta\mid \sad,\xad)|+ |\partial_\xad \Phi(\zeta\mid \sad,\xad)| 
         \\
         \leq M\left(|\zeta|^{1/4}\chi_{(-\infty,-\delta)}(\zeta)+(\delta_0(\alpha)(\log|\zeta|)^2+|\zeta|^{-|\alpha|/2})\chi_{(-\delta,0)}(\zeta)\right),\quad \zeta<0,
    \end{multline*}
    is valid as $\sad\to +\infty$, uniformly for $\xad>0$ in compacts of $(0,+\infty)$.
\end{enumerate}
\end{prop}

\begin{proof}
The RHP \ref{RHP:RmatrixModelBesselanalysis} for $\bm R$ is equivalent to the following integral equation:
%
$$
\bm R(\zeta)=\bm I+\frac{1}{2\pi \ii}\int_{\Gamma_{\bm R}}\bm R_-(s)(\bm J_{\bm R}(s)-\bm I)\frac{\dd s}{s-\zeta},\quad \zeta\in \C\setminus \Gamma_{\bm R}.
$$
From this identity, condition (iii) in RHP~\ref{RHP:RmatrixModelBesselanalysis} may be improved to
$$
\bm R(\zeta)=\bm I+\frac{1}{\zeta}\bm R_{1}+\Boh(\zeta^{-2}),\quad \zeta\to \infty,\quad \text{with}\quad \bm R_{1}\deff-\frac{1}{2\pi \ii}\int_{\Gamma_{\bm R}}\bm R_-(s)(\bm J_{\bm R}(s)-\bm I)\dd s.
$$
Thanks to Proposition~\ref{prop:RLpestimatesBesselasympt}, $\bm R_{1}=\Boh(\ee^{-\sad})$ as $\sad\to +\infty$. On the other hand, from \eqref{deff:Rbesselasymp} and \eqref{eq:asympexpPsibes} we get the relation
\begin{equation}\label{eq:Psiinfty1R1Psibes1}
\bm \Psi_{\infty,1}=\bm R_{1}+\bm\Psi_{\infty,1}^\bes(\alpha).
\end{equation}
Combining with \eqref{deff:pqr} and \eqref{eq:Psiinftyk1}, the result about $\msf p$ follows.

The result about $\Phi$ follows in a similar manner, as we now outline. The claimed estimate for $\zeta<0$ and away from the origin follows from \eqref{eq:asymptPhizetainftyminus}. From relations~\eqref{eq:PsiinftytoPhiinfty}, \eqref{eq:PhiPhikUpsilonk} and \eqref{eq:uxrelation}, it follows that
\begin{align}\label{eq:estPhiPsiBessel1}
\Phi( 4\zeta/\xad^2 ) & =
\left(\frac{\xad}{2}\right)^{1/2}\left[\bm\Psi_\infty(\zeta)\left(\bm I+\chi_{\mcal S_+}(\zeta)\ee^{\pi\ii\alpha}(1+\ee^{-\msf h_\infty(\zeta)})\bm E_{21}\right)  \right]_{11} 
\nonumber
\\ 
& =
\begin{dcases}
\left(\dfrac{\xad}{2}\right)^{1/2}  \left(\bm\Psi_{\infty}\right)_{11}(\zeta), & \zeta>0,\\
\left(\dfrac{\xad}{2}\right)^{1/2}
\left[ \left(\bm\Psi_{\infty}\right)_{11}(\zeta)+\ee^{\pi\ii\alpha}(1+\ee^{-\msf h_\infty(\zeta)})\left(\bm\Psi_{\infty}\right)_{12}(\zeta) \right],&  \zeta<0.
\end{dcases}
\end{align}
For $-\varepsilon<\zeta<0$, we use  \eqref{eq:localparPsiinfty} and \eqref{deff:Rbesselasymp}, together with Proposition~\ref{prop:RLpestimatesBesselasympt}, to express
$$
\bm\Psi_{\infty,+}(\zeta)=\left(\bm I+\Boh(\ee^{-\sad})\right)\bm\Psi_{\alpha,0}^\bes(\zeta)(\bm I+\lambda_+(\zeta)\bm E_{12})\zeta^{\alpha\sp_3/2}_+(\bm I+\msf a_+(\zeta)\bm E_{12})(\bm I-\ee^{\pi \ii\alpha}(1+\ee^{-\msf h_\infty(\zeta)})\bm E_{21}).
$$
The term $\bm\Psi_{\alpha,0}^\bes$ is independent of $\sad$ and analytic in $\zeta$, hence bounded for $\zeta$ in a neighborhood of the origin. From \eqref{eq:hatUpsilonasyfactor}, $\msf a_+(\zeta)=\Boh(1+\delta_0(\alpha)\log|\zeta|)$. From \eqref{eq:lambdabeh}, we also learn that $(\bm I+\lambda_+(\zeta)\bm E_{12})\zeta^{\alpha\sp_3/2}_+=\Boh(|\zeta|^{-|\alpha|/2}+\delta_0(\alpha)\log |\zeta|)$.
Finally, from \eqref{def:hinfty}, we see that $\ee^{-\msf h_\infty(\zeta)}=\Boh(\ee^{-\sad})$ as $\sad\to +\infty$, uniformly for $\zeta<0$. The result now follows from \eqref{eq:estPhiPsiBessel1}.
\end{proof}

\subsection{Asymptotic analysis of $\bm \Psi_\infty$ as $\msf u\to +\infty$}\label{sec:Psiinftylargeuanalysis}\hfill 

We also need to perform the asymptotic analysis of $\bm \Psi_\infty$ as $\msf u\to +\infty$. The latter limit corresponds to $\xad \to 0^+$, and this asymptotic analysis will provide (singular) boundary conditions for the functions $\msf p$ and $\Phi$.

Our starting point in this asymptotic analysis will already be the function $\bm R$ from \eqref{deff:Rbesselasymp}, which is the solution to RHP~\ref{RHP:RmatrixModelBesselanalysis}. The remaining challenge here is the analysis of the jump for $\bm R$ as $\msf u\to+\infty$, as opposed to $\sad\to +\infty$ provided in the previous section. In other words, we need the analogues of Lemmas~\ref{lem:estR1}--\ref{lem:estR3}, which we now provide.

\begin{lemma}\label{lem:R1uanalysis}
    The estimate
    $$
    \|\bm J_{\bm R}-\bm I\|_{\bm L^1\cap L^\infty(\partial D_\varepsilon)}=\Boh(\msf u^{-(1+\alpha)/m}),\quad \msf u \to +\infty,
    $$
    holds uniformly for $\sad\geq \sad_0$ with fixed $\sad_0\in \R$.
\end{lemma}
\begin{proof}
    We analyze $\bm J_{\bm R}$ along $\partial D_\varepsilon$ through its explicit expression from \eqref{eq:jumpRbessboundary}, which for convenience of the reader we recall to be
\begin{multline*}
\bm J_{\bm R}(\zeta) =
    \bm I+\bm \Psi^\bes_{\alpha,0}(\zeta)\zeta^{\alpha\sp_3/2}\left[ \lambda(\zeta)\zeta^{-\alpha}\bm E_{12} -\ee^{-\msf h_\infty(\zeta)}\left(\chi_{\mcal S_+}(\zeta)\ee^{\alpha\pi \ii}-\chi_{\mcal S_-}(\zeta)\ee^{-\alpha\pi\ii}\right) \right. \\
    \left. \times (\bm I+(\msf a(\zeta)+\lambda(\zeta)\zeta^{-\alpha})\bm E_{12})\bm E_{21}\left(\bm I-\msf a(\zeta)\bm E_{12}\right) \right]\zeta^{-\alpha\sp_3/2}\bm \Psi^\bes_{\alpha,0}(\zeta)^{-1}.
\end{multline*}
We next estimate the right-hand side above term by term.
    
    The function $\bm\Psi_{\alpha,0}^\bes$ and its matrix inverse $(\bm\Psi_{\alpha,0}^\bes)^{-1}$ are analytic in $\zeta$ and independent of $\msf u$, hence remain uniformly bounded for $\zeta\in \partial D_\varepsilon$. The inequality \eqref{eq:boundlambdabessanalysis}, together with a straightforward asymptotic analysis of the right-hand side through the classical saddle point method, unravels the rough estimate
    $$
    \lambda(\zeta)=\Boh\left(\msf u^{-(1+\alpha)/m}\right),\quad \msf u\to +\infty,
    $$
    valid uniformly for $\sad\geq \sad_0$, and uniformly for $\zeta$ in compacts of $\C\setminus \Gamma_0$. A standard argument using deformation of contours shows that the same estimate is also valid down to the boundary values $\lambda_\pm(\zeta)$ for $\zeta$ along $\Gamma_0$, and therefore it is valid uniformly for $\zeta\in \partial D_\varepsilon$. The factors $\zeta^{\alpha\sp_3/2}$ and $\msf a$ are continuous functions of $\zeta \in \C\setminus \Gamma_0$, with continuous boundary values along $\Gamma_0$, and independent of $\msf u$. Hence, they remain bounded for $\zeta\in \partial D_\varepsilon$.

    Finally, from the explicit form of $\msf h_\infty$ from \eqref{def:hinfty} and \eqref{eq:ineqzetam}, we see that
    $$
    |\ee^{-\msf h_\infty(\zeta)}|\leq \ee^{-\sad -\msf u|\zeta|^{m}/2},\quad \zeta\in \mcal S_+\cup\mcal S_-.
    $$
    Hence, for some $\eta>0$, the estimate $\ee^{-\msf h(\zeta)}\chi_{\mcal S^\pm}(\zeta)=\Boh(\ee^{-\eta \msf u})$ holds uniformly for $\zeta\in \partial D_\varepsilon$. 

    Plugging in all these estimates into \eqref{eq:jumpRbessboundary}, we obtain the lemma in the $L^\infty$ norm, and the $L^1$ claim follows from it and compactness of $\partial D_\varepsilon$.
\end{proof}

\begin{remark}\label{rk:JRexp}
    The proof of Lemma~\ref{lem:R1uanalysis} actually shows a slightly stronger result, namely that
    $$
    \bm J_{\bm R}(\zeta)=\bm I+\lambda(\zeta)\bm \Psi^\bes_{\alpha,0}(\zeta)\bm E_{12}\bm \Psi^\bes_{\alpha,0}(\zeta)^{-1}+\Boh(\ee^{-\eta \msf u}),\quad \msf u\to \infty,
    $$
    uniformly for $\zeta\in \partial D_\varepsilon$. Performing the change of variables $\msf u^{1/m}s=v$ in \eqref{eq:lambdainfty}, we obtain the asymptotic expansion
    \begin{equation}\label{eq:lambdaexpansionuinfty}
    \lambda(\zeta)\sim \sum_{k=0}^\infty \frac{\lambda_k}{\ii \msf u^{\frac{\alpha+1+k}{m}}\zeta^{k+1}},\quad \lambda_k=\lambda_k(\sad)\deff \frac{(-1)^k}{2\pi }\int_{-\infty}^0 \frac{|v|^{\alpha+k}\ee^{-\sad -(-1)^mv^m}}{1+\ee^{-\sad -(-1)^mv^m}} \dd v, \quad k\in \mathbb{Z}_{\geq 0}.
    \end{equation}
    This, together with the analyticity of $\bm\Psi_{\alpha,0}^\bes$ near the origin, allows us to compute an asymptotic expansion of $\bm J_\bm R$ of the form
    \begin{equation}\label{eq:expansionJRuinfty}
    \bm J_{\bm R}(\zeta)\sim\bm I+\sum_{k=1}^\infty \frac{1}{\msf u^{\frac{\alpha+k}{m}}}\bm J_k(\zeta),\quad \msf u\to +\infty,\qquad \text{with}\quad  \bm J_k(\zeta)\deff \frac{\lambda_{k-1}}{\ii \zeta^k}\bm \Psi^\bes_{\alpha,0}(\zeta)\bm E_{12}\bm \Psi^\bes_{\alpha,0}(\zeta)^{-1},
    \end{equation}
    valid uniformly for $\zeta\in \partial D_\varepsilon$. Notice that $\bm J_k$ has a pole of order at most $k$ at $\zeta=0$.
\end{remark}

\begin{lemma}\label{lem:R2uanalysis}
    Given $\sad_0\in \R$, there exists $\eta>0$ such that the estimate
    $$
    \|\bm J_{\bm R}-\bm I\|_{L^1\cap L^\infty((\Gamma_+\cup\Gamma_-)\setminus \overline D_\varepsilon)}=\Boh(\ee^{-\eta \msf u}),\quad \msf u\to +\infty,
    $$
    is valid uniformly for $\sad\geq \msf s_0$.
\end{lemma}

\begin{proof}
    Proceeding as in the proof of Lemma~\ref{lem:estR2}, we see that
    $$
    \|\bm J_{\bm R}-\bm I\|_{L^1\cup L^\infty((\Gamma_+\cap\Gamma_-)\setminus \overline D_\varepsilon)}\leq \|\bm\Psi_{\alpha,-}^\bes \bm E_{21}(\bm\Psi_{\alpha,-}^\bes)^{-1}\|_{{L^\infty((\Gamma_+\cup\Gamma_-)\setminus \overline D_\varepsilon)}}\|\ee^{-\msf h_\infty}\|_{{L^1\cap L^\infty((\Gamma_+\cup\Gamma_-)\setminus \overline D_\varepsilon)}},
    $$
    and the norm of $\bm\Psi_{\alpha,-}^\bes \bm E_{21}(\bm\Psi_{\alpha,-}^\bes)^{-1}$ above is finite. From \eqref{eq:ineqzetam}, the norms of $\ee^{-\msf h_\infty}$ above are $\Boh(\ee^{-\eta \msf u})$, and the result follows.
\end{proof}

\begin{lemma}\label{lem:R3uanalysis}
    Given $\sad_0\in \R$, there exists $\eta>0$ such that the estimate
    $$
    \|\bm J_{\bm R}-\bm I\|_{L^1\cap L^\infty(\Gamma_0\setminus \overline D_\varepsilon)}=\Boh(\ee^{-\eta \msf u}),\quad \msf u\to +\infty,
    $$
    is valid uniformly for $\sad\geq \msf s_0$.
\end{lemma}

\begin{proof}
    From the very definition of $\bm J_\bm R$ in RHP~\ref{RHP:RmatrixModelBesselanalysis}, it follows that, for $\zeta\in \Gamma_0\setminus \overline D_\varepsilon)$,
    $$
    |\bm J_{\bm R}(\zeta)-\bm I|\leq |\ee^{-\msf h_\infty(\zeta)}|\left(|\bm\Psi_{\alpha,-}^\bes(\zeta)\bm E_{11}\bm\Psi_{\alpha,-}^\bes(\zeta)^{-1}|+|\bm\Psi_{\alpha,-}^\bes(\zeta)\bm E_{22}\bm\Psi_{\alpha,-}^\bes(\zeta)^{-1}|\right).
    $$
    The result now follows from \eqref{eq:esthPsiBessGamma0sinfty}.
\end{proof}

The asymptotic analysis is concluded with the next result.

\begin{prop}\label{prop:estRuinfty}
    The estimates
    $$
    \|\bm R-\bm I\|_{L^\infty(\C\setminus \Gamma_{\bm R})}=\Boh(\msf u^{-(1+\alpha)/m}) \quad \text{and}\quad \|\bm R_\pm-\bm I\|_{L^1\cap L^2\cap L^\infty(\Gamma_{\bm R})}=\Boh(\msf u^{-(1+\alpha)/m})
    $$
    are valid as $\msf u\to +\infty$, uniformly for $\sad\geq \sad_0$ with fixed $\sad_0\in \R$.
\end{prop}
\begin{proof}
    With the help of Lemmas~\ref{lem:R1uanalysis}, \ref{lem:R2uanalysis} and \ref{lem:R3uanalysis}, the conclusion is straightforward from perturbation theory for RHPs, and we skip it.
\end{proof}

We now draw the conclusions from this asymptotic analysis, namely asymptotic formulas for the function $\msf p$ and $\Phi$ from \eqref{deff:pqr} and \eqref{eq:PhiPhikUpsilonk}. For the next statement and its proof, recall that $\xad =2\msf u^{-\frac{1}{2m}}$, see \eqref{eq:uxrelation}.

\begin{prop}\label{prop:pPhibehsmallx}
With $\lambda_0$ as in \eqref{eq:lambdaexpansionuinfty}, the estimate
$$
\msf p(\sad,\xad)=-\frac{4\alpha^2-1}{8\xad}+\frac{2\pi\lambda_0}{2^{2\alpha+1}\Gamma(\alpha+1)^2} \xad^{2\alpha+1}+\Boh(\xad^{2\alpha+3}), \quad \xad \to 0^+,
$$
holds.

In addition, the asymptotic formula
\begin{equation}\label{eq:asymptPhismallx}
\Phi(\zeta\mid \sad,\xad) = \sqrt{\pi}\xad^{1/2} I_\alpha(\xad\zeta^{1/2} ) \left(1 + \Boh(\xad^{2+2\alpha}(1+\xad|\zeta|^{1/2})) \right),\quad \zeta\in \R\setminus \{0\},
\end{equation}
is valid as $\xad\to 0^+$, uniformly for $\zeta \in \R\setminus \{0\}$, and uniformly for $\sad\geq \sad_0$ with fixed $\sad_0\in \R$.
\end{prop}

\begin{remark}
    Even though \eqref{eq:asymptPhismallx} is valid for $\zeta\in \R\setminus \{0\}$, it becomes useful only when $|\zeta|\ll \xad^{-(4\alpha+2)}$, as in this case the $\Boh(\cdot)$ term is truly a small error term. Nevertheless, this is a genuine asymptotic formula, in particular, for $\zeta$ in compacts of $\R$.
\end{remark}

\begin{proof}
    From \eqref{deff:pqr}, \eqref{eq:Psiinfty1R1Psibes1} and \eqref{eq:Psiinftyk1}, it follows that
    \begin{equation}\label{eq:pR1identity}
    \msf p(\sad,\xad)=-\frac{4\alpha^2-1}{8\xad}+\frac{2\ii}{\xad}(\bm R_1)_{12},
    \end{equation}
    where we recall that
    $$
    \bm R_1=-\frac{1}{2\pi \ii}\int_{\Gamma_{\bm R}}\bm R_-(s)(\bm J_{\bm R}(s)-\bm I)\dd s=-\frac{1}{2\pi \ii}\int_{\partial D_\varepsilon}\bm R_-(s)(\bm J_{\bm R}(s)-\bm I)\dd s+\Boh(\ee^{-\eta \msf u}), 
    $$
    and the last equality is a consequence of Proposition~\ref{prop:estRuinfty} and Lemmas~\ref{lem:R2uanalysis} and \ref{lem:R3uanalysis}.
Writing $\bm R_-=(\bm R_--\bm I)+\bm I$, and appealing again to Proposition~\ref{prop:estRuinfty} and Remark~\ref{rk:JRexp}, we get
$$
    \bm R_1=-\frac{\xad^{2\alpha+2}}{2^{2\alpha+3}\pi\ii}\int_{\partial D_\varepsilon}\bm J_1(s)\dd s+\Boh(\xad^{4\alpha+4}),\quad \xad\to 0^+,
    $$
    with $\bm J_1$ as in \eqref{eq:expansionJRuinfty}. This coefficient has a simple pole at $\zeta=0$ and no other poles, so a residue calculation provides
    $$
   \int_{\partial D_\varepsilon}\bm J_1(s)\dd s=-2\pi\lambda_0 \bm\Psi_{\alpha,0}^\bes(0)\bm E_{12}\bm \Psi_{\alpha,0}^\bes(0)^{-1},
    $$
    with $\lambda_0$ given in \eqref{eq:lambdaexpansionuinfty}. We update \eqref{eq:pR1identity} to 
\begin{align*}
    \msf p(\sad,\xad) & =-\frac{4\alpha^2-1}{8\xad}+\frac{ \lambda_0}{2^{2\alpha+1}} \xad^{2\alpha+1}\left[\bm\Psi_{\alpha,0}^\bes(0)\bm E_{12}\bm \Psi_{\alpha,0}^\bes(0)^{-1}\right]_{12}+\Boh(\xad^{4\alpha+3})\\ 
    & = -\frac{4\alpha^2-1}{8\xad}+\frac{\lambda_0}{2^{2\alpha+1}} \xad^{2\alpha+1}\left(\left[\bm\Psi_{\alpha,0}^\bes(0)\right]_{11}\right)^{2}+\Boh(\xad^{4\alpha+3})
    \\
     & = -\frac{4\alpha^2-1}{8\xad}+\frac{2\pi\lambda_0}{2^{2\alpha+1}\Gamma(\alpha+1)^2} \xad^{2\alpha+1}+\Boh(\xad^{4\alpha+3}), \quad \xad \to 0^+,
    \end{align*}
    
where we have made use of \eqref{eq:BesParaZero} and \eqref{eq:1stcolumnPsiBesPhiBes} in the last equality. 

Next, for the asymptotics of $\Phi$, we first use \eqref{eq:estPhiPsiBessel1} and \eqref{deff:Rbesselasymp} to express
$$
\Phi(\zeta)=\frac{\xad^{1/2}}{\sqrt{2}}\times 
\begin{cases}
    \left[ \bm R(\xad^2\zeta/4)\bm\Psi_{\alpha}^\bes(\xad^2\zeta/4) \right]_{11}, & \zeta >\frac{4}{\xad^2}\varepsilon \\ 
    (\xad^2\zeta/4)^{\alpha/2}\left[ \bm R(\xad^2\zeta/4)\bm\Psi_{\alpha,0}^\bes(\xad^2\zeta/4) \right]_{11}, & 0<\zeta< \frac{4}{\xad^2}\varepsilon.
\end{cases}
$$

The term in the second line above depends solely on the first column of $\bm\Psi_{\alpha,0}^\bes$. In view of \eqref{eq:PsibesPsi0bes11}, \eqref{eq:PsibesPsi0bes21} and \eqref{eq:1stcolumnPsiBesPhiBes}, we are then able to re-express this identity in any of the two equivalent and uniform ways,
\begin{equation}
\label{eq:unfoldPhiuconv}
\Phi(\zeta)= \frac{\xad^{1/2}}{\sqrt{2}} (\xad^2\zeta/4)^{\alpha/2}\left[ \bm R(\xad^2\zeta/4)\bm\Psi_{\alpha,0}^\bes(\xad^2\zeta/4) \right]_{11}=\frac{\xad^{1/2}}{\sqrt{2}}\left[ \bm R(\xad^2\zeta/4)\bm\Psi_{\alpha}^\bes(\xad^2\zeta/4) \right]_{11}, \quad \zeta>0. 
\end{equation}

Applying Proposition~\ref{prop:estRuinfty}, we obtain
\begin{align*}
\Phi(\zeta)& =\frac{\xad^{1/2}}{\sqrt{2}}\left[ \left(\bm I+\Boh(\xad^{2+2\alpha})\right)\bm\Psi_\alpha^\bes(\xad^2\zeta/4) \right]_{11} \\
& = \frac{\xad^{1/2} }{\sqrt{2}}\left(\bm\Psi_\alpha^\bes(\xad^2\zeta/4 )\right)_{11}\left(1 + \Boh(\xad^{2+2\alpha})+\frac{[\bm\Psi_\alpha^\bes(\xad^2\zeta/4 )]_{21}}{[\bm\Psi_\alpha^\bes(\xad^2\zeta/4)]_{11}}\Boh(\xad^{2+2\alpha}) \right),
\end{align*}
valid as $\xad\to 0^+$, uniformly for $\zeta>0$. Given a sufficiently large $M>$, we see from \ref{eq:asymptPsiBes1112} that the quotient $(\bm\Psi^\bes_\alpha)_{21}/(\bm\Psi^\bes_\alpha)_{11}$ is $\Boh(\zeta^{1/2}\xad)$ for $\xad^2\zeta\geq M $. On the other hand, this same quotient is an entire function of $\zeta$ (see \eqref{deff:BessParamExplicit}, \eqref{eq:1stcolumnPsiBesPhiBes} and \eqref{eq:identityIalphaFalpha}), and therefore it is $\Boh(1)$ for $0\leq \zeta\xad^2\leq  M $. In summary, we simplify this last asymptotic identity to
\begin{align*}
\Phi(\zeta) = \frac{\xad^{1/2}  }{\sqrt{2}}\left(\bm\Psi_\alpha^\bes(\xad^{2}\zeta/4 )\right)_{11}\left(1 + \Boh(\xad^{2+2\alpha}(1+\xad|\zeta|^{1/2})) \right),\quad \xad \to 0^+.
\end{align*}
The asymptotic claim for $\zeta>0$ now follows from \eqref{eq:1stcolumnPsiBesPhiBes} and \eqref{deff:BessParamExplicit}.

The claim for $\zeta<0$ is very similar: a cumbersome but direct calculation shows that the first identity in \eqref{eq:unfoldPhiuconv} still holds for $\zeta<0$, and the remaining of the arguments remain pretty much the same.
\end{proof}

\subsection{Asymptotic analysis of the model RHP with admissible data}\label{sec:analysisgtinfty}\hfill 

Recall that the matrix-valued function $\bm \Psi=\bm \Psi_\gt$ was introduced in \eqref{deff:Psitauadmissibleh} as the solution to the model problem RHP~\ref{RHP:model}, with admissible data $\msf h=\msf h_\gt$ as introduced in Definition~\ref{def:admissibledata}. As highlighted in \eqref{eq:httohinfty} {\it et seq.}, we expect that $\bm \Psi_\gt\to \bm \Psi_\infty$ as $\gt\to +\infty$, with the latter being the solution to the RHP with data $\msf h_\infty$ given in \eqref{def:hinfty}, and the goal of this section is to show such convergence.

The asymptotic analysis is very similar to the one carried out in the previous section. The first step is the construction of a local parametrix $\bm P_\gt$ in a fixed disk $D_\varepsilon$ near the origin, which is the solution to the following RHP. 

\begin{rhp} \label{rhp:modelformodelanalysis}
Find a $2\times 2$ matrix-valued function $\bm P_\gt:D_\varepsilon\setminus \Gamma\to \C^{2\times 2}$ with the following properties.
\begin{enumerate}[(i)]
\item $\bm P_\gt$ is analytic on $D_\varepsilon\setminus \Gamma$.
\item On $\Gamma\cap D_\varepsilon\setminus \{0\}$, $\bm P_\gt$ satisfies the jump relation
$$
\bm P_{\gt,+}(\zeta)=\bm P_{\gt,-}(\zeta)\bm J_{_\gt}(\zeta),
$$
where we recall that
$$
\bm J_{\gt}(\zeta)=
\begin{dcases}
    \bm I+(1+\ee^{-\msf h_\gt(\zeta)})\ee^{\pm\pi \ii \alpha}\bm E_{21}, & \zeta\in \Gamma_\pm, \\
    \dfrac{1}{1+\ee^{-\msf h_\gt(\zeta)}}\bm E_{12}-(1+\ee^{\msf h_\gt(\zeta)})\bm E_{21}, & \zeta\in \Gamma_0,
\end{dcases}
$$
with $\msf h_\gt$ as in Definition~\ref{def:admissibledata}.
\item As $\gt\to \infty$,
$$
\bm P_\gt(\zeta)=(\bm I+\Boh(\ee^{-\sad}\gt^{-2}))\bm \Psi_\infty(\zeta),
$$
uniformly for $\zeta\in \partial D_\varepsilon$, and also uniformly for $\sad\geq \sad_0$ with any fixed $\sad_0\in \R$.
\end{enumerate}
\end{rhp}

To construct the solution to this RHP, we first introduce the analogue of \eqref{eq:lambdainfty} to our context here, namely the function
\begin{equation}\label{eq:defflambdatau}
\lambda_\gt(\zeta)\deff \frac{1}{2\pi \ii}\int_{-\infty}^0 \left(\frac{1}{1+\ee^{-\msf h_\gt(s)}}-1\right)\frac{|s|^\alpha}{s-\zeta}\dd s, \quad \zeta\in \mathbb{C} \setminus \Gamma_0.
\end{equation}
Using that $\msf h_\gt$ is real-valued along the negative real axis, we obtain the inequality
\begin{equation}\label{eq:ineqfundlambdagt}
\left|\frac{1}{1+\ee^{-\msf h_\gt(\zeta)}}-1\right|\leq \ee^{-\msf h_\gt(\zeta)},\quad \zeta<0.
\end{equation}
Combined with the integrability condition asked in Definition~\ref{def:admissibledata}--(iii), we are ensured that $\lambda_\gt$ is indeed well-defined for $\zeta\in \C\setminus \Gamma_0$. In fact it is the Cauchy transform of a function in $L^1\cap L^\infty(\Gamma_0)$, thus $\lambda_{\gt,\pm}$ belongs to $L^p$ for any $p\in (1,\infty)$. Since $\msf h_\gt$ is analytic on a disk growing with $\gt$, we also see that for any compact $K\subset \Gamma_0$ fixed, there exists $\gt_0=\gt_0(K)$ for which $\lambda_{\gt,\pm}\in L^\infty(K)$ for any $\gt\geq \gt_0$. This latter claim follows from deformation of contours in the definition of $\lambda_\gt$. 

Fix a compact $K\subset \C\setminus \{0\}$. Using the inequality \eqref{eq:ineqfundlambdagt} and again the analyticity of $\msf h_\gt$ in a $\gt$-independent neighborhood of the real axis, standard arguments show that there exists a constant $C_K>0$, depending solely on $K$, for which
$$
|\lambda_\gt(\zeta)|\leq C_K \int_{-\infty}^0 |s|^{\alpha} \ee^{-\msf h_\gt(\zeta)}\dd s,
$$
for every $\zeta\in K$.

Similarly as for $\lambda$, the function $\lambda_\gt$ satisfies the jump condition
\begin{equation}\label{eq:jumplambdat}
\lambda_{\gt,+}(\zeta)-\lambda_{\gt,-}(\zeta)=\left(\frac{1}{1+\ee^{-\msf h_\gt(\zeta)}}-1\right)|\zeta|^\alpha=-\frac{|\zeta|^\alpha \ee^{-\msf h_\gt(\zeta)}}{1+\ee^{-\msf h_\gt(\zeta)}},\quad \zeta<0,
\end{equation}
and the behavior
\begin{equation}\label{eq:lambdataubehorigin}
\lambda_\gt(\zeta)=
\begin{cases}
\Boh(\zeta^\alpha), & -1<\alpha<0, \\
\Boh(\log\zeta), & \alpha =0, \\
\Boh(1), & \alpha>0,
\end{cases}
\end{equation}
as $\zeta\to 0$.

Using the function $\lambda_\gt$, we construct a solution to RHP~\ref{rhp:modelformodelanalysis} in the form
\begin{multline}\label{eq:localparPsit}
\bm P_\gt(\zeta)\deff \bm E_\gt(\zeta)\left(\bm I+\lambda_\gt(\zeta)\bm E_{12}\right)\bm \zeta^{\alpha\sp_3/2}(\bm I+\msf a(\zeta)\bm E_{12})\\
\times \left[\bm I- (\chi_{\mcal S_+}(\zeta)\ee^{\pi \ii \alpha}- \chi_{\mcal S_-}(\zeta)\ee^{-\pi \ii \alpha})(1+\ee^{-\msf h_\gt(\zeta)})\bm E_{21}\right], \quad \zeta\in  D_\varepsilon\setminus \Gamma, 
\end{multline}
where we recall that $\chi_{\mcal S_\pm}$ is the characteristic function of the set $\mcal S_\pm$ previously introduced in \eqref{def:mcalSpm}, the function $\msf a$ is as in \eqref{eq:hatUpsilonasyfactor}, and with the term $\bm \Psi_{\infty,0}=\bm \Psi_0$ being the matrix-valued analytic function obtained when we apply Proposition~\ref{Prop:behmodelproblemorigin} to $\bm\Psi=\bm\Psi_\infty$, we used the matrix prefactor
\begin{equation}\label{eq:Egterror}
\bm E_\gt(\zeta)\deff \bm \Psi_{\infty,0}(\zeta)\left(1+\ee^{-\msf h_\infty(\zeta)}\right)^{\sp_3/2}\left(\bm I-\left(\frac{\ee^{-\msf h_\gt(\zeta)}\msf a(\zeta)\zeta^{\alpha}}{1+\ee^{-\msf h_\gt(\zeta)}}+\lambda_\gt(\zeta)\right)\bm E_{12}\right), \quad \zeta\in D_\varepsilon.
\end{equation}

We now verify that \eqref{eq:localparPsit} is indeed a solution to RHP~\ref{rhp:modelformodelanalysis}.

First of all, the factor $\bm E_\gt$ is analytic near the origin. Indeed, all the terms involved in the definition of $\bm E_\gt$ are analytic except across the negative axis. For $\zeta<0$, we compute
$$
\bm E_{\gt,-}(\zeta)^{-1}\bm E_{\gt,+}(\zeta)=\bm I-\left(\lambda_{\gt,+}(\zeta)-\lambda_{\gt,-}(\zeta)+\frac{\ee^{-\msf h_\gt(\zeta)}}{1+\ee^{-\msf h_\gt(\zeta)}}\left(\msf a_+(\zeta)\ee^{\pi \ii \alpha}-\msf a_-(\zeta)\ee^{-\pi \ii \alpha}\right)|\zeta|^{\alpha}\right)\bm E_{12},
$$
and using \eqref{eq:jumplambdat} and \eqref{eq:jumpafactor} the right-hand side above simplifies to $\bm I$. This shows that $\bm E_\gt$ is analytic across the negative axis as well, so it has an isolated singularity at the origin. From the definition of $\bm E_\gt$,  \eqref{eq:hatUpsilonasyfactor} and \eqref{eq:lambdataubehorigin}, we see that $\bm E_\gt(\zeta)=\boh(\zeta^{-1})$ as $\zeta\to 0$, so this singularity is in fact removable.

We now come back to $\bm P_\gt$. Every factor in the right-hand side of \eqref{eq:localparPsit} is analytic on $\C\setminus \Gamma$, so the same holds for $\bm P_\gt$, that is, RHP~\ref{rhp:modelformodelanalysis}--(i) is satisfied.

All the factors on the first line of the right-hand side of \eqref{eq:localparPsit} are analytic across $\Gamma_\pm$, and the verification of RHP~\ref{rhp:modelformodelanalysis}--(ii) on $\Gamma_\pm$ is immediate. After a cumbersome but straightforward calculation using \eqref{eq:jumplambdat} and \eqref{eq:hatUpsilonasyfactor}, we verify that RHP~\ref{rhp:modelformodelanalysis}--(ii) is also satisfied along $\Gamma_0$.

Finally, we now verify RHP~\ref{rhp:modelformodelanalysis}--(iii). Using the definition of $\bm E_\gt$ and Proposition~\ref{Prop:behmodelproblemorigin} applied to $\bm \Psi_{\infty}$, we express for $\zeta\in \partial D_\varepsilon$, 
\begin{multline}\label{eq:PtPsiinftyboundary}
\bm P_\gt(\zeta)\bm \Psi_{\infty}(\zeta)^{-1}=
\bm \Psi_{\infty,0}(\zeta)\zeta^{\frac{\alpha}{2}\sp_3}\left(\bm I+\msf a(\zeta)\bm E_{12}\right)(1+\ee^{-\msf h_\infty(\zeta)})^{\sp_3/2}\\
\times \left[ \bm I- \left( \frac{1}{1+\ee^{-\msf h_\gt(\zeta)}}-\frac{1}{1+\ee^{-\msf h_\infty(\zeta)}} \right)\msf a(\zeta)\bm E_{12} \right] 
\left( \bm I-\left( \chi_{\mcal S_+}(\zeta)\ee^{\pi \ii\alpha}-\chi_{\mcal S_-}(\zeta)\ee^{-\pi \ii \alpha} \right)(\ee^{-\msf h_\gt(\zeta)}-\ee^{-\msf h_\infty(\zeta)}) \bm E_{21}\right) \\
\times (1+\ee^{-\msf h_\infty(\zeta)})^{-\sp_3/2}(\bm I-\msf a(\zeta)\bm E_{12})\zeta^{-\frac{\alpha}{2}\sp_3}\bm \Psi_{\infty,0}(\zeta)^{-1}.
\end{multline}

The term on the last line of \eqref{eq:PtPsiinftyboundary} is the inverse of the term on the right-hand side of the first line, and these terms are independent of $\gt$. From the explicit expression for $\msf h_\infty$ (recall \eqref{eq:httohinfty}), it follows that we can choose $\varepsilon>0$ sufficiently small and independent of $\sad\in \R$ in such a way that such poles remain within a positive distance from $D_\varepsilon$. As such, the fraction $(1+\ee^{-\msf h_\infty(\zeta)})^{-1}$ is bounded on $D_\varepsilon$ uniformly for $\sad\geq \sad_0$ with any fixed $\sad_0\in \R$. All in all, it means that the terms on the first and last lines of \eqref{eq:PtPsiinftyboundary} are bounded functions of $\zeta \in \partial D_\varepsilon$, uniformly in $\sad\geq \sad_0$. Hence, to conclude RHP~\ref{rhp:modelformodelanalysis}--(iii) it suffices to prove that the terms on the second line of \eqref{eq:PtPsiinftyboundary} are of the form $\bm I+\Boh(\ee^{-\sad}\gt^{-2})$.

From the convergence in \eqref{eq:httohinfty}, we have
\begin{equation}\label{eq:bounddiffehgtinfty}
\ee^{-\msf h_\gt(\zeta)}-\ee^{-\msf h_\infty(\zeta)}= \ee^{-\msf h_\infty(\zeta)}(1-\ee^{\msf h_\infty(\zeta)-\msf h_\gt(\zeta)})=\Boh(\gt^{-2}\ee^{-\sad}),
\end{equation}
as $\gt\to \infty$, uniformly for $\zeta$ in compacts (and in particular for $\zeta \in \partial D_\varepsilon$) and also uniformly for $\sad\geq \sad_0$ with any fixed $\sad_0 \in \mathbb{R}$. This estimate takes care of the second term in the second line of \eqref{eq:PtPsiinftyboundary}. 

We already observed that $(1+\ee^{-\msf h_\infty(\zeta)})^{-1}$ remains uniformly bounded for $\zeta\in \partial D_\varepsilon$, and from the uniform convergence \eqref{eq:httohinfty} we learn that $(1+\ee^{-\msf h_\gt(\zeta)})^{-1}$ remains too uniformly bounded for $\zeta\in D_\varepsilon$ and $\sad\in \R$, and also uniformly as $\gt\to +\infty$. This way, we now bound
\begin{equation}\label{eq:convsigmatausigmainfty}
\left|
 \frac{1}{1+\ee^{-\msf h_\gt(\zeta)}}-\frac{1}{1+\ee^{-\msf h_\infty(\zeta)}}
\right|
=
\frac{|\ee^{-\msf h_\gt(\zeta)}-\ee^{-\msf h_{\infty}(\zeta)}|}{|1+\ee^{-\msf h_\gt(\zeta)}||1+\ee^{-\msf h_\infty(\zeta)}|} =\Boh(\ee^{-\sad }\gt^{-2}),\quad \gt\to \infty,
\end{equation}
uniformly for $\zeta$ in compacts (and again in particular for $\zeta\in \partial D_\varepsilon$), where we used again \eqref{eq:bounddiffehgtinfty}. This is the required estimate for the first term in the second line of \eqref{eq:PtPsiinftyboundary}, and concludes the proof that $\bm P_\gt$ from \eqref{eq:localparPsit} solves RHP~\ref{rhp:modelformodelanalysis}--(iii).

To complete the asymptotic analysis, we use transform
\begin{equation}\label{eq:RPsiadm}
\bm R(\zeta)\deff
\begin{cases}
\bm \Psi_\gt(\zeta)\bm \Psi_\infty(\zeta)^{-1},\quad \zeta\in \C\setminus (\Gamma\cup \overline D_\varepsilon), \\
\bm \Psi_\gt(\zeta)\bm P_\gt(\zeta)^{-1},\quad \zeta\in D_\varepsilon\setminus \Gamma.
\end{cases}
\end{equation}

Recall that
$$
\Gamma_{\bm R} \deff \Gamma\cup \partial D_\varepsilon\setminus D_\varepsilon,
$$
we have that $\bm R$ satisfies the following RHP.
\begin{rhp} \label{rhp:R}
Find a $2\times 2$ matrix-valued function $\bm R:\C\setminus \Gamma_{\bm R}\to \C^{2\times 2}$ with the following properties.
\begin{enumerate}[(i)]
\item $\bm R$ is analytic on $\C\setminus \Gamma_{\bm R}$.
\item The entries of $\bm R_\pm$ are continuous along $\Gamma_\bm R$, and satisfy the jump relation $\bm R_+(\zeta)=\bm R_-(\zeta)\bm J_{\bm R}(\zeta)$ with 
$$
\bm J_{\bm R}(\zeta)\deff
\begin{dcases}
\bm I+\left(\ee^{-\msf h_\gt(\zeta)}-\ee^{-\msf h_\infty(\zeta)}\right)\ee^{\pm\pi \ii \alpha} \bm \Psi_{\infty,-}(\zeta)\bm E_{21}\bm \Psi_{\infty,-}(\zeta)^{-1}, & \zeta\in \Gamma_\pm\setminus \overline{D}_\varepsilon, \\
\bm I+\left(\ee^{-\msf h_\infty(\zeta)}-\ee^{-\msf h_\gt(\zeta)}\right)\bm \Psi_{\infty,-}(\zeta)\left(\frac{\bm E_{11}}{1+\ee^{-\msf h_\gt(\zeta)}}-\frac{\bm E_{22}}{1+\ee^{-\msf h_\infty(\zeta)}}\right)\bm\Psi_{\infty,-}(\zeta)^{-1}, & \zeta\in \Gamma_0\setminus \overline{D}_\varepsilon, \\
\bm P_\gt(\zeta)\bm \Psi_\infty(\zeta)^{-1}, & \zeta\in \partial D_\varepsilon,
\end{dcases}
$$
and where we orient $\partial D_\varepsilon$ in the clockwise direction.

\item As $\zeta\to \infty$, we have 
$$
\bm R(\zeta)=\bm I+\Boh\left(\zeta^{-1}\right).
$$
\end{enumerate}
\end{rhp}

Next, we estimate the jumps of $\bm R$ to conclude the asymptotic analysis.

\begin{lemma}\label{lem:jumpRPsiRHPanalysis}
Fix $\sad_0\in \R$. For any $\kappa\in (0,2)$, the estimate
$$
\| \bm J_{\bm R}-\bm I\|_{L^1\cap L^\infty(\Gamma_{\bm R})}=\Boh(\ee^{-\sad} \gt^{-\kappa})
$$
holds uniformly for $\sad\geq \sad_0$ as $\gt \to \infty$.
\end{lemma}
\begin{proof}
For a given $\kappa\in (0,2)$, set $\nu\deff (2-\kappa)/(m+1)\in (0,2/(m+1))$. We will make the estimates in three separate subsets of $\Gamma_{\bm R}$, namely on $\partial D_\varepsilon$, $(\Gamma_{\bm R}\cap D_{\gt^\nu})\setminus \partial D_\varepsilon$ and $\Gamma_{\bm R}\setminus D_{\gt^\nu}$.

From RHP~\ref{rhp:modelformodelanalysis}--(iii) we immediately obtain
\begin{equation}\label{eq:RestModelNorm1}
\| \bm J_{\bm R}-\bm I\|_{L^1\cap L^\infty(\partial D_\varepsilon)}=\Boh(\ee^{-\sad} \gt^{-2}).
\end{equation}

Next, we move to the required estimates on $\Gamma_{\bm R}\cap D_{\gt^\nu}\setminus \partial D_\varepsilon$. We learn from \eqref{eq:httohinfty} and Definition~\ref{def:admissibledata}--(ii) that there exists $M>0$, independent of $\sad,\gt $, for which
$$
|\msf h_\infty(\zeta)-\msf h_\gt(\zeta)|\leq \frac{M}{\gt^{2-(m+1)\nu}}=\frac{M}{\gt^\kappa},\quad \zeta\in D_{\gt^\nu}.
$$
In particular $\msf h_\gt-\msf h_\infty\to 0$ uniformly for $\zeta\in D_{\gt^\nu}$ as $\gt\to \infty$. From the inequality $|1-\ee^w|\leq |w|\ee^{|w|}$, valid for any $w\in \C$, and the fact that $|\ee^{-\msf h_\infty(\zeta)+\sad}|$ is independent of $\sad$ and bounded by $\ee^{-\eta |\zeta|^m}$ along $\Gamma$ for some $\eta>0$, we learn
\begin{equation}\label{eq:boundhgtpointwise}
|\ee^{-\msf h_\infty(\zeta)}-\ee^{-\msf h_\gt(\zeta)}|\leq |\ee^{-\msf h_\infty(\zeta)}||\msf h_\infty(\zeta)-\msf h_\gt(\zeta)|\ee^{|\msf h_\infty(\zeta)-\msf h_\gt(\zeta)|}
\leq \frac{M\ee^{-\sad}\ee^{-\eta |\zeta|^m}}{\gt^{\kappa}},\quad \zeta\in \Gamma\cap D_{\gt^\nu},
\end{equation}
for a possibly different constant $M>0$, but still independent of $\sad,\gt$. Because $\msf h_\gt,\msf h_\infty$ are real-valued along $(-\infty,0)$, we have $|1+\ee^{-\msf h_\gt}|^{-1},|1+\ee^{-\msf h_\infty}|^{-1}\leq 1$. Then, using \eqref{eq:boundhgtpointwise} and Proposition~\ref{prop:estPsiinfty1}, we conclude from $\bm J_{\bm R}$ in RHP~\ref{rhp:R}--(ii) that
\begin{equation}\label{eq:RestModelNorm2}
\|\bm J_{\bm R}-\bm I\|_{L^{1}\cap L^\infty((\Gamma_{\bm R}\cap D_{\gt^\nu})\setminus \partial D_\varepsilon)}\leq \frac{M\ee^{-\sad}}{\gt^\kappa} \|\zeta^{1/2}\ee^{-\eta |\zeta|^m}\|_{L^{1}\cap L^\infty((\Gamma_{\bm R}\cap D_{\gt^\nu})\setminus \partial D_\varepsilon)}=\Boh(\ee^{-\sad} \gt^{-\kappa}).
\end{equation}

Finally, we move to obtaining the appropriate bounds on $\Gamma_{\bm R}\setminus D_{\gt^\nu}$. The very definition of $\msf h_\infty$ in \eqref{def:hinfty} already shows that for some $\eta_1>0$,
$$
|\ee^{-\msf h_\infty(\zeta)}|\leq \ee^{-\sad} \ee^{-2\eta_1 |\zeta|^m}\leq \ee^{-\sad-\eta_1 \gt^{m\nu} } \ee^{-\eta_1 |\zeta|^m}, \quad \zeta\in \Gamma_{\bm R}\setminus D_{\gt^\nu}.
$$
Next, from Definition~\ref{def:admissibledata}--(iii) we learn that
$$
|\ee^{-\msf h_\gt(\zeta)}|\leq \ee^{-\sad} \ee^{-\Lambda \gt^\nu/2}\ee^{-\Lambda |\zeta|^\nu/2},\quad \zeta\in \Gamma_{\bm R}\setminus D_{\gt^\nu}.
$$
Combining these two inequalities, we obtain the estimate
\begin{equation}\label{eq:boundhgtpointwise02}
|\ee^{-\msf h_\gt(\zeta)}-\ee^{-\msf h_\infty(\zeta)}|\leq M\ee^{-\sad-\eta \gt^{\nu}-\eta|\zeta|^\nu},\quad \zeta\in \Gamma_{\bm R}\setminus D_{\gt^\nu}
\end{equation}
for yet new constants $\eta>0,M>0$, both independent of $\sad$. Therefore, from the definition of $\bm J_\bm R$ and using again Proposition~\ref{prop:estPsiinfty1}, we obtain
\begin{equation}\label{eq:RestModelNorm3}
\|\bm J_{\bm R}-\bm I\|_{L^{1}\cap L^\infty(\Gamma_{\bm R}\setminus D_{\gt^\nu})}\leq M\ee^{-\sad-\eta\gt^\nu} \|\zeta^{1/2}\ee^{-\eta |\zeta|^\nu}\|_{L^{1}\cap L^\infty(\Gamma_{\bm R}\setminus D_{\gt^\nu})}=\Boh(\ee^{-\sad-\eta\gt^\nu}).
\end{equation}

Estimates \eqref{eq:RestModelNorm1}, \eqref{eq:RestModelNorm2} and \eqref{eq:RestModelNorm3} combined conclude the proof.
\end{proof}

\begin{remark}\label{rm:estexphtauhinfty}
    For later reference, observe that \eqref{eq:boundhgtpointwise} and \eqref{eq:boundhgtpointwise02} may be combined into the simpler estimate
    $$
    \ee^{-\msf h_\gt(\zeta\mid \sad)}-\ee^{-\msf h_\infty(\zeta\mid \sad)}=\Boh\left( \frac{\ee^{-\eta |\zeta|^\alpha -\sad}}{\gt^\kappa}\right),\quad \gt\to\infty,
    $$
    for some $\alpha>0$ and $\eta>0$, and valid uniformly for $\zeta\in \R$.
\end{remark}

We conclude the asymptotic analysis with the next statement.
\begin{prop}\label{prop:RestmodelPsiinfty}
Fix $\sad_0\in \R$. For any $\kappa\in (0,2)$ the estimates
$$
\|\bm R-\bm I\|_{L^\infty(\C\setminus \Gamma_{\bm R})}=\Boh(\ee^{-\sad}\gt^{-\kappa}), \quad \|\bm R_\pm-\bm I\|_{L^\infty(\Gamma_{\bm R})}=\Boh(\ee^{-\sad}\gt^{-\kappa}),\quad \|\bm R_\pm-\bm I\|_{L^1\cap L^2(\Gamma_{\bm R})}=\Boh(\ee^{-\sad}\gt^{-\kappa})
$$
are valid as $\gt\to \infty$, uniformly for $\sad\geq \sad_0$.
\end{prop}
\begin{proof}
Once again, the result follows from standard perturbation theory of RHPs and the $L^1\cap L^\infty$ estimates provided by Lemma~\ref{lem:jumpRPsiRHPanalysis}.
\end{proof}

\begin{remark}\label{rmk:expansioninftymodeladm}
    For admissible $\bm\Psi=\bm\Psi_\gt$, the asymptotic expansion in RHP~\ref{RHP:model}--(iii) is valid as $\zeta\to \infty$ in some neighborhood of $\infty$ which in principle may depend on $\gt$. Nevertheless, thanks to Proposition~\ref{prop:RestmodelPsiinfty}, $\bm\Psi_\gt$ may be well-modeled at $\zeta$ by $\bm \Psi_\infty$, and consequently we may choose a uniform (in $\gt$) neighborhood of $\zeta=\infty$ on which RHP~\ref{RHP:model}--(iii) is valid for every $\gt$ sufficiently large.
\end{remark}

Proposition~\ref{prop:RestmodelPsiinfty} finishes the asymptotic analysis of $\bm\Psi_\gt$, and in the next section we draw the main consequences that will be useful for later.

\section{Consequences of the asymptotic analysis for the model problem}\label{sec:conseqmodelproblem}

We now summarize the consequences of the asymptotic analysis $\bm \Psi_\gt\to \bm \Psi_\infty$ provided by Proposition~\ref{prop:RestmodelPsiinfty}, and results surrounding it.

\subsection{Convergence of the kernel}\hfill 

Using Cauchy's integral formula, we write for any $\xi,\zeta\in \C$,
$$
\bm R(\xi)^{-1}\bm R(\zeta)=\bm I+\bm R(\xi)^{-1}\left(\bm R(\zeta)-\bm R(\xi)\right)=\bm I +\bm R(\xi)^{-1}\frac{\zeta-\xi}{2\pi \ii}\int_\gamma \frac{\bm R(u)}{(u-\zeta)(u-\xi)} \dd u, 
$$
where $\gamma$ is any contour encircling both $\zeta$ and $\xi$ in the counterclockwise direction. Applying Proposition~\ref{prop:RestmodelPsiinfty}, we obtain
$$
\bm R(\xi)^{-1}\bm R(\zeta)=\bm I+\Boh(\ee^{-\sad}\gt^{-\kappa}(\zeta-\xi)),\quad \gt\to \infty,
$$
uniformly in $\zeta,\xi$ and $\sad\geq \sad_0$, for any $\kappa\in (0,2)$.

We unfold \eqref{eq:RPsiadm}, concluding that
$$
\bm \Psi_\gt(\xi)^{-1}\bm \Psi_\gt(\zeta)=\bm \Psi_\infty(\xi)^{-1}(\bm I+\Boh(\ee^{-\sad }\gt^{-\kappa}(\xi-\zeta)))\bm \Psi_\infty(\zeta), \quad \gt\to \infty,
$$
also uniformly in $\zeta,\xi$ and $\sad \geq \sad_0$, for any $\kappa\in (0,2)$.
Since $\bm\Psi_\infty$ is bounded on compact subsets of $\C\setminus \{0\}$, and using \eqref{eq:convsigmatausigmainfty}, we conclude that
\begin{multline*}
\frac{1}{\xi-\zeta}\left[\left(\bm I+\frac{\ee^{-\pi\ii \alpha}}{\sigma_\gt(\xi)}\bm E_{21}\right)\bm \Psi_\gt(\xi)^{-1}\bm \Psi_\gt(\zeta)\left(\bm I-\frac{\ee^{-\pi\ii \alpha}}{\sigma_\gt(\zeta)}\bm E_{21}\right)\right]_{21,+}
=\\
\frac{1}{\xi-\zeta}\left[\left(\bm I+\frac{\ee^{-\pi\ii \alpha}}{\sigma_\infty(\xi)}\bm E_{21}\right)\bm \Psi_\infty(\xi)^{-1}\bm \Psi_\infty(\zeta)\left(\bm I-\frac{\ee^{-\pi\ii \alpha}}{\sigma_\infty(\zeta)}\bm E_{21}\right)\right]_{21,-}+\Boh(\ee^{-\sad}\gt^{-\kappa}), \quad \gt\to \infty,
\end{multline*}
valid uniformly in $\zeta,\xi$ in compacts of $(-\infty,0)$ and $\sad\geq \sad_0$, for any $\kappa\in (0,2)$, where recall that $\sigma_{\gt/\infty}(\xi)=(1+\ee^{-\msf{h}_{\gt/\infty}(\zeta)})^{-1}$. A comparison with \eqref{eq:KalphaRHPrepr} using \eqref{eq:sigmaphiinfity} reveals the next result.

\begin{theorem}\label{thm:kernelconvPsigt}
    For any $\kappa\in (0,2)$, the estimate
    $$
    \frac{1}{\zeta-\xi}\left[\left(\bm I+\frac{\ee^{-\pi\ii \alpha}}{\sigma_\gt(\xi)}\bm E_{21}\right)\bm \Psi_\gt(\xi)^{-1}\bm \Psi_\gt(\zeta)\left(\bm I-\frac{\ee^{-\pi\ii \alpha}}{\sigma_\gt(\xi)}\bm E_{21}\right)\right]_{21,-}=\frac{2\pi\ii \ee^{-\pi\ii \alpha}\msf K_\alpha(-\zeta,-\xi)}{\sqrt{\sigma_\infty(\xi)}\sqrt{\sigma_\infty(\zeta)}}+ \Boh(\ee^{-\sad}\gt^{-\kappa})
    $$
    is valid as $\gt\to \infty$, uniformly for $\zeta,\xi$ in compacts of $(-\infty,0)$ and uniformly for $\sad\geq -\sad_0$ with any fixed $\sad_0 \in \mathbb{R}$.
\end{theorem}

\subsection{Asymptotics for a relevant integral}\hfill 

Fix $\zeta_0>0$ and introduce
\begin{equation}\label{deff:Itauint}
\msf I_\gt(\sad)\deff
\int_\sad^\infty
\int_{-\gt^2\zeta_0}^0
\frac{\ee^{-\msf h_\gt(\zeta\mid u)}}{(1+\ee^{-\msf h_\gt(\zeta\mid u)})^2}
\left[\bm\Delta_\zeta\left[\bm\Psi_\gt(\zeta)\left(\bm I-\frac{\ee^{-\pi\ii\alpha}}{\sigma_\gt(\zeta)}\bm E_{21}\right)\right]\right]_{21,-}
\dd \zeta \dd u,
\end{equation}
where we recall that $\bm\Delta_\zeta$ is as explained in \eqref{deff:Deltax}. This double integral will give the leading contribution to the multiplicative statistics \eqref{eq:MultStats}.

If we apply Theorem~\ref{thm:kernelconvPsigt} with $\xi\to \zeta$ (which is possible due to the uniformity of the error term therein) and \eqref{eq:KalphaRHPreprdiag}, we obtain as a consequence
$$
\left[\bm\Delta_\zeta\left[\bm\Psi_\gt(\zeta\mid \sad)\left(\bm I-\frac{\ee^{-\pi\ii\alpha}}{\sigma_\gt(\zeta\mid \sad)}\bm E_{21}\right)\right]\right]_{21,-}=\frac{2\pi\ii \ee^{-\pi\ii\alpha}}{\sigma_\infty(\zeta)}\msf K_\alpha(-\zeta,-\zeta\mid \sad)+\boh(1),\quad \gt\to \infty,
$$
for $\zeta$ in $(-\infty,0)$ pointwise. But we want to have a more precise quantitative error estimate, valid also for $\zeta$ near $0$ to be able to perform the integration in \eqref{deff:Itauint} when $\gt \to \infty$, and for that we split into two cases, namely $\zeta<-\varepsilon$ and $-\varepsilon<\zeta<0$. 

For $\zeta<-\varepsilon$, a combination of \eqref{eq:RPsiadm} with Proposition~\ref{prop:RestmodelPsiinfty} yields
$$
\bm\Psi_{\gt,-}(\zeta\mid \sad)=\left(\bm I+\Boh\left(\frac{\ee^{-\sad}}{\gt^{\kappa}}\right)\right)\bm\Psi_{\infty,-}(\zeta\mid \sad),\quad \gt\to \infty,
$$
uniformly for $\zeta<-\varepsilon$ and $\sad\geq \sad_0$ with any $\sad_0\in \R$, and where $\kappa\in (0,2)$ is arbitrary. Using this estimate in combination with Remark~\ref{rm:estexphtauhinfty}, 
we obtain
\begin{multline*}
\bm\Psi_{\gt,-}(\zeta\mid \sad)\left(\bm I-\ee^{-\pi\ii\alpha}(1+\ee^{-\msf h_\gt(\zeta\mid \sad)})\bm E_{21}\right) = \\
\left(\bm I+\Boh\left(\frac{\ee^{-\sad}}{\gt^{\kappa}}\right)\right)\bm\Psi_{\infty,-}(\zeta\mid \sad)\left(\bm I-\ee^{-\pi\ii\alpha}\left(1+\ee^{-\msf h_\infty(\zeta\mid \sad)} +\Boh\left(\frac{\ee^{-\sad-\eta|\zeta|^{\alpha}}}{\gt^{\kappa}}\right)\right)\bm E_{21}\right),\quad \gt\to \infty,
\end{multline*}
with uniform error term as before. This identity may be differentiated term by term, and after a cumbersome but straightforward calculation it implies 
\begin{multline*}
\bm\Delta_\zeta\left[\bm\Psi_{\gt,-}(\zeta\mid \sad)\left(\bm I-\ee^{-\pi\ii\alpha}(1+\ee^{-\msf h_\gt(\zeta\mid \sad)})\bm E_{21}\right)\right] 
= \\ 
\bm\Delta_\zeta\left[\bm\Psi_{\infty,-}(\zeta\mid \sad)\left(\bm I-\ee^{-\pi\ii\alpha}(1+\ee^{-\msf h_\infty(\zeta\mid \sad)})\bm E_{21}\right)\right] 
+\Boh\left(\frac{\ee^{-\sad}}{\gt^{\kappa}}\right)
\bm E_{21}\left[\bm I+\bm\Delta_\zeta\bm\Psi_{\infty,-}(\zeta)\left(\bm I+\Boh(1)\bm E_{21}\right)\right].
\end{multline*}

From \eqref{eq:asymptRHPpsimodel} we get the rough estimate $\bm\Psi_{\infty,-}(\zeta)=\Boh(\zeta_-^{1/4})$, valid for $\zeta\leq -\varepsilon$. Taking the $(2,1)$-entry and using \eqref{eq:KalphaRHPreprdiag} we update the above to
\begin{multline}\label{eq:Psigtasymplocalout}
\bm\Delta_\zeta\left[\bm\Psi_\gt(\zeta\mid \sad)\left(\bm I-\ee^{-\pi\ii\alpha}(1+\ee^{-\msf h_\gt(\zeta\mid \sad)})\bm E_{21}\right)\right]_{21,-} 
=  \\ 
\frac{2\pi\ii \ee^{-\pi\ii\alpha}}{\sigma_\infty(\zeta)}\msf K_\alpha(-\zeta,-\zeta\mid \sad)
+\Boh\left(\frac{\ee^{-\sad}\zeta_-^{1/4}}{\gt^{\kappa}}\right),\quad \gt\to\infty,
\end{multline}
valid, as before, uniformly for $\zeta\leq -\varepsilon$, $\sad\geq \sad_0$ and for any $\kappa\in (0,2)$.

Next, we consider $-\varepsilon\leq \zeta<0$. In this case, using \eqref{eq:localparPsit}, \eqref{eq:RPsiadm} and Proposition~\ref{prop:RestmodelPsiinfty}, we obtain
\begin{multline}\label{eq:Psigtasymptlocal}
\bm\Psi_{\gt,-}(\zeta)\left(\bm I-\ee^{-\pi\ii\alpha}(1+\ee^{-\msf h_\gt(\zeta\mid \sad)})\bm E_{21}\right)= \\ 
\left(\bm I+\Boh\left(\frac{\ee^{-\sad}}{\gt^{\kappa}}\right)\right) \bm E_\gt(\zeta)\left(\bm I+\lambda_{\gt,-}(\zeta)\bm E_{21}\right)\zeta^{\alpha\sp_3/2}_-\left(\bm I+\msf a_-(\zeta)\bm E_{12}\right),\quad \gt\to \infty,
\end{multline}
valid uniformly for $-\varepsilon\leq \zeta<0$ and $\sad\geq \sad_0$, and for any $\kappa\in (0,2)$. For the record, we recall that $\bm E_\gt$ is as in \eqref{eq:Egterror}, $\lambda_\gt$ is as in \eqref{eq:defflambdatau}, and $\msf a$ is as in \eqref{eq:hatUpsilonasyfactor}.

Using the explicit expression for $\bm E_\gt$ in \eqref{eq:Egterror} and Proposition~\ref{Prop:behmodelproblemorigin} for $\bm\Psi=\bm\Psi_\infty$, we obtain the representation
\begin{equation*}
\bm E_{\gt}(\zeta)\left(\bm I+\lambda_{\gt,-}(\zeta)\bm E_{12}\right)\zeta^{-\alpha\sp_3/2}\left(\bm I+\msf a_-(\zeta)\bm E_{12}\right)=  
\bm \Psi_{\infty,-}(\zeta)\left(\bm I-\ee^{-\pi\ii\alpha}(1+\ee^{-\msf h_\infty(\zeta\mid \sad)})\bm E_{21}\right)
\left(\bm I+(\ast)\bm E_{12}\right),
\end{equation*}
which is valid algebraically (not asymptotically!) for any $\zeta<0$, and where $(\ast)$ represents a scalar that does not affect what comes next. 

Returning this identity in \eqref{eq:Psigtasymptlocal}, we obtain
\begin{multline}\label{eq:Psigtasymptlocal02}
\bm\Psi_{\gt,-}(\zeta)\left(\bm I-\ee^{-\pi\ii\alpha}(1+\ee^{-\msf h_\gt(\zeta\mid \sad)})\bm E_{21}\right)= 
\left(\bm I+\Boh\left(\frac{\ee^{-\sad}}{\gt^\kappa}\right)\right) \\ \times 
\bm \Psi_{\infty,-}(\zeta)\left(\bm I-\ee^{-\pi\ii\alpha}(1+\ee^{-\msf h_\infty(\zeta\mid \sad)})\bm E_{21}\right)
\left(\bm I+(\ast)\bm E_{12}\right)
,\quad \gt\to \infty,
\end{multline}
valid uniformly for $-\varepsilon\leq \zeta<0$ and $\sad\geq \sad_0$ as before.

The matrix term of the form $\bm I+(\ast)\bm E_{21}$ on the right-most side does not affect the final result of applying $\bm\Delta_\zeta$ and taking the $(2,1)$-entry in this identity. Thus, from this identity and again Theorem~\ref{thm:kernelconvPsigt} we obtain
\begin{multline*}
\left[\bm\Delta_\zeta\left[\bm\Psi_{\gt}(\zeta\mid \sad)\left(\bm I-\ee^{-\pi\ii\alpha}(1+\ee^{-\msf h_\gt(\zeta\mid \sad)})\bm E_{21}\right)\right] 
\right]_{21,-}
=\frac{2\pi\ii \ee^{-\pi\ii\alpha}}{\sigma_\infty(\zeta)}\msf K_\alpha(-\zeta,-\zeta\mid \sad)
\\
+\left[
\left(\bm I+\frac{\ee^{-\pi\ii\alpha}}{\sigma_\infty(\zeta)}\bm E_{21}\right)\bm \Psi_{\infty}(\zeta)^{-1}
\Boh\left(\frac{\ee^{-\sad}}{\gt^\kappa}\right)
\bm \Psi_{\infty}(\zeta)\left(\bm I-\frac{\ee^{-\pi\ii\alpha}}{\sigma_\infty(\zeta)} \bm E_{21}\right)
\right]_{21,-},\quad \gt\to \infty.
\end{multline*}

Again thanks to Proposition~\ref{Prop:behmodelproblemorigin}, we see that
\begin{align*}
\bm \Psi_{\infty}(\zeta)\left(\bm I-\frac{\ee^{-\pi\ii\alpha}}{\sigma_\infty(\zeta)} \bm E_{21}\right) & =
\bm\Psi_{\infty,0}(\zeta)\sigma_\infty(\zeta)^{-\sp_3/2}\zeta^{\alpha\sp_3/2}\left(\bm I+\msf a(\zeta)\sigma_\infty(\zeta)\bm E_{12}\right) \\
& =
\Boh\left(1\right)\zeta^{\alpha\sp_3/2}\left(\bm I+\msf a(\zeta)\sigma_\infty(\zeta)\bm E_{12}\right).
\end{align*}

Using this estimate, the previous estimate simplifies to
\begin{equation}\label{eq:Psigtasymplocalin}
\left[\bm\Delta_\zeta\left[\bm\Psi_{\gt}(\zeta\mid \sad)\left(\bm I-\ee^{-\pi\ii\alpha}(1+\ee^{-\msf h_\gt(\zeta\mid \sad)})\bm E_{21}\right)\right] 
\right]_{21,-}
=\frac{2\pi\ii \ee^{-\pi\ii\alpha}}{\sigma_\infty(\zeta)}\msf K_\alpha(-\zeta,-\zeta\mid \sad)+\Boh\left(\frac{\ee^{-\sad} \zeta^\alpha}{\gt^{\kappa}}\right),
\end{equation}
valid as $\gt\to\infty$, uniformly for $-\varepsilon\leq \zeta<0$, and uniformly for $\sad\geq \sad_0$ with any $\sad_0\in \R$, and where $\kappa\in (0,2)$ is arbitrary. 

Finally, we now insert \eqref{eq:Psigtasymplocalout} and \eqref{eq:Psigtasymplocalin} in \eqref{deff:Itauint} and use also the uniform decay of $\msf h_\gt$ as $\zeta\to -\infty$ which is ensured by Definition~\ref{def:admissibledata}--(iii), the convergence from Remark~\ref{rmk:expansioninftymodeladm}, 
and arrive at the following result.

\begin{theorem}\label{thm:Itayasympt}
    For any $\kappa\in (0,2)$ and any $\sad_0\in \R$, the estimate
    $$
    \msf I_\gt(s)=2\pi\ii \ee^{-\pi\ii\alpha}\int_\sad^\infty \int_{-\infty}^0 \msf K_\alpha(-\zeta,-\zeta\mid u) \frac{\ee^{-\msf h_\infty(\zeta\mid u)}}{1+\ee^{-\msf h_\infty(\zeta\mid u)}}\dd \zeta\dd u+\Boh\left(\frac{\ee^{-\sad}}{\gt^{\kappa}}\right),\quad \gt\to\infty,
    $$
    holds uniformly for $\sad\geq \sad_0$.
\end{theorem}

This result finishes the asymptotic analysis necessary on the model problem.

\subsection{Proof of main theorems on integrable equations}\hfill

In this section we complete the proof of our main results that connect the limiting probabilistic quantities to integrable equations. They are essentially a recollection of several calculations that we already performed.

\begin{proof}[Proof of Theorem~\ref{thm:intsysPhi}]

Define $\Phi$ as in \eqref{eq:PhiPhikUpsilonk}. The nonlocal equation \eqref{eq:nonlocalPDEintro} for it was already obtained in Theorem~\ref{thm:Phinonlocalcore}. The asymptotic behavior of $\Phi$ as $\zeta\to \pm\infty$ are given in \eqref{eq:asymptPhizetainftyplus} and \eqref{eq:asymptPhizetainftyminus}. The asymptotic behavior as $\xad\to 0^+$, in turn, was obtained in Proposition~\ref{prop:pPhibehsmallx}.
\end{proof}

\begin{proof}[Proof of Theorem~\ref{thm:TWtyperepres}]

Define $\msf p$ as in \eqref{deff:pqr}. Then, identity \eqref{deff:pcoreint} shows that this definition coincides with the one given in \eqref{eq:deffpintro}, and Proposition~\ref{prop:pPhibehsmallx} yields its asymptotic behavior as $\xad\to 0^+$ as claimed by Theorem~\ref{thm:TWtyperepres}.

It remains to show the connection with $\msf L_\alpha^\bes$. For that, we start from \eqref{eq:KalphaRHPreprdiag} and use \eqref{eq:PsiinftytoPhiinfty} and \eqref{eq:UpsilonPhi} to express
$$
\msf K_\alpha(-\zeta,-\zeta\mid \sad,\xad)=\frac{\ee^{\pi\ii\alpha} \uad^{-1/m}\sigma_\infty(\zeta)}{2\pi\ii}\left[\bm\Upsilon(\uad^{1/m}\zeta)^{-1}\bm\Upsilon'(\uad^{1/m}\zeta) \right]_{21,-}.
$$
Performing the change of variables $\xi=-\uad^{1/m}\zeta=-\frac{4}{\xad^2}\zeta$ in \eqref{deff:Halpha}, we obtain from \eqref{eq:sigmaphiinfity} that
\begin{equation}\label{eq:LBessdoubleint}
-\log\msf L_\alpha^\bes(\sad,\xad)=
\frac{\ee^{\pi\ii\alpha}}{2\pi\ii}\int_\sad^\infty \int_{-\infty}^0\left[\bm\Upsilon(\xi\mid \sad=u,\xad)^{-1}\bm\Upsilon'(\xi\mid \sad=u,\xad)\right]_{21,-} \partial_\sad \sigma_{\bm\Phi}(\xi\mid \sad=u) \dd\xi\dd u.
\end{equation}
Observe that $\sigma_{\bm\Phi}$ is independent of $\xad$, recall \eqref{deff:sigmaPhi}. A direct calculation from the jump of $\bm\Upsilon$ (which coincides with the jump of $\bm\Phi_\infty$ in \eqref{eq:jumpPhiInfty}) shows that $[\bm\Upsilon^{-1}\bm\Upsilon']_{21,-}=\ee^{-2\pi\ii\alpha}[\bm\Upsilon^{-1}\bm\Upsilon']_{21,+}$. Using then \eqref{eq:Upsilonrowrelation}, \eqref{eq:PhiPhikUpsilonk} and \eqref{eq:jumpphiffction}, we get 
$$
\left[\bm\Upsilon^{-1}\bm\Upsilon'\right]_{21,-}=\ii \ee^{-2\pi\ii\alpha}\left[\Phi\partial_\xad\Phi'-\Phi'\partial_\xad\Phi\right].
$$
Taking one $\xad$-derivative and using \eqref{eq:PhikpartialPhik}, we get that the $\xad$-derivative of the right-hand side of the above formula is given by $\ii\ee^{-2\pi\ii\alpha}\Phi^2$.
Using this relation in \eqref{eq:LBessdoubleint}, we see that
$$
-\partial_\xad\left(\log\msf L_\alpha^\bes(\sad,\xad)\right)=\frac{\ee^{-\pi\ii\alpha}}{2\pi}\int_\sad^\infty \int_{-\infty}^0\Phi(\xi\mid \sad=u,\xad)^2 \partial_\sad \sigma_{\bm\Phi}(\xi\mid \sad=u) \dd\xi\dd u.
$$
When obtaining this identity, we exchanged the derivatives with the integrals, a calculation which is justified by the bounds obtained in \eqref{prop:estimatesBessellimitmodelpphi}--(ii) and the exponential decay of $\sigma_{\bm\Phi}$. From the definition of $\msf p$ in \eqref{eq:deffpintro} it follows that this identity is equivalent to \eqref{eq:TWrepr}, concluding the proof of Theorem~\ref{thm:TWtyperepres}.
\end{proof}

\begin{proof}[Proof of Theorem~\ref{thm:m1PDE}]

Suppose now that $m=1$. In particular, from \eqref{deff:sigmaPhi} we see
that 
$
\sigma_{\bm \Phi}(\zeta)=1/(1+\ee^{-\sad + \zeta})
$.
By setting 
\begin{equation}
    \widehat{\bm \Phi}_\infty(\zeta)\deff \begin{pmatrix}
        1 & 0
        \\
        -\ii \frac{\sad \xad}{2} & 1 
    \end{pmatrix}{\bm \Phi}_\infty(\zeta+\sad),
\end{equation}
it then follows from RHP \ref{rhp:Phiinfty}, \eqref{deff:pqr} and direct calculations that $\widehat{\bm \Phi}_\infty$ satisfies the following RHP.
\begin{rhp}\label{rhp:hatPhiinfty}
Find a $2\times 2$ matrix-valued function $ \widehat{\bm \Phi}_\infty:\C\setminus (-\infty,-\sad]\to \C^{2\times 2}$ with the following properties.
\begin{enumerate}[(i)]
\item $\widehat{\bm \Phi}_\infty$ is analytic on $\C\setminus (-\infty,-\sad]$.
\item The entries of $\widehat{\bm \Phi}_{\infty,\pm}$ are continuous along $(-\infty,-\sad)$ and satisfy
\begin{equation}\label{eq:jumphatPhiInfty}
\widehat{\bm \Phi}_{\infty,+}(\zeta)=\widehat{\bm \Phi}_{\infty,-}(\zeta)
\begin{pmatrix}
    \ee^{\pi\ii\alpha} &  \frac{1}{1+\ee^{\zeta}}
    \\
   0 &\ee^{-\pi\ii\alpha}
\end{pmatrix},
\qquad \zeta<0.
\end{equation}
\item As $\zeta\to \infty$, we have 
\begin{equation}\label{eq:hatPhiinftyinfty}
    \widehat{\bm \Phi}_\infty(\zeta)=\left( \bm I+\frac{\widehat{\bm \Phi}_{\infty,1}}{\zeta}+\Boh(\zeta^{-2}) \right)\zeta^{-\sp_3/4}\bm U_0\ee^{\xad\zeta^{1/2}\sp_3}\left( \bm I+(\ee^{\pi \ii \alpha}\chi_{\mcal S_+}(\zeta)-\ee^{-\pi \ii \alpha}\chi_{\mcal S_-}(\zeta))\bm E_{21} \right),
\end{equation}
where
\begin{equation}\label{eq:hatphiinfty1}
\widehat{\bm \Phi}_{\infty,1}=
\begin{pmatrix}
        1 & 0
        \\
        -\ii \frac{\sad \xad}{2} & 1 
    \end{pmatrix}
    \begin{pmatrix}
    \msf q-\frac{\sad}{4} & -\ii \msf p \\ \ii \msf r & -\msf q+\frac{\sad}{4}
\end{pmatrix}
    \begin{pmatrix}
        1 & 0
        \\
        \ii \frac{\sad \xad}{2} & 1 
    \end{pmatrix}+
    \begin{pmatrix}
        \frac{\sad^2\xad^2}{8} & -\ii \frac{\sad \xad}{2}
        \\
        -\frac{\ii}{24}\sad^2\xad(\sad\xad^2+3) & -\frac{\sad^2\xad^2}{8}
    \end{pmatrix}.
\end{equation}

\item As $\zeta \to -\sad$, we have 
\begin{equation}\label{eq:hatPhiinfty0}
    \widehat{\bm \Phi}_{\infty}(\zeta)=\widehat{\bm \Phi}_{\infty,0}(\zeta)(\zeta+\sad)^{\alpha\sp_3/2}\left(\bm I+\msf a(\zeta+\sad)\bm E_{12}\right)(1+\ee^{\zeta})^{\sp_3/2},
\end{equation}
where
\begin{equation}\label{eq:hatphiinfty0}
    \widehat{\bm \Phi}_{\infty,0}(\zeta)=\begin{pmatrix}
        1 & 0
        \\
        -\ii \frac{\sad \xad}{2} & 1 
    \end{pmatrix}\bm \Phi_{\infty,0}(\zeta+\sad)
\end{equation}
is analytic near $\zeta=-\sad$ and 
%
$\msf a$ is as in \eqref{eq:hatUpsilonasyfactor}.
\end{enumerate}
\end{rhp}

Note that $(\widehat{\bm \Phi}_{\infty,1})_{12}=-\ii(\msf p+\frac{\sad\xad}{2})$, we define, similar to \eqref{eq:UpsilonPhi},
\begin{equation}\label{eq:hatUpsilonPhi}
\widehat{\bm \Upsilon}(\zeta) \deff \left(\bm I+\ii \left(\msf p +\frac{\sad\xad}{2}\right)\bm E_{21}\right)\widehat{\bm \Phi}_\infty(\zeta).
\end{equation}
Thus,  $\widehat{\bm \Upsilon}$ satisfies the ODE
\begin{equation}\label{eq:LaxeqtionhatUpsilon}
\partial_\xad \widehat{\bm \Upsilon}(\zeta)=
\widehat{\bm B}(\zeta)
\widehat{\bm\Upsilon}(\zeta),
\end{equation}
where
\begin{equation}\label{def:hatB}
    \widehat{\bm B}(\zeta)=
\begin{pmatrix}
0 & -\ii
\\
\ii(\zeta+2\partial_\xad\msf p+\sad) & 0
\end{pmatrix}.
\end{equation}

Since the jump matrix of $\widehat{\bm \Phi}_{\infty}$ in \eqref{eq:jumphatPhiInfty} is independent of $\sad$ as well, we obtain from \eqref{eq:hatphiinfty0} that 
\begin{equation}\label{eq:LaxeqtionhatUpsilons}
\partial_\sad \widehat{\bm \Upsilon}(\zeta)=
\widehat{\bm A}(\zeta)
\widehat{\bm\Upsilon}(\zeta),
\end{equation}
where
\begin{equation}\label{def:hatA}
\widehat{\bm A}(\zeta)=\begin{pmatrix}
    0 & 0
    \\
    \ii(\partial_\sad\msf p+\frac{\xad}{2}) & 0
\end{pmatrix}
+
\frac{1}{\zeta+\sad}\begin{pmatrix}
    \aad & \bad
    \\
    \cad & -\aad
\end{pmatrix}
\end{equation}
for some functions $\aad$, $\bad$ and $\cad$ depending on $\xad$ and $\sad$.  The compatibility condition $\partial^2 _{\sad  \xad} \widehat{\bm \Upsilon}=\partial ^2_{\xad  \sad} \widehat{\bm \Upsilon} $ yields the zero curvature equation
$$
\partial_\sad \widehat{\bm B}-\partial_\xad \widehat{\bm A}=\widehat{\bm A}\widehat{\bm B}-\widehat{\bm B}\widehat{\bm A}.
$$
Inserting \eqref{def:hatB} and \eqref{def:hatA} into the above equation and comparing the coefficients of $\Boh(1)$ and $\Boh(1/(\zeta+s))$ terms, we arrive at
\begin{equation}\label{eq:abc}
    \bad=-\ii \left(\partial_{\sad}\msf p+\frac{\xad}{2}\right), \quad \aad=-\frac12\left(\partial^2_{\sad\xad}\msf p+\frac 12\right), \quad
    \cad=\ii \partial_\xad\aad-2\bad\partial_{\xad}\msf p,
\end{equation}
and 
\begin{equation}\label{eq:cderivative}
\partial_{\xad}\cad=4\ii\aad\partial_{\xad}\msf p. 
\end{equation}
From \eqref{eq:abc}, it is readily seen that
$$
\cad=\ii\left(
-\frac12\partial^3_{\sad\xad\xad}\msf p+2\partial_{\xad}\msf p\partial_{\sad}\msf p+\xad \partial_{\xad}\msf p
\right).
$$
This, together with \eqref{eq:cderivative}, implies the PDE for $\msf p$ claimed in Theorem~\ref{thm:m1PDE}.
\end{proof}

\section{Asymptotic analysis of the RHP for OPs}\label{sec:RHPOPsanalysis}

Our effort so far was the in-depth study of the model problem introduced in Section~\ref{sec:modelproblem}. Now we come back to our original model, that is, extracting large $n$ information on the point process with joint distribution \eqref{eq:Lagdef}. In this section, we carry out the asymptotic analysis of the associated polynomials, and in this analysis the solution $\bm \Psi_\gt$ of the model problem with admissible data, as well as its limit $\bm\Psi_\infty$, will play a central role.

\subsection{Equilibrium measures and related quantities}\label{sec:eqmeasure}\hfill

In the course of the asymptotic analysis, we will need certain properties of equilibrium measures that we now briefly review. The properties we now mention can be found, e.g., in \cite{Saff_book, deift_kriecherbauer_mclaughlin}, see also \cite{AlvesSilva2024} for the specific structure of Cauchy transforms of log equilibrium measures with hard edge.

The {\it equilibrium measure} of the interval $[0,+\infty)$ in the external field $V$ is the unique probability measure supported on $[0,+\infty)$ that minimizes the functional
$$
\mu\mapsto \iint \log\frac{1}{|x-y|}\dd\mu(x)\dd\mu(y)+\int V(x)\dd\mu(x)
$$
over all Borel probability measures $\mu$ supported on $[0,+\infty)$. This measure is uniquely characterized by the {\it Euler-Lagrange variational conditions:} there exists a constant $\ell_V\in \R$ for which
\begin{equation}\label{eq:EulerLagrange}
\int \log\frac{1}{|x-y|}\dd\mu_V(y)+\frac{1}{2}V(x)+\ell_V
\begin{cases}
    =0, & x\in \supp\mu_V, \\ 
    \geq 0, & x\in [0,+\infty).
\end{cases}
\end{equation}

Under our assumptions, the equilibrium measure $\mu_V$ is absolutely continuous with respect to the Lebesgue measure, say $\dd\mu_V(x)=\mu_V'(x)\dd x$, for some function $\mu'_V$. The regularity conditions on $\mu_V$ that were mentioned in Assumption~\ref{deff:condQ} translate into the following precise conditions.

\begin{assumption}[One-cut regularity condition]\label{assump:onecutregular} We assume that the potential $V$ is such that the following holds.
\begin{enumerate}[(i)]
    \item For some $a>0$, $\supp\mu_V=[0,a]$.
    \item The inequality in \eqref{eq:EulerLagrange} is strict, that is, 
    \begin{equation}\label{eq:EulerLagrangestrict}
    \int \log\frac{1}{|x-y|}\dd\mu_V(y)+\frac{1}{2}V(x)+\ell_V>0,\quad x\in (a,\infty).
    \end{equation}
    \item $\mu'_V(x)>0$ for $x\in (0,a)$.
    \item For some constants $\kappa_0,\kappa_a>0$, we have
$$
\mu'_V(x)=\frac{\kappa_0}{\sqrt{x}}(1+\Boh(x)),\; x\searrow 0,\qquad \text{and}\qquad \mu'_V(x)=\kappa_a\sqrt{a-x}(1+\Boh(x-a)),\; x \nearrow a.
$$
\end{enumerate}    
\end{assumption}

For the next transformation, we need to consider quantities related to the equilibrium measure. For
$$
C^{\mu_V}(z)\deff \int \frac{\dd \mu_V(s)}{s-z},\quad z\in \C\setminus [0,a],
$$
being the Cauchy transform of the equilibrium measure, we associate to it its $\phi$ function,
\begin{equation}\label{deff:phifction}
\phi(z)\deff \int_0^z\left(C^{\mu_V}(s)+\frac{1}{2}V'(s)\right)\dd s,\quad z\in \C\setminus (-\infty,a],
\end{equation}
with the path of integration starting on the upper half-plane. Under Assumption~\ref{assump:onecutregular}, the Cauchy transform satisfies an algebraic equation of the form
$$
\left(C^{\mu_V}(z)+\frac{V'(z)}{2}\right)^2=\frac{z-a}{z}h_V(z)^2,
$$
for some polynomial $h_V$ with $h_V(0),h_V(a)\neq 0$. Standard arguments show that
\begin{equation}\label{eq:jumpsphi}
\phi_+(z)+\phi_-(z)=0, \; 0<z<a, \quad \phi_+(z)-\phi_-(z)=2\pi \ii \mu_V([0,z])\in \ii \R_+, \quad 0<z<a.
\end{equation}
In particular, $h_V(x)<0$ on $(0,a)$, and 
$$
\mu'_V(x)=\frac{1}{\pi}\sqrt{\frac{a-x}{x}}|h_V(x)|, \quad 0<x<a.
$$
Furthermore, a combination of \eqref{eq:jumpsphi}, Cauchy-Riemann equations and the inequality \eqref{eq:EulerLagrangestrict} yields that
\begin{equation}\label{eq:ineqphi}
\phi(x)>0 \; \text{for }x>a,\qquad \text{and}\qquad \re \phi(z)<0 \text{ for }z \text{ slightly above or below }(0,a).
\end{equation}

A local analysis based on conformal properties of $\phi$ shows that for some $\delta>0$,
$$
\re \phi(z)<0 \text{ for } z\in D_\delta\setminus [0,\delta)\qquad \text{and}\qquad \phi(z)<0 \text{ for } -\delta<z<0.
$$
Consequently, the function
$$
\psi(z)\deff \frac{1}{4}\phi(z)^2
$$
is a conformal map in a neighborhood of the origin, and the behavior near the origin in Assumption~\ref{assump:onecutregular}--(iv) ensures that
\begin{equation}\label{eq:conformalmap}
\psi(z)=-\msf c_V z(1+\Boh(z)),\quad z\to 0,\quad \text{for }\msf c_V\deff \pi^2\kappa_0^2,
\end{equation}
and where we recall that $\kappa_0>0$ is as in Assumption~\ref{assump:onecutregular}--(iv). This definition of $\msf c_V$ is consistent with \eqref{eq:msfcVintro}.

Finally, a direct calculation shows that, as $z\to \infty$, 
\begin{equation}\label{eq:asymptphi}
\phi(z)=\frac{V(z)}{2}+\ell_V-\log z+\frac{\phi_\infty}{z}+\Boh(z^{-2}),
\end{equation}
for some $\phi_\infty\in \R$, where $\ell_V$ is as in Assumption~\ref{assump:onecutregular}--(ii).

\subsection{The RHP for OPs}\hfill

We start with the RHP for the corresponding OPs.
\begin{rhp}\label{rhp:YOPS}
Find a $2\times 2$ matrix-valued function $\bm Y:\C\setminus [0,+\infty)\to \C^{2\times 2}$ with the following properties.
\begin{enumerate}[\rm (i)]
\item $\bm Y:\C\setminus [0,\infty)\to \C^{2\times 2}$ is analytic.
\item The matrix $\bm Y$ has continuous boundary values $\bm Y_\pm$ along $(0,+\infty)$, and they are related by $\bm Y_+(x)=\bm Y_-(x)\bm J_{\bm Y}(x)$, $x>0$, with
$$
\bm J_{\bm Y}(x)\deff \bm I+\omega_n(x)\bm E_{12},\quad x>0.
$$
\item As $z\to \infty$,
$$
\bm Y(z)=\left(\bm I+\frac{1}{z}\bm Y_1+\Boh(z^{-2})\right)z^{n\sp_3}.
$$
\item As $z\to 0$,
$$
\bm Y(z)=
\Boh\begin{pmatrix}
1 & h_\alpha(z) \\ 1 & h_\alpha(z)
\end{pmatrix}, 
$$
where the $\mathcal O$-term is to be read entry-wise and here and from now on we set
\begin{equation}\label{deff:hfctionlocalbeh}
h_\alpha(z) \deff
\begin{cases}
z^\alpha, & \alpha<0,\\
\log z, & \alpha=0,\\
1, & \alpha>0.
\end{cases}
\end{equation}
\end{enumerate}
\end{rhp}

The weight $\omega_n(x)=\omega_n(x\mid \sad)$ was introduced explicitly in \eqref{eq:deffweight} and depends on $\sad$, and consequently so does $\bm Y(z)=\bm Y(z\mid \sad)$ and any other quantity associated to the orthogonal polynomials, but sometimes we omit this dependence from our notation.

As it is now classical (see \cite{Fokas1992}), the solution to this RHP is related to the monic OP $\msf P_j$ of degree $j$ for the weight $\omega_n$ via
$$
\bm Y_{11}(z)=\msf P_n(z)\quad \text{and}\quad \bm Y_{21}(z)=-2\pi\ii (\gamma_{n-1}^{(n)})^2\msf P_{n-1}(z).
$$
As a consequence of the Christoffel-Darboux formula,
\begin{equation}\label{eq:relKnRHPY}
\msf K_n(x,y)=\frac{1}{2\pi \ii(x-y)} \bm e_2^T\bm Y_+(y)^{-1}\bm Y_+(x)\bm e_1, \; x\neq y,\quad \text{and}\quad \msf K_n(x,x)=\frac{1}{2\pi \ii}\bm e_2^T \bm Y_+(x)^{-1}\bm Y_+'(x)\bm e_1, 
\end{equation}
where $\bm e_1\deff (1,0)^T$ and $\bm e_2\deff (0,1)^T$ are the canonical base vectors for $\R^2$. Using the explicit form
$$
\partial_u \omega_n(x\mid u)=
\ee^{-u-n^{2m}Q(x)}\sigma_n(x\mid u)\omega_n(x\mid u)=\omega_n(x\mid u)\partial_u \log\sigma_n(x\mid u),
$$
the identity given by \eqref{eq:deffformula} updates to
\begin{align}
\log \msf L_n^Q(\sad) 
 & =-\frac{1}{2\pi\ii}
\int_{\sad}^{\infty} \int_0^\infty \left[ \bm Y(x\mid u)^{-1}\bm Y'(x\mid u)\right]_{21,+}x^\alpha \ee^{-nV(x)}\partial_u \sigma_n(x\mid u) \dd x \dd u \nonumber \\
& =-\frac{1}{2\pi\ii}
\int_{\sad}^{\infty} \int_0^\infty \left[ \bm Y(x\mid u)^{-1}\bm Y'(x\mid u)\right]_{21,+}\omega_n(x\mid u)\partial_u \log \sigma_n(x\mid u) \dd x \dd u.
\label{eq:defformulamain}
\end{align}
This deformation formula will be the starting point for our asymptotic analysis of $\msf L_n^Q(\sad)$, but first we need to obtain asymptotics for $\bm Y$.

To obtain asymptotics for $\bm Y$ we once again employ the Deift-Zhou nonlinear steepest descent method. The transformations needed in this case are now canonical, so we go over them briefly. When compared with more standard situations the only major obstacle is the construction of a local parametrix near the origin, which appears in virtue of the factor $\sigma_n$ which scales singularly near the origin. In this construction, we will use the model problem that we previously studied in Sections~\ref{sec:modelproblemfull}--\ref{sec:conseqmodelproblem}.

\subsection{First transformation: introduction of the $g$-function}\hfill 

The first transformation consists in introducing the $g$-function, which here is replaced by the $\phi$-function from \eqref{deff:phifction}.

Transform
$$
\bm T(z)\deff \ee^{-n\ell_V\sp_3}\bm Y(z)\ee^{n(\phi(z)-V(z)/2)\sp_3},\quad z\in \C\setminus [0,\infty).
$$
Then $\bm T$ satisfies the following RHP.
\begin{rhp} Find a $2\times 2$ matrix-valued function $\bm Y:\C\setminus [0,+\infty)\to \C^{2\times 2}$ with the following properties.
\begin{enumerate}[(i)]
\item $\bm T:\C\setminus [0,\infty)\to \C^{2\times 2}$ is analytic.
\item The matrix $\bm T$ has continuous boundary values $\bm T_\pm$ along $(0,+\infty)$, and they are related by $\bm T_+(x)=\bm T_-(x)\bm J_{\bm T}(x)$, $x>0$, with
$$
\bm J_{\bm T}(x)\deff
\begin{pmatrix}
\ee^{n(\phi_+(x)-\phi_-(x))} & \sigma_n(x)x^\alpha\ee^{-n(\phi_+(x)+\phi_-(x))} \\
0 & \ee^{-n(\phi_+(x)-\phi_-(x))}
\end{pmatrix},
\quad x>0.
$$
\item As $z\to \infty$,
$$
\bm T(z)=\bm I+\frac{1}{z}\bm T_1+\Boh(z^{-2}).
$$
with
$$
\bm T_1\deff \ee^{-n\ell_V\sp_3}\bm Y_1\ee^{n\ell_V\sp_3}+n\phi_\infty\sp_3.
$$
\item As $z\to 0$,
$$
\bm T(z)=
\Boh\begin{pmatrix}
1 & h_\alpha(z) \\ 1 & h_\alpha(z)
\end{pmatrix}, 
$$
where we recall that $h_\alpha$ is given in \eqref{deff:hfctionlocalbeh}.
\end{enumerate}
\end{rhp}

\subsection{Second transformation: opening of lenses}\hfill 

\begin{figure}
\begin{tikzpicture}[>=latex]
	\draw[-] (-7,0)--(7,0);
	\draw[-] (-3,-3)--(-3,3);
	\draw[fill] (-3,0) circle (0.1cm);
	\draw[fill] (3,0) circle (0.1cm);
	
	\node[above] at (-3.25,0) {\large 0};	
	\node[above] at (3.1,0.05) {\large $a$};	
	\node[above] at (-0.4,0.3) {\large $\mathcal L^+$};	
	\node[below] at (-0.4,-0.3) {\large $\mathcal L^-$};	
		
	\draw[->, ultra thick] (3,0) -- (5,0);
	\draw[->, ultra thick] (-3,0) -- (0,0);
	\draw[-, ultra thick] (-3,0) -- (7,0);

	\draw[-, ultra thick] (-3,0) to [out=60, in=135] (3,0);
	\draw[-, ultra thick] (-3,0) to [out=-60, in=225] (3,0);

	\draw[->, ultra thick] (-0.4,1.38) to (-0.2,1.38);
	\draw[->, ultra thick] (-0.4,-1.38) to (-0.2,-1.38);
	
\end{tikzpicture}
\caption{The contour $\Gamma_{\bm S}$ for RHP \ref{rhp:S}.} \label{fig:lenses}
\end{figure}
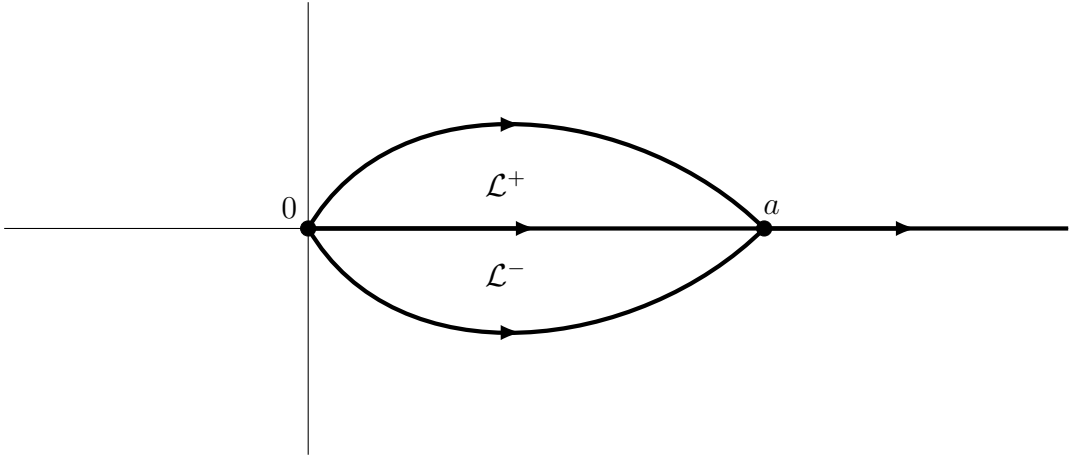

The next step is to open lenses $\mcal L^\pm$ around the interval $[0,a]$. For that, we modify
$$
\bm S(z)\deff
\begin{dcases}
\bm T(z)
\left(\bm I\mp \frac{z^{-\alpha}\ee^{2n\phi(z)}}{\sigma_n(z)}\bm E_{21}\right),& z\in \mcal L^\pm, \\
\bm T(z), & \text{otherwise}.
\end{dcases}
$$
In the above and what follows, $z^\alpha$ is always the principal branch, with branch cut on $(-\infty,0]$.
Introduce the new contour
$$
\Gamma_{\bm S}\deff [0,\infty)\cup \partial \mcal L^+\cup \partial\mcal L^-.
$$
Using the properties of the function $\phi$ discussed in Section~\ref{sec:eqmeasure}, in particular the jump relations \eqref{eq:jumpsphi}, it follows that 
the matrix $\bm S$ solves the following RHP.
\begin{rhp}\label{rhp:S} Find a $2\times 2$ matrix-valued function $\bm S:\C\setminus \Gamma_{\bm S}\to \C^{2\times 2}$ with the following properties.
\begin{enumerate}[(i)]
\item $\bm S:\C\setminus \Gamma_{\bm S}\to \C^{2\times 2}$ is analytic.
\item The matrix $\bm S$ has continuous boundary values $\bm S_\pm$ along the open subarcs of $\Gamma_{\bm S}$, and they are related by $\bm S_+(z)=\bm S_-(z)\bm J_{\bm S}(z)$, $z\in \Gamma_{\bm S}\setminus \{0,a\}$, with
$$
\bm J_{\bm S}(z)\deff
\begin{dcases}
z^\alpha\sigma_n(z)\bm E_{12}-\frac{z^{-\alpha}}{\sigma_n(z)}\bm E_{21}, & 0<z<a, \\
\bm I+\frac{z^{-\alpha}}{\sigma_n(z)}\ee^{2n\phi(z)}\bm E_{21}, & z\in \partial \mathcal L^\pm\setminus \R, \\
\bm I+\sigma_n(z)z^\alpha\ee^{-2n\phi(z)}\bm E_{12}, & z>a.
\end{dcases}
$$
\item As $z\to \infty$,
$$
\bm S(z)=\bm I+\frac{1}{z}\bm S_1+\Boh(z^{-2}).
$$
with $\bm S_1\deff \bm T_1$.
\item As $z\to 0$, the matrix $\bm S$ has the behavior 
\begin{align*}
\bm S(z) = 
\begin{cases}
\Boh\begin{pmatrix}
1 & h_\alpha(z) \\ 1 & h_\alpha(z)
\end{pmatrix}, & \text{outside the lens,} \\
\Boh\begin{pmatrix}
h_{-\alpha}(z) & h_\alpha(z) \\ h_{-\alpha}(z) & h_\alpha(z)
\end{pmatrix}, & \text{inside the lens,} \\
\end{cases}
\end{align*}
with $h_\alpha$ given in \eqref{deff:hfctionlocalbeh}.
\end{enumerate}
\end{rhp}

Thanks to \eqref{eq:ineqphi}, we 
have that $\bm J_{\bm S}(z)\to \bm I$ for $z\in \Gamma_{\bm S}\setminus [0,a]$ in an appropriate sense, and we now proceed to the construction of global and local parametrices. These parametrices are matrix-valued functions that satisfy the jumps for $\bm S$ in the interval $(0,a)$ and in neighborhoods of $z=0,a$, and will produce good approximations for $\bm S$ near these sets.

\subsection{The global parametrix}\hfill 

The global parametrix RHP is obtained after neglecting the exponentially small jumps of $\bm S$. Concretely, the global parametrix $\bm G$ is the solution to the following RHP.
\begin{rhp} Find a $2\times 2$ matrix-valued function $\bm G:\C\setminus [0,a]\to \C^{2\times 2}$ with the following properties.
\begin{enumerate}[(i)]
\item $\bm G:\C\setminus [0,a]\to \C^{2\times 2}$ is analytic.
\item The matrix $\bm G$ has continuous boundary values $\bm G_\pm$ along $[0,a]$, and they are related by 
$$
\bm G_+(z)=\bm G_-(z)\left(z^{\alpha}\sigma_n(z)\bm E_{12}-\frac{z^{-\alpha}}{\sigma_n(z)}\bm E_{21}\right),\quad 0<z<a.
$$
\item As $z\to \infty$,
$$
\bm G(z)=\bm I+\Boh(z^{-1}).
$$
\item The matrix $\bm G$ has at worse $L^2$-integrable singularities at $z=0,a$.
\end{enumerate}
\end{rhp}

The construction of the solution $\bm G$ is standard, and we now describe it. As we will see, $\bm G$ depends on $\sigma_n$, so in fact $\bm G=\bm G_n$, and we will also need to understand the behavior of $\bm G_n$ as $n\to \infty$, which we also describe in this section.

The solution $\bm G$ is explicitly given by
\begin{equation}\label{deff:globalparametrix}
\bm G(z)\deff \ee^{-\msf p_0\sp_3}\msf D_\infty^{\sp_3}\bm M(z)\msf D(z)^{-\sp_3}\ee^{\msf p(z)\sp_3},\quad z\in \C\setminus [0,a],
\end{equation}
where each of the quantities is as follows. The function $\msf D$ is analytic on $\C\setminus [0,a]$, and satisfies
$$
\msf D_+(z)\msf D_-(z)=z^\alpha \text{ for }z\in (0,a), \quad \text{and}\quad \msf D(z)\to \msf D_\infty\neq 0 \text{ as }z\to \infty.
$$
With standard methods, we find
$$
\msf D(z)\deff \frac{z^{\alpha/2}}{\msf d(z)^{\alpha/2}},\quad z\in \C\setminus [0,a],
$$
where $\msf d(z)$ is the conformal map from $\C\setminus [0,a]$ to the exterior of the unit disk, normalized to satisfy $\lim_{x\to +\infty} x^{-1}\msf d(x)\revdeff \msf d_\infty>0$. A simple calculation shows that
$$
\msf d(z)=\frac{2}{a}z-1+\frac{2}{a}z^{1/2}\left(z-a\right)^{1/2},\quad z\in \C\setminus [0,a],
$$
with principal branch of the root, so that $\msf d(z)=4z/a+\Boh(1)$ as $z\to \infty$ and therefore $\msf d_\infty=4/a$. Thus, $\msf D_\infty=2^{-\alpha} a^{\alpha/2}$.

The function $\msf p(z)$ solves the scalar RHP
$$
\msf p_+(z)+\msf p_-(z)=-\log\sigma_n(z), \quad 0<z<a,
$$
with behavior
$$
\msf p(z)=\Boh(1),\quad z\to 0,a,\infty.
$$
With standard methods, we find
\begin{equation}\label{deff:msfp}
\msf p(z)\deff \frac{(z(z-a))^{1/2}}{2\pi }\int_0^a \frac{\log(\sigma_n(x))}{\sqrt{|x(x-a)|}}\frac{\dd x}{x-z},\quad z\in \C\setminus [0,a],
\end{equation}
with principal branch of the square root.
The constant $\msf p_0$ is determined from
$$
\msf p(z)=\msf p_0+\frac{\msf p_1}{z}+\Boh(z^{-2}),\quad z\to\infty,
$$
from which we find
$$
\msf p_0\deff -\frac{1}{2\pi}\int_0^a \frac{\log(\sigma_n(x))}{\sqrt{|x(x-a)|}} \dd x,\quad
\msf p_1\deff -\frac{1}{2\pi}\int_0^a \frac{x\log(\sigma_n(x))}{\sqrt{|x(x-a)|}} \dd x+ \frac{a}{2}\msf p_0.
$$
For
$$
\bm U_0\deff \frac{1}{\sqrt{2}}
\begin{pmatrix}
1 & \ii \\ \ii & 1
\end{pmatrix},\quad \msf g(z)\deff \frac{z-a}{z},
$$
the matrix $\bm M$ is
\begin{equation}\label{eq:deffMglobalparamOP}
\bm M(z)=\bm U_0\msf g(z)^{\sp_3/4}\bm U_0^{-1},
\end{equation}
which solves the jump
$$
\bm M_+(z)=\bm M_-(z)(\bm E_{12}-\bm E_{21}).
$$
The function $\msf p=\msf p(\cdot\mid n)$ and the constant $\msf p_0=\msf p_0(n)$ both vary with $n$, but we mostly omit this dependence in our notation. It is however important to establish its behavior as $n\to\infty$. Using Proposition~\ref{prop:intLogestimate} we estimate these terms, the result is summarized in the next proposition.
\begin{prop}\label{prop:boundspn}
Set
$$
\msf p_\infty(\sad)\deff \frac{a^{-1/2}\tad^{-1/(2m)}}{2\pi m}\msf F_{\frac{1}{2m}-1}(\sad),
$$
where $\msf F_\beta(\sad)=-\beta\Gamma(\beta)\Li_{\beta+2}(-\ee^{-\sad})$ with $\Li_{\beta}$ being the polylogarithms. For any $\sad_0\in \R$ fixed, the estimate
$$
\msf p_0=\frac{1}{n}\msf p_\infty(\sad)+\Boh(n^{-1-2m}),\quad n\to \infty,
$$
holds uniformly for $\sad \geq \sad_0$, and the estimate
$$
\msf p(z)=\left(\frac{z-a}{z}\right)^{1/2}\frac{1}{n}\msf p_\infty(\sad)+\Boh(n^{-1-2m}),\quad n\to \infty
$$
holds uniformly for $\sad \geq \sad_0$, and also uniformly for $z$ in compacts of $\C\setminus \{0,a\}$, where for $z\in (0,a)$ it should be understood for boundary values $\msf p_{\pm}(z),\msf p_{\infty,\pm}$.
\end{prop}

\begin{proof}
We may estimate $\msf p_0$ and $\msf p$ using Proposition \ref{prop:intLogestimate}--(ii), once we identify $\beta=-\frac12$, $Q(x)=t q(x)$, and respectively $f(x)=|x-a|^{-1/2}$ and $f(x)=|x-a|^{-1/2} (x-z)^{-1}$ (or rather we look at the real and imaginary part of $f$). In the case that $z\in (0,a)$, the reader may verify that the (half) residue contribution corresponding to $x=z$ to the integral defining $\msf p(z)$ does not contribute to the dominant order, and one may effectively apply Proposition \ref{prop:intLogestimate}--(ii) for some $\delta<z$. We obtain the stated result after invoking Proposition \ref{prop:intLogestimate}--(i).
\end{proof}

As a consequence, we have the following proposition.
\begin{prop}\label{prop:boundglobalparam}
The global parametrix $\bm G$ defined in \eqref{deff:globalparametrix} satisfies
$$
\bm G(z)=\left(\bm I+\Boh(n^{-1-2m})\right)\msf D_\infty^{\sp_3}\bm M(z)\msf D(z)^{-\sp_3},\quad n\to \infty,
$$
uniformly for $z$ in compacts of $\C\setminus \{ a,0\}$, where for $z\in (0,a)$ this relation should be understood as valid (uniformly) for either of the $\pm$-boundary values. In particular, the boundary values $\bm G_\pm$ remain uniformly bounded on compacts of $(0,a)$.
\end{prop}

\begin{proof}
Proposition~\ref{prop:boundspn} ensures that as $n\to \infty$,
$$
\ee^{-\msf p_0\sp_3}=\bm I+\Boh(n^{-1-2m}), \quad \text{as well as}\quad \ee^{\msf p(z)\sp_3}=\bm I+\Boh(n^{-1-2m}),
$$
with the latter being valid uniformly for $z\in \C\setminus \{0,a\}$ in compacts. Because $\msf D_\infty,\msf D$ and $\bm M$ are independent of $n$ and are bounded on compacts of $\C\setminus \{0,a\}$, we then have that
$$
\msf D_\infty^{\sp_3}\bm M(z)\msf D(z)^{-\sp_3}\ee^{\msf p(z)\sp_3}=\left(\bm I+\Boh(n^{-1-2m})\right)\msf D_\infty^{\sp_3}\bm M(z)\msf D(z)^{-\sp_3},
$$
and the result follows from the explicit expression \eqref{deff:globalparametrix}.
\end{proof}

With the global parametrix now constructed, we move forward to the local parametrices.

\subsection{The local parametrix near the hard edge}\label{sec:localparamhardedge}\hfill \\
The local parametrix $\bm P^{(0)}$, constructed in a neighborhood $U_0$ of $z=0$, reads as follows.
\begin{rhp} Find a $2\times 2$ matrix-valued function $\bm P^{(0)}: U_0\setminus \Gamma_{\bm S}\to \C^{2\times 2}$ with the following properties
\begin{enumerate}[(i)]
\item $\bm P^{(0)}$ is analytic on $U_0\setminus \Gamma_{\bm S}$.
\item It has jump $\bm P^{(0)}_+(z)=\bm P^{(0)}_-(z)\bm J_{\bm S}(z)$, $z\in \Gamma_{\bm S}\cap U_0$.
\item For $z\in \partial U_0$, it behaves as
$$
\bm P^{(0)}(z)=(\bm I+\boh(1))\bm G(z),\quad n\to \infty.
$$
%
%
\end{enumerate}
\end{rhp}

Making the change
\begin{equation}\label{eq:transfLP}
\bm L(z)=\bm P^{(0)}(z)\ee^{-n\phi(z)\sp_3}(-z)^{\alpha \sp_3/2},\quad z\in U_0\setminus \Gamma_{\bm S},
\end{equation}
with the principal branch of the root $(\cdot)^{\alpha/2}$, it follows that $\bm L$ should satisfy the following RHP.

\begin{rhp} Find a $2\times 2$ matrix-valued function $\bm L:U_0\setminus \Gamma_{\bm S}\to \C^{2\times 2}$ with the following properties.
\begin{enumerate}[(i)]
\item $\bm L$ is analytic on $U_0\setminus \Gamma_{\bm S}$.
\item It has jump $\bm L_+(z)=\bm L_-(z)\bm J_{\bm L}(z)$, $z\in \Gamma_{\bm S}\cap U_0$, with
$$
\bm J_{\bm L}(z)\deff
\begin{dcases}
\sigma_n(z)\bm E_{12}-\frac{1}{\sigma_n(z)}\bm E_{21}, & z>0, \\
\bm I+\frac{\ee^{\mp \pi \alpha \ii}}{\sigma_n(z)}\bm E_{21}, & z\in (\partial \mcal L^\pm\cap U_0)\setminus [0,\infty), \\
\end{dcases}
$$
\item For $z\in \partial U_0$, it behaves as
$$
\bm L(z)=(\bm I+\boh(1))\bm G(z)(-z)^{\alpha \sp_3/2}\ee^{-n\phi(z)\sp_3},\quad n\to \infty.
$$
%
\end{enumerate}
\end{rhp}

The next step is to use the conformal map $\psi$ from \eqref{eq:conformalmap} to translate the RHP in the $z$-plane to a RHP in the transformed plane. For that, we look at
$$
\zeta\deff n^{2}\psi(z), \quad \text{that is,} \quad z=\psi^{-1}\left(\frac{\zeta}{n^2}\right),
$$
as a conformal change of variables. This correspondence satisfies
$$
z<0 \quad \text{ if, and only if, \quad } \zeta>0.
$$
This way, it follows that $2\zeta^{1/2}=2(n^2\psi(z))^{1/2}=n(\phi(z)^2)^{1/2}=\pm n \phi(z)$ for the principal branch of the log and some choice of the sign. Because $\phi(z)<0$ for $z<0$, we have in fact
$$
2\zeta^{1/2}=-n\phi(z).
$$
We also have to account for the transformation of the jump from the $z$-plane to the $\zeta$-plane. For that, for $\psi^{-1}$ being the inverse of $\psi$, we take $\delta>0$ such that $\psi^{-1}$ is analytic on the disk $D_{2\delta}(0)$ in the $w$ plane, and such that $\psi(U_0)= D_\delta(0)$. We also make sure that the lips of the lenses are mapped to the contours $\Gamma^\pm$ from \eqref{deff:thetambasiccontours}. 

Introduce the new function
$$
\msf H^Q(w)\deff Q(\psi^{-1}(w)),\quad |w|<\delta,
$$
which is an analytic function on $D_\delta(0)$ that satisfies
$$
\msf H^Q(w)=(-1)^m\frac{\tad}{\msf c_V^m}w^{m}(1+\Boh(w)),\quad w\to 0.
$$
In particular, the function
\begin{equation}\label{eq:sigmantohzeta}
\msf h_n(\zeta)\deff \sad +n^{2m}Q(z)=\sad+n^{2m}Q\left(\psi^{-1}\left(\frac{\zeta}{n^{2}}\right)\right)=\sad+n^{2m}\msf H^Q\left(\frac{\zeta}{n^{2}}\right),
\end{equation}
is analytic on $|\zeta|\leq \delta n^{2}$ and satisfies
$$
\msf h_n(\zeta)=\sad+(-1)^m\uad \zeta^m+\Boh(n^{-2}),\quad n\to\infty,\qquad \uad\deff \frac{\tad}{\msf c_V^m}>0,
$$
uniformly for $\zeta$ in compacts of $\C$, or also the more refined estimate
$$
\msf h_n(\zeta)=\sad +(-1)^m \msf u\zeta^m+\Boh(\zeta^{m+1}n^{-2}),\quad n\to\infty,
$$
which is valid for $|\zeta|\leq \delta n^{2}$.

We extend $\msf H^Q$ in a $C^\infty$ way to a neighborhood of $\Gamma$, so that with the identification $\tau=n$, this function $\msf h_n$ is admissible in the sense of Definition~\ref{def:admissibledata}. In particular, from the considerations above and \eqref{eq:convsigmatausigmainfty},
\begin{equation}\label{eq:convgsigmanhinfty}
\left|\sigma_n(x)-\sigma_\infty(\zeta(x))\right|=\Boh(\ee^{-\sad }n^{-2}),\quad n \to \infty,\qquad \sigma_\infty(\zeta)\deff \frac{1}{1+\ee^{-\msf h_\infty(\zeta)}},
\end{equation}
uniformly for $\zeta$ in compacts of $\C$, where we recall that $\msf h_\infty(\zeta)=\sad+(-1)^m\msf u\zeta^m$ is as in \eqref{def:hinfty}.

Moving forward with the construction for $\bm L$, set 
\begin{equation}\label{eq:deffEn}
\bm E_n(z)\deff \bm E_1(z)n^{\sp_3/2},\quad \bm E_1(z)\deff \msf D_\infty^{\sp_3}\bm M(z)(-z)^{\alpha\sp_3/2}\msf D(z)^{-\sp_3}\bm U_0(\psi(z))^{\sp_3/4},\quad z\in U_\delta.
\end{equation}
A direct calculation from its definition shows that $\bm E_1$ is analytic near $z=0$.

With all these definitions, we construct the solution $\bm L$ in the form
\begin{equation}\label{eq:LEPhi}
\bm L(z)=\bm E_n(z)\wh{\bm \Psi}_n(z),\qquad \wh{\bm \Psi}_n(z)=\sp_3\bm \Psi_n(\zeta)\sp_3, \quad \zeta=\zeta(z)=n^{2}\psi(z),
\end{equation}
where $\bm\Psi_n$ should satisfy the following RHP.
\begin{rhp}\label{rhp:localparamPsin}
Find a $2\times 2$ matrix-valued function $\bm \Psi_n:\C\setminus \Gamma\to \C^{2\times 2}$ with the following properties.
\begin{enumerate}[(i)]
\item $\bm \Psi_n$ is analytic on $\C\setminus \Gamma$.
\item The jump matrix $\bm J_n$ for $\bm \Psi_n$ is
$$
\bm J_n(\zeta)\deff
\begin{dcases}
\bm I+(1+\ee^{-\msf h_n(\zeta)})\ee^{\pi \alpha\ii}\bm E_{21}, & \zeta\in \Gamma_1, \\
\frac{1}{1+\ee^{-\msf h_n(\zeta)}}\bm E_{12}-(1+\ee^{-\msf h_n(\zeta)})\bm E_{21}, & \zeta\in \Gamma_2, \\
\bm I+(1+\ee^{-\msf h_n(\zeta)})\ee^{-\pi \alpha\ii}\bm E_{21}, & \zeta\in \Gamma_3,
\end{dcases}
$$

\item As $\zeta\to \infty$,
\begin{equation}\label{eq:asympmodelprobl}
\bm \Psi_n(\zeta)=\zeta^{-\sp_3/4}\bm U_0\left(\bm I+\Boh(\zeta^{-1/2})\right)\ee^{2\zeta^{1/2}\sp_3},
\end{equation}
where the error term is valid uniformly for $|\zeta|>R$, for some $R>0$ independent of $n$.
\end{enumerate}
\end{rhp}

A direct comparison shows that, with the identification $n=\gt$, the solution $\bm \Psi_n$ of the model RHP~\ref{RHP:model} is a solution $\bm\Psi_n$ of this RHP. In fact, RHP~\ref{rhp:localparamPsin} (i) and (ii) are in direct correspondence with RHP~\ref{RHP:model}, and a direct calculation shows that RHP~\ref{RHP:model}--(iii) implies RHP~\ref{rhp:localparamPsin}--(iii) (see also Remark~\ref{rmk:expansioninftymodeladm}). From now on, we refer to this choice $\bm \Psi_n$ as being the solution to RHP~\ref{rhp:localparamPsin}, even though RHP~\ref{rhp:localparamPsin} as posed may admit more than one solution\footnote{To ensure unique solution, we would have to specify the behavior as $\zeta\to 0$, but for clarity of presentation we opt for the approach presented without this condition a priori, only a posteriori}.

In particular, Proposition~\ref{Prop:behmodelproblemorigin} ensures that
\begin{multline}\label{eq:Phinlocalbeh}
\bm \Psi_n(\zeta)=\bm \Psi^{(n)}_0(\zeta)\zeta^{\alpha\sp_3/2}(\bm I+\msf a(\zeta)\bm E_{12})\left( \bm I-(\chi_{\mcal S_+}(\zeta)\ee^{\pi \ii \alpha}-\chi_{\mcal S_-}(\zeta)\ee^{-\pi\ii \alpha})\bm E_{21} \right)
\\
\times (1+\ee^{-\msf h_n(\zeta)})^{\sp_3/2}, \; \zeta\to 0,
\end{multline}
where $\bm \Psi_0^{(n)}$ is analytic near the origin. 

With the found $\bm L$, we now recover $\bm P^{(0)}$ from \eqref{eq:transfLP}. We now collect some properties that will play some role later on.

\begin{prop}\label{lem:fundpropertiesE1Phin}
The following properties hold.
\begin{enumerate}
\item [\rm{(i)}] The matrix $\bm E_1$ is independent of $n$, and it is analytic in a neighborhood of $z=0$. 
\item [\rm{(ii)}]For any $\sad_0\in \R$, the function $\bm P^{(0)}$ satisfies
\begin{equation}\label{eq:matchingPo}
\bm P^{(0)}(z)=\left(\bm I+\Boh(n^{-1})\right)\bm G(z),\quad n\to \infty,
\end{equation}
uniformly for $z\in \partial U_\delta$, and uniformly for $\sad\geq \sad_0$.

\item [\rm{(iii)}] The entries of the matrix $\bm \Psi_{n,\pm}(\zeta)^{\pm 1}$ remain bounded as $n\to\infty$, uniformly for $\zeta\in \R$ with $1/M\leq |\zeta|\leq M$ and $\sad\geq \sad_0$ for any $\sad_0\in \R$ and any $M>0$.

\item [\rm{(iv)}]
For any $\sad_0\in \R$ fixed, there exist $n_0>0$ and $M>0$, for which the matrix $\bm \Psi_0^{(n)}$ is analytic and uniformly bounded on $|\zeta|<1/M$, for any $n\geq n_0$ and any $\sad\geq \sad_0$.

\item [\rm{(v)}] For any $\sad_0>0$ fixed, there exist $n_0>0$ and $M>0$ for which the expansion \eqref{eq:asympmodelprobl} is valid on $|\zeta|>M$, uniformly for $n\geq n_0$ and $\sad\geq \sad_0$.
\end{enumerate}
\end{prop}

\begin{proof}
From its explicit expression, it is immediate that $\bm E_1$ is independent of $n$ and it is analytic on $U_\delta\setminus [0,a]$. A straightforward calculation shows that, in addition, it has no jumps along $(0,a)$, so its singularity at $z=0$ is isolated. From the behavior of its entries, we see that in fact this singularity is removable, and the analyticity of $\bm E_1$ follows.

The estimate on $\bm L$ claimed in (ii) follows from the correspondence $\zeta=n^2\psi(z)$, the RHP satisfied by $\bm \Psi_n$, namely RHP~\ref{RHP:model}--(iii), and Proposition~\ref{prop:boundglobalparam}. We skip details.

The claims (iii) and (iv) are true if we replace $\bm \Psi_n$ by $\bm\Psi_\infty$. Since $\bm\Psi_n\to\bm \Psi_\infty$ uniformly on compacts, standard arguments show that they are true for $\bm\Psi_n$ as well, as long as we take $n$ sufficiently large.

Finally, claim (v) is the same as Remark~\ref{rmk:expansioninftymodeladm} for $\bm\Psi_\gt=\bm\Psi_n$.
\end{proof}

These considerations finish the construction of the local parametrix.

\subsection{The local parametrix near the soft edge}\hfill \\
The local parametrix $\bm P^{(a)}$ near the soft edge $z=a$ is constructed in the usual way. Its main building block is the RHP for the Airy function that we describe next.

Set
$$
\Gamma_{\bm A}\deff \R\cup (\ee^{2\pi \ii/3}\infty,0]\cup (\ee^{-2\pi \ii /3}\infty,0],
$$
with the orientation of the arcs on real axis to be the natural orientation induced by $\R$, and with the non-real arcs oriented from $\infty$ to $0$. The Airy parametrix $\bm A$ is the solution to the following RHP.

\begin{rhp}\label{rhp:Airy}
Find a $2\times 2$ matrix-valued function $\bm A:\C\setminus \Gamma_{\bm A}\to \C^{2\times 2}$ with the following properties.
\begin{enumerate}[(i)]
\item $\bm A$ is analytic on $\C\setminus \Gamma_{\bm A}$.
\item The jump condition $\bm A_+(\zeta)=\bm A_-(\zeta)\bm J_{\bm A}(\zeta)$, $\zeta\in \Gamma_{\bm A}$ holds, where the jump matrix $\bm J_{\bm A}$ for $\bm A$ is
$$
\bm J_{\bm A}(\zeta)\deff
\begin{dcases}
\bm I+\bm E_{12}, & \zeta\in (0,\infty),\\
\bm I+\bm E_{21}, & \zeta\in (\ee^{\pm 2\pi \ii/3},0),\\
\bm E_{12}-\bm E_{21}, & \zeta\in (-\infty,0).
\end{dcases}
$$
\item As $\zeta\to \infty$,
\begin{equation}\label{eq:expAiryparaminfinity}
\bm A(\zeta)=\bm \zeta^{-\sp_3/4}\bm U_0\left(\bm I+\Boh(\zeta^{-3/2})\right) \ee^{-\frac{2}{3}\zeta^{3/2}\sp_3}.
\end{equation}
\end{enumerate}
\end{rhp}

The solution $\bm A$ can be constructed explicitly using Airy functions \cite{deift_book}, but we will not need its explicit expression. We now explain how to use it in our framework to construct the local parametrix near $z=a$.

Since we are assuming $Q$ to be analytic on a neighborhood of the positive real axis, the function $\sigma_n$ admits an analytic continuation to a neighborhood of $z=a$. Furthermore, its poles are the solutions to
$$
Q(z)=-\frac{\sad}{n^{2m}}+\frac{(2k+1)\pi \ii}{n^{2m}}.
$$
In particular, assuming that $\sad \geq -\sad_0$, with $\sad_0>0$ fixed, we see that there are no poles of $\sigma_n$ near $z=a$. Furthermore, $Q(a)>0$, and therefore for the same range of values of $\sad$,
$$
\sad +n^{2m}\re Q(z)\geq \eta n^{2m},
$$
for some $\eta>0$, on any (small enough) neighborhood $D_\delta(a)$ of $a$. In particular, this means that the functions $\sigma_n$ and $\sigma_n^{-1}$ both admit an analytic square root on a neighborhood $D_\delta(a)$ independent of $n$ and $\sad$, and for this square root,
$$
\sigma_n(z)^{\sp_3/2}=\bm I+\Boh(\ee^{-\eta n^{2m}}),\quad n\to \infty,
$$
uniformly for $z\in \overline{D}_\delta(a)$ and also uniformly for $\sad\geq -\sad_0$.

Let us also introduce
\begin{equation}\label{deff:tildephi}
\wt\phi(z)\deff \phi(z)+\pi \ii \quad \text{and}\quad \wt\psi(z)\deff \left(\frac{3}{2}\wt\phi(z)\right)^{2/3}.
\end{equation}
With standard arguments, we see that $\wt\psi$ is a conformal map from a neighborhood of $z=a$ to a neighborhood of the origin, and it maps $(a,a+\delta)$ to the positive real axis. With this in mind, the local parametrix near $z=a$ takes the form
\begin{equation}\label{eq:localparamAiry}
\bm P^{(a)}(z)\deff \bm E^{(a)}_n(z)\bm A(n^{2/3}\wt\psi(z)) \ee^{n\wt\phi(z)\sp_3}\sigma_n(z)^{-\sp_3/2}z^{-\alpha\sp_3/2},
\end{equation}
with
\begin{equation}\label{eq:matchingfactorAiryparam}
\bm E^{(a)}_n(z)\deff \bm E^{(a)}_1(z)n^{\sp_3/6},\quad \bm E^{(a)}_1(z)\deff \msf D_\infty^{\sp_3}\bm M(z)\msf D(z)^{-\sp_3}z^{\alpha\sp_3/2}\bm U_0^{-1}\wt\psi(z)^{\sp_3/4}.
\end{equation}

Also, it is straightforward to show that this factor $\bm E^{(a)}_1$ is analytic near $z=a$, and that the matching
\begin{equation}\label{eq:matchingPa}
\bm P^{(a)}(z)=(\bm I+\Boh(n^{-1}))\bm G(z),\quad n\to \infty,\quad z\in \partial U_a\deff D_\delta(a),
\end{equation}
holds uniformly for $\sad\geq -\sad_0$ with any fixed $\sad_0>0$.

\subsection{Conclusion of the asymptotic analysis}\hfill 

We are ready to conclude the asymptotic analysis. Set
$$
\Gamma_{\bm R}\deff \Gamma_{\bm S}\setminus \left(U_0\cup U_a\cup (0,a)\right),
$$
and orient the parts $\partial U_0$ and $\partial U_0$ of $\Gamma_{\bm R}$ in the clockwise direction. Let $\bm P^{(a)}$ and $\bm P^{(0)}$ be the local parametrices near $z=a$ and $z=0$, respectively, and set
$$
\bm P(z)\deff
\begin{cases}
\bm P^{(0)}(z), & z\in U_0, \\
\bm P^{(a)}(z), & z\in U_a, \\
\bm G(z), & z\in \C\setminus \left(\Gamma_{\bm S}\cup \overline U_0\cup \overline U_a\right).
\end{cases}
$$
Make the transformation
$$
\bm R(z)\deff \bm S(z)\bm P(z)^{-1},\quad z\in \C\setminus \Gamma_{\bm R}.
$$
Then $\bm R$ satisfies the following RHP.
\begin{rhp} \label{rhp:RfinalOP}
Find a $2\times 2$ matrix-valued function $\bm R:\C\setminus \Gamma_{\bm R}\to \C^{2\times 2}$ with the following properties.
\begin{enumerate}[(i)]
\item $\bm R$ is analytic on $\C\setminus \Gamma_{\bm R}$.
\item The jump condition $\bm R_+(\zeta)=\bm R_-(\zeta)\bm J_{\bm R}(\zeta)$, $\zeta\in \Gamma_{\bm A}$ holds, where 
$$
\bm J_{\bm R}(z)\deff
\begin{dcases}
\bm G(z)\bm J_{\bm S}(z)\bm G(z)^{-1},& z\in \Gamma_{\bm R}\setminus (\partial U_0\cup \partial U_a), \\
\bm P^{(0)}(z)\bm G(z)^{-1},& z\in \partial U_0, \\
\bm P^{(a)}(z)\bm G(z)^{-1}, & z\in \partial U_a.
\end{dcases}
$$
\item As $z\to \infty$,
$$
\bm R(z)=\bm I+\Boh(z^{-1}),
$$
where the error term is valid uniformly for $|\zeta|>R$ with some $R>0$ independent of $n$.
\end{enumerate}
\end{rhp}

With standard arguments, in particular thanks to \eqref{eq:matchingPo} and \eqref{eq:matchingPa}, we obtain the next result.
\begin{theorem}\label{thm:estRforOP}
The matrix $\bm R$ satisfies
$$
\bm R(z)=\bm I+\Boh(n^{-1}),\quad n\to \infty,
$$
uniformly for $z\in \C\setminus \Gamma_{\bm R}$. Furthermore, this estimate holds for boundary values $\bm R_\pm$ along $\Gamma_{\bm R}$, in $L^\infty$ and $L^1$ norms, uniformly for $\sad \geq -\sad_0$ with any fixed $\sad_0>0$.
\end{theorem}

This last result concludes the asymptotic analysis, and now we move on to deriving its consequences.

\section{Conclusion of main results}\label{sec:conclusionmainresults}

Now that the asymptotic analysis for orthogonal polynomials is complete, we are ready to draw its consequences, in particular proving our main results.

For reference throughout this section, we trace back all the transformations of the RH analysis for OPs, which yields
\begin{equation}\label{eq:unfoldYR}
\bm Y(z)=
\ee^{n\ell_V\sp_3}\bm R(z)\bm P(z)\left(\bm I+\left(\chi_{\mcal L^+}(z)-\chi_{\mcal L^-}(z)\right)\frac{z^{-\alpha}\ee^{2n\phi(z)}}{\sigma_n(z)}\bm E_{21}\right)\ee^{-n(\phi(z)-V(z)/2)\sp_3},\quad z\in \C,
\end{equation}
where we recall that $\mcal L^\pm$ are the lenses used in the transformation $\bm T\mapsto\bm S$ (see Figure~\ref{fig:lenses}), and $\phi$ is as in \eqref{deff:phifction}.

\subsection{Asymptotics for the kernel: proof of Theorem~\ref{thm:kernelintro}}\hfill 

To obtain the asymptotic behavior of the kernel, we follow a standard route in RHPs. However, in the case considered here, the local parametrix $\bm \Psi_n$ depends on $n$ in a nontrivial manner, so for completeness we decided to present the detailed argument.

Thanks to \eqref{eq:relKnRHPY}, the correlation kernel \eqref{deff:CorrKernel} can be expressed in terms of the solution of the RHP~\ref{rhp:YOPS} for OPs as
\begin{align*}
    \widehat{\msf K}_n(x,y \mid  s)
    = \frac{\sqrt{\omega_n(x \mid s) \omega_n(y \mid s)}}{2\pi \ii(x-y)}\bm e_2^T
    \bm Y_+(y)^{-1} \bm Y_+(x)\bm e_1.
\end{align*}
Next, stressing that we are interested in $x,y\in (0,\infty)$ in a neighborhood of the hard edge, we plug \eqref{eq:unfoldYR} into this identity. Thanks to Theorem~\ref{thm:estRforOP}, standard arguments show that
\begin{align*}
    \bm R(y)^{-1} \bm R(x) = \bm I + \mathcal O\left(\frac{x-y}{n}\right), 
\end{align*}
as $n\to\infty$, uniformly for $x,y$ in compact sets of $\C$ (in particular, the uniform convergence on $\mathbb C\setminus \Gamma_{\bm R}$ can be extended to compact subsets of $\mathbb C$). Thus,
\begin{multline*}
\wh{\msf K}_n(x,y)=\frac{(xy)^{\alpha/2}\sqrt{\sigma_n(x)}\sqrt{\sigma_n(y)} \ee^{-n(\phi(x)+\phi(y))} }{2\pi\ii (x-y)} \\
\times\left[
\left(\bm I-\frac{y^{-\alpha}\ee^{2n\phi(y)}}{\sigma_n(y)}\bm E_{21}\right)\bm P^{(0)}(y)^{-1}\left(\bm I+\Boh\left(\frac{x-y}{n}\right)\right)\bm P^{(0)}(x)\left(\bm I+\frac{x^{-\alpha}\ee^{2n\phi(x)}}{\sigma_n(x)}\bm E_{21}\right)
\right]_{21,+},
\end{multline*}
valid uniformly for $x,y$ in compacts of $(0,\infty)\cap U_0$, where we recall that $U_0$ is the neighborhood where the local parametrix $\bm P=\bm P^{(0)}$ was constructed in Section~\ref{sec:localparamhardedge}.

Using now the very construction of $\bm P^{(0)}$ from \eqref{eq:LEPhi} and \eqref{eq:transfLP}, we simplify this last expression to
\begin{multline*}
\wh{\msf K}_n(x,y)=-\frac{\ee^{\pi\ii \alpha}\sqrt{\sigma_n(x)}\sqrt{\sigma_n(y)}  }{2\pi\ii (x-y)} 
\Bigg[
\left(\bm I+\frac{\ee^{-\pi \ii \alpha}}{\sigma_n(y)}\bm E_{21}\right)\bm\Psi_n(\zeta(y))^{-1}\sp_3 n^{-\sp_3/2}\bm E_1(y)^{-1}\\ 
\times
\left(\bm I+\Boh\left(\frac{x-y}{n}\right)\right)\bm E_1(x)n^{\sp_3/2}\sp_3\bm\Psi_n(\zeta(x))\left(\bm I-\frac{\ee^{-\pi\ii \alpha}}{\sigma_n(x)}\bm E_{21}\right)
\Bigg]_{21,+},
\end{multline*}

The function $\bm E_1$ is analytic near the origin and independent of $n$. Standard arguments show that
$$
\bm E_1(y)^{-1}\bm E_1(x)=\bm I+\Boh(x-y).
$$
Hence, setting from now on
$$
\zeta\deff \zeta(x)=n^2\psi(x),\quad \xi\deff \zeta(y)=n^2\psi(y),
$$
we update the expression above to
\begin{multline*}
\wh{\msf K}_n(x,y)=-\frac{\ee^{\pi\ii \alpha}\sqrt{\sigma_n(x)}\sqrt{\sigma_n(y)}  }{2\pi\ii (x-y)} \\
\times 
\Bigg[
\left(\bm I+\frac{\ee^{-\pi \ii \alpha}}{\sigma_n(y)}\bm E_{21}\right)\bm\Psi_n(\xi)^{-1}
\left(\bm I+\Boh\left(x-y\right)\right)\bm\Psi_n(\zeta)\left(\bm I-\frac{\ee^{-\pi\ii \alpha}}{\sigma_n(x)}\bm E_{21}\right)
\Bigg]_{21,-}.
\end{multline*}
To obtain the formula above, we moved from a $+$-boundary value to a $-$-boundary value in virtue of the change of orientation coming from $x\mapsto \zeta$; see \eqref{eq:conformalmap}.
With $\msf c_V$ being the constant in \eqref{eq:conformalmap}, let us scale
$$
x=\frac{u}{\msf c_V n^2},\quad y=\frac{v}{\msf c_{V} n^2},
$$
with $u,v>0$, so that, from \eqref{eq:conformalmap},
$$
\zeta(x)=-u +\Boh(n^{-2})\quad \text{and}\quad \xi =\zeta(y)=-v+\Boh(n^{-2}).
$$
For $u,v$ in compact subsets of the positive axis, the values $\zeta=\zeta(x)$ and $\xi=\zeta(y)$ remain uniformly in compact subsets of the negative axis, where the model problem $\bm\Psi_{n,-}$ is uniformly bounded in $n$ (see Proposition~\ref{lem:fundpropertiesE1Phin}--(iii)). Recalling also \eqref{eq:convgsigmanhinfty}, we update this estimate to
\begin{multline*}
\frac{1}{\msf c_V n^2}\wh{\msf K}_n\left(\frac{u}{\msf c_{V}n^2},\frac{v}{\msf c_{V}n^2}\right)=\frac{\ee^{\pi\ii \alpha}\sqrt{\wh\sigma_n(-u)}\sqrt{\wh\sigma_n(-v)}  }{2\pi\ii (v-u)} \\
\times 
\Bigg[
\left(\bm I+\frac{\ee^{-\pi \ii \alpha}}{\wh\sigma_n(-v)}\bm E_{21}\right)\bm\Psi_\gt(-v)^{-1}
\bm\Psi_\gt(-u)\left(\bm I-\frac{\ee^{-\pi\ii \alpha}}{\wh\sigma_n(-u)}\bm E_{21}\right)
\Bigg]_{21,-}+\Boh(n^{-2}).
\end{multline*}
where $\bm\Psi_\gt=\bm\Psi_n$ and $\wh\sigma_n(\zeta)=\sigma_n(x(\zeta))$ is admissible and such that $\wh\sigma_n(u)\to \sigma_\infty(u)=(1+\ee^{-\msf h_\infty(u)})^{-1}$.

Recalling that $\sigma_\infty(\zeta)= (1+\ee^{-\msf h_\infty(\zeta)})^{-1}$ and \eqref{eq:convgsigmanhinfty}, and applying Theorem~\ref{thm:kernelconvPsigt} we conclude the proof of Theorem~\ref{thm:kernelintro}.

\subsection{Asymptotics for the multiplicative statistics: proof of Theorem~\ref{thm:asympLn}}\hfill\\
In this section, $\partial$ denotes the derivative with respect to the parameter $\sad$, whereas $'$ will be used for the derivative with respect to the spectral (the RHP) variable, so that for instance
$$
\partial \sigma_n(x\mid \sad)=\frac{\partial \sigma_n}{\partial \sad}(x\mid \sad)\quad \text{and} \quad \sigma'_n(x\mid \sad)=\frac{\partial \sigma_n}{\partial x}(x\mid \sad).
$$

Let us choose the neighborhoods $U_0$ and $U_a$ where the local parametrices were constructed so that, to simplify  notation,
$$
U_0\cap (0,\infty)=(0,\delta)\quad \text{and}\quad U_a\cap (0,\infty)=(a-\delta,a+\delta).
$$

We split \eqref{eq:defformulamain} as
\begin{equation}\label{eq:defformulamain2}
\log\msf L_n^Q(\sad)=-\frac{1}{2\pi \ii}\int_\sad^\infty \left( \msf J_{\rm h}(u)+\msf J_{\rm b}(u)+\msf J_{\rm s}(u)+\msf J_{\rm g}(u)\right)\dd u,
\end{equation}
where the terms $\msf J_{\rm h},\msf J_{\rm b},\msf J_{\rm s},\msf J_{\msf g}$ correspond to the $x$-integrals over the neighborhood $(0,\delta)$ of the hard edge region, over the bulk region $(\delta,a-\delta)$, over the neighborhood $(a-\delta,a+\delta)$ of the soft edge, and over the gap interval $(a+\delta,\infty)$ away from the support of $\mu_V$, respectively. We split the analysis of each such term in different sections below. 

\subsubsection{Contributions near the hard edge}\hfill

In this section we analyze the integral
\begin{equation}\label{deff:Jh}
\int_{\sad}^\infty \msf J_{\rm h}(u)\dd u,\quad \text{with}\quad 
\msf J_{\rm h}(\sad)\deff 
\int_0^\delta \left[ \bm Y(x\mid \sad)^{-1}\bm Y'(x\mid \sad)\right]_{21,+}\omega_n(x\mid \sad)\partial \log\sigma_n(x\mid \sad) \dd x,
\end{equation}
which will turn out to give the leading contribution to \eqref{eq:defformulamain2}.

The idea is standard: we unwrap the transformations $\bm Y\mapsto \cdots \mapsto \bm R$ of the RHP analysis, with careful considerations along the way. When doing so, we will ultimately approximate $\bm Y$ by the local parametrix $\bm P^{(0)}$ near the hard edge, and this parametrix will give us the leading contribution.

For $x\in (0,\delta]$, that is, in the neighborhood where we constructed the hard edge parametrix, we start from \eqref{eq:unfoldYR} and use \eqref{eq:transfLP} and \eqref{eq:LEPhi} to write
$$
\bm Y_+(x)=\ee^{n\ell_V\sp_3}\bm R_+(x)\bm E_{n}(x)\wh{\bm \Psi}_{n,+}(x)\left(\bm I+ \frac{\ee^{- \pi \ii\alpha}\bm E_{21}}{\sigma_n(x)}\right)x^{-\alpha\sp_3/2}\ee^{nV(x)\sp_3/2}\ee^{\pi\ii \alpha\sp_3/2},
$$
where $\bm E_n$ and $\wh{\bm \Psi}_n$ are as in \eqref{eq:LEPhi} and \eqref{eq:deffEn}. 

With the notation $\bm\Delta_x,\bm\Delta_\zeta$ introduced in \eqref{deff:Deltax}--\eqref{deff:Deltaxzeta} in mind, we compute
\begin{multline}\label{eq:J1h00}
[\bm Y(x)^{-1}\bm Y'(x)]_{21,+}=-\ee^{\pi\ii\alpha+nV(x)}x^{-\alpha}\left[
\bm\Delta_x\left(\bm\Psi_n(\zeta(x))\left(\bm I-\frac{\ee^{-\pi\ii\alpha}}{\sigma_n(x)}\bm E_{21}\right)\right)
\right]_{21,+}\\ 
- \ee^{\pi\ii\alpha+nV(x)}x^{-\alpha}
\left[
\left(\bm I+\frac{\ee^{-\pi\ii\alpha}}{\sigma_n(x)}\bm E_{21}\right)
\bm R_{\rm h}(x) 
\left(\bm I-\frac{\ee^{-\pi\ii\alpha}}{\sigma_n(x)}\bm E_{21}\right)
\right]_{21,+},
\end{multline}
%
%
which is valid for $x\in (0,\delta)$ and where we have set
\begin{equation*}
\bm R_{\rm h}(x)=\bm R_{\rm h}(x\mid \sad )\deff
    {\bm\Psi}_n(\zeta(x))^{-1} \sp_3 n^{-\sp_3/2}\left(\bm\Delta_x\bm E_1(x) +\bm E_1(x)^{-1}\bm\Delta_x\bm R(x)\bm E_1(x) \right) n^{\sp_3/2}\sp_3 {\bm\Psi}_n(\zeta(x)).
\end{equation*}
Plugging the result into \eqref{deff:Jh}, we obtain the identity
\begin{multline}\label{eq:J1h1}
\msf J_{\rm h}(u)=-\ee^{\pi\ii \alpha}\int_0^\delta \zeta'(x)
\left[
\bm\Delta_\zeta\left(\bm\Psi_n(\zeta)\left(\bm I-\frac{\ee^{-\pi\ii\alpha}}{\sigma_n(x(\zeta)\mid u)}\bm E_{21}\right)\right)
\right]_{21,+} \partial \sigma_n(x\mid u)
\dd x \\
- \ee^{\pi \ii \alpha} \int_0^\delta 
\left[\left(\bm I+ \frac{\ee^{- \pi \ii \alpha}}{\sigma_n(x\mid u)}\bm E_{21}\right)\bm R_{{\rm h},+}(x)\left(\bm I- \frac{\ee^{- \pi \ii \alpha}}{\sigma_n(x\mid u)}\bm E_{21}\right)\right]_{21} 
\partial\sigma_n(x\mid u) 
\dd x.
\end{multline}
We now obtain a bound for the integral on the second line. From Proposition~\ref{lem:fundpropertiesE1Phin}--(i) and Theorem~\ref{thm:estRforOP} we estimate
\begin{equation}
n^{-\sp_3/2}\bm \Delta_x\bm E_1(x)n^{\sp_3/2}=\Boh(n),\qquad \text{and}\qquad 
n^{-\sp_3/2}\bm E_1(x)^{-1}\bm\Delta_x\bm R(x)\bm E_1(x)n^{\sp_3/2}=\Boh(1),
\end{equation}
as $n\to \infty$, so that
\begin{equation}\label{eq:estHE1}
\bm R_{\rm h}(x)={\bm\Psi}_n(\zeta(x))^{-1} \Boh(n){\bm\Psi}_n(\zeta(x)), \quad n\to \infty,
\end{equation}
where the error term is valid uniformly for $x\in (0,\delta)$ and uniformly for $\sad\geq -\sad_0$ with any fixed $\sad_0>0$. 

To estimate the withstanding term $\bm \Psi_n$, we explore the change of variables $\zeta=n^2\psi(x)$, which maps the interval $[0,\delta]$ to an interval of the form $[-R_n,0]$, with $R_n=\Boh(n^2)$. We split such interval in three pieces
$$
[-R_n,0]=[-R_n,-M]\cup [-M,-1/M]\cup [-1/M,0],
$$
where $M>0$ is independent of $n$ and chosen so that Proposition~\ref{lem:fundpropertiesE1Phin}--(iii),(iv),(v) are all valid. In what follows, recall also that $1/\sigma_n(z)=1+\ee^{\msf h_n(\zeta)}$, with $\msf h_n=\msf h_\gt$ admissible in the sense of Definition~\ref{def:admissibledata}.

In the compact interval $[-M,-1/M]$ the convergence $\ee^{-\msf h_n(\zeta)}\to \ee^{-\sad +(-1)^{m+1}\msf u\zeta^m}$ takes place uniformly, and this limit is bounded in this same interval, uniformly also for $\sad\geq -\sad_0$. Combining with Proposition~\ref{lem:fundpropertiesE1Phin}--(iii) and \eqref{eq:estHE1}, we obtain
\begin{equation}\label{eq:estHE2}
\left[
\left(\bm I+\frac{\ee^{-\pi\ii\alpha}}{\sigma_n(x)}\bm E_{21}\right)
\bm R_{\rm h}(x) 
\left(\bm I-\frac{\ee^{-\pi\ii\alpha}}{\sigma_n(x)}\bm E_{21}\right)
\right]_{21,+}=\Boh(n), \quad n\to\infty,
\end{equation}
valid uniformly for $\zeta(x)\in [-M,-1/M]$ and uniformly for $\sad\geq -\sad_0$. 

On the interval $[-1/M,0]$, we use \eqref{eq:Phinlocalbeh} and Proposition~\ref{lem:fundpropertiesE1Phin}--(iv) and write
\begin{equation*}
{\bm \Psi}_{n,+}(\zeta(x))\left(\bm I-\frac{\ee^{-\pi\ii\alpha}}{\sigma_n(x)}\bm E_{21}\right)= 
\Boh(1)\zeta^{\frac{\alpha}{2}\sp_3}_-(\bm I+\msf a_-(\zeta)\bm E_{12})\sigma(x)^{\sp_3/2},\quad n\to\infty,
\end{equation*}
and therefore using again \eqref{eq:estHE1},
\begin{equation}\label{eq:estHE3}
\left[
\left(\bm I+\frac{\ee^{-\pi\ii\alpha}}{\sigma_n(x)}\bm E_{21}\right)
\bm R_{\rm h}(x) 
\left(\bm I-\frac{\ee^{-\pi\ii\alpha}}{\sigma_n(x)}\bm E_{21}\right)
\right]_{21,+}=\Boh\left(n^{1+2\alpha}|x|^\alpha\right),\quad n\to \infty,
\end{equation}
valid uniformly for $\zeta=\zeta(x)\in [-1/M,0]$ and $\sad\geq -\sad_0$, and where we used that $\zeta=n^2\psi(x)$, $\psi$ is conformal near the origin, and $\sigma_n$ is bounded along the real axis.

Finally, on the interval $[-R_n,-M]$ we now use Proposition~\ref{lem:fundpropertiesE1Phin}--(v) and conclude that
\begin{multline*}
{\bm \Psi}_{n,+}(\zeta(x))\left(\bm I-\frac{\ee^{-\pi\ii\alpha}}{\sigma_n(x)}\bm E_{21}\right)= \\
\Boh(1)\zeta^{-\frac{1}{4}\sp_3}_- \bm U_0 (\bm I-(1+\ee^{-\msf h_n(\zeta)})\ee^{-\pi\ii\alpha-4\ii|\zeta|^{1/2}}\bm E_{21})\ee^{-2\ii |\zeta|^{1/2}\sp_3},\quad \zeta(x)\in [-R_n,-M], \quad n\to\infty,
\end{multline*}
also with uniform error term for $\sad\geq -\sad_0$, and emphasizing that this error term depends on the fixed value $M$ but it is independent of $R_n$. As said before, $\msf h_n$ is admissible, and Definition~\ref{def:admissibledata}--(iii) implies in particular that $(1+\ee^{-\msf h_n})$ is bounded along the negative axis, uniformly for $\sad \geq -\sad_0$ and uniformly in $n$ as well. With this observation in mind, the last estimate and \eqref{eq:estHE1} together imply
\begin{equation}\label{eq:estHE4}
\left[
\left(\bm I+\frac{\ee^{-\pi\ii\alpha}}{\sigma_n(x)}\bm E_{21}\right)
\bm R_{\rm h}(x) 
\left(\bm I-\frac{\ee^{-\pi\ii\alpha}}{\sigma_n(x)}\bm E_{21}\right)
\right]_{21,+}=
\Boh(n^2(1+|x|^{1/2})),\quad n \to \infty,
\end{equation}
uniformly for $\zeta=\zeta(x)\in [-R_n,-M]$ and $\sad\geq -\sad_0$. In fact, $x$ is bounded in the corresponding set, and the error term above is effectively $\Boh(n^2)$.

Again having in mind that $\zeta=\zeta(x)=n^2\psi(x)$ and that $\psi(x)$ is an $n$-independent conformal map near the origin, we use \eqref{eq:estHE2},\eqref{eq:estHE3} and \eqref{eq:estHE4}, obtaining that there exists a constant $C>0$ for which the integral on the second line of \eqref{eq:J1h1} is bounded by
\begin{equation}
Cn
\left(
n^{2\alpha}\int_0^{\zeta^{-1}(-1/M)} \partial\sigma_n(x\mid u) x^\alpha \dd x+ \int_{\zeta^{-1}(-1/M)}^{\zeta^{-1}(M)} \partial\sigma_n(x\mid u)\dd x+n\int_{\zeta^{-1}(M)}^\delta \partial \sigma_n(x\mid u) \dd x
\right),
\end{equation}
where the additional factor $n$ in the last term comes from the factor $\zeta^{\sp_3/4}=(n^2\psi(x))^{\sp_3/4}$ in \eqref{eq:estHE4}. An explicit calculation shows that
$$
0\leq \partial\sigma_n(x\mid \sad)=\frac{\ee^{-\sad-n^{2m}Q(x)}}{(1+\ee^{-\sad-n^{2m}Q(x)})^2}\leq \ee^{-\sad -n^{2m}Q(x)}.
$$
Having in mind that $M>0$ is fixed and $R_n=\Boh(n^2)$, we obtain that \eqref{eq:estHE4} - and thus the second line in \eqref{eq:J1h1} - is bounded by
\begin{equation}\label{eq:estHE5}
Cn \ee^{-u } \left( n^{2\alpha}\int_{0}^{1/\eta} x^{\alpha}\ee^{-n^{2m}Q(x)}\dd x+ n\int_{1/\eta}^\eta \ee^{-n^{2m}Q(x)}\dd x\right)
\end{equation}
for some $\eta,C>0$, which we emphasize may depend on $\sad_0>0$ fixed but it is independent of $\sad\geq -\sad_0$. Under our conditions on $Q$ (recall Assumption~\ref{asu:Q}), it is immediate that 
$$
n\int_{1/\eta}^\infty \ee^{-n^{2m}Q(x)} \dd x=\Boh(\ee^{-\eta'n^{2m}}),
$$
for some $\eta'>0$ independent of $\sad$, and the change of variables $n^2x=y$ yields the estimate
$$
n^{2\alpha}\int_{0}^{1/\eta} x^{\alpha}\ee^{-n^{2m}Q(x)}\dd x=\Boh(n^{-2}).
$$
Using these estimates in \eqref{eq:J1h1}, we finally arrive at the estimate
\begin{equation}
\msf J_{\rm h}(u)=-\ee^{\pi\ii \alpha}\int_0^\delta \zeta'(x)
\left[
\bm\Delta_\zeta\left(\bm\Psi_n(\zeta)\left(\bm I-\frac{\ee^{-\pi\ii\alpha}}{\sigma_n(x(\zeta)\mid u)}\bm E_{21}\right)\right)
\right]_{21,+} \partial \sigma_n(x\mid u)
\dd x
+ \Boh\left(\frac{\ee^{-u}}{n}\right),
\end{equation}
valid as $n\to\infty$, uniformly for $u\geq -\sad_0$, for any $\sad_0>0$. We now change variables $x\mapsto \zeta$ in the integral. For that, notice that
$$
\sigma_n(x(\zeta)\mid u)=\frac{1}{1+\ee^{-\msf h_n(\zeta\mid u)}},
$$
where $\msf h_n=\msf h_\gt$ is admissible in the sense of Definition~\ref{def:admissibledata}. With $R=R_n$ determined by relation
$$
-n^2R_n=\zeta(\delta),
$$
we know that $R_n>0$ is bounded from above, and also from below away $0$, and this estimate becomes
\begin{equation}\label{eq:finalestimateJh}
\msf J_{\rm h}(u)=\ee^{\pi\ii \alpha}\int_{-n^2R}^0 
\left[
\bm\Delta_\zeta\left[\bm\Psi_n(\zeta)\left(\bm I-\ee^{-\pi\ii\alpha}\left(1+\ee^{-\msf h_n(\zeta\mid u)}\right)\bm E_{21}\right)\right]
\right]_{21,-} \frac{\ee^{-\msf h_n(\zeta\mid u)}}{(1+\ee^{-\msf h_n(\zeta\mid u)})^2}
\dd \zeta
+ \Boh\left(\frac{\ee^{-u}}{n}\right),
\end{equation}
valid as $n\to\infty$, uniformly for $u\geq -\sad_0$, for any $\sad_0>0$ fixed.

Using now Theorem~\ref{thm:Itayasympt}, we obtain
\begin{equation}\label{eq:estJhfinal}
\int_\sad^\infty\msf J_{\rm h}(u) \dd u= 2\pi\ii 
\int_\sad^\infty \int_{-\infty}^0 \msf K_\alpha(-\zeta,-\zeta\mid u) \frac{\ee^{-\msf h_\infty(\zeta\mid u)}}{1+\ee^{-\msf h_\infty(\zeta\mid u)}}\dd \zeta\dd u+\Boh\left(\frac{\ee^{-\sad}}{n}\right),\quad n\to\infty,
\end{equation}
uniformly for $\sad\geq \sad_0$.

\subsubsection{Contributions away from the support}\hfill

Returning to the analysis of \eqref{eq:defformulamain2}, 
we now estimate
\begin{equation}\label{deff:Jg}
\int_{\sad}^\infty \msf J_{\rm g}(u)\dd u,\quad \text{with}\quad 
\msf J_{\rm g}(\sad)\deff 
\int_{a+\delta}^\infty \left[\bm \Delta_x \bm Y(x\mid \sad)\right]_{21,+}\omega_n(x\mid \sad)\partial \log\sigma_n(x\mid \sad) \dd x.
\end{equation}
In the interval $[a+\delta,\infty)$, the unwrap of the transformations yields
$$
\bm Y(z)=\ee^{n\ell_V\sp_3}\bm R(z)\bm G(z)\ee^{-n(\phi(x)-V(x)/2)\sp_3}.
$$
Therefore, using the explicit expression for $\bm G$ in \eqref{deff:globalparametrix} we obtain
\begin{equation*}\label{eq:DeltaYboundsg01}
\left[\bm\Delta_x\bm Y(x)\right]_{21,+}=
\frac{\ee^{-2n\phi(x)+nV(x)+2\msf p(x)}}{\msf D(x)^2}
\left[
\bm \Delta_x \bm M(z)+
\bm M(x)^{-1}\msf D_\infty^{\sp_3}\ee^{\msf p_0\sp_3}\bm \Delta_x\bm R_+(x)\ee^{-\msf p_0}\msf D_\infty^{-\sp_3}\bm M(x)
\right]_{21}.
\end{equation*}

A direct calculation from \eqref{eq:deffMglobalparamOP} reveals that
\begin{equation}
\left[\bm\Delta_x\bm M(x)\right]_{21}=\frac{\ii}{4}\frac{\msf g'(x)}{\msf g(x)}=\frac{\ii}{4}\frac{ a}{x(x-a)},
\end{equation}
which is clearly bounded on $[a+\delta,\infty)$. Likewise, all the entries of $\bm M$ are bounded in this same interval, because so is $\msf g$. The term $\msf D$ is discussed after \eqref{deff:globalparametrix}, and it is continuous and nonzero on the interval $[a+\delta,\infty)$, and converges to $\msf D_\infty\neq 0$ at $\infty$. Therefore the factor $\msf D(x)^{-2}$ is also bounded on $[a+\delta,\infty)$. The terms $\msf D_\infty$ and $\msf p_0$ are constant in $x$, with the former being independent of $n$ as well while the latter is uniformly bounded in $n$ thanks to Proposition~\ref{prop:boundspn}. Finally, thanks to Theorem~\ref{thm:estRforOP} we know that the factor $\bm\Delta_x\bm R_+$ is $\Boh(n^{-1})$.

All in all, we conclude that, with the exception of the exponential component, all terms on the right hand side of \eqref{eq:DeltaYboundsg01} are bounded, and therefore for some $M>0$,
$$
\left|\left[\bm \Delta_x \bm Y(x\mid\sad)\right]_{21,+}\omega_n(x\mid \sad)\partial \log\sigma_n(x\mid \sad)\right|
\leq 
\frac{\ee^{-\sad-n^{2m}Q(x)}}{(1+\ee^{-\sad-n^{2m}Q(x)})^2} x^\alpha\ee^{-2n\phi(x)+2\msf p(x)},\quad x\geq a+\delta,
$$
where we used the explicit expressions for $\omega_n$ and $\sigma_n$ from \eqref{eq:deffweight} and \eqref{def:sigman}.
The function $\msf p$ does depend on both $n$ and $\sad$, but thanks to Proposition~\ref{prop:boundspn} it remains bounded uniformly both for $\sad\geq -\sad_0$ and $x\geq a+\delta$ as $n\to \infty$. Therefore, for a new constant $M>0$ the right-hand side above is bounded by
$$
M x^\alpha \ee^{-\sad-n^{2m}Q(x)-2n\phi(x)}.
$$
We integrate this bound and use that $Q(x)>0$ on the positive axis, concluding that
$$
\left|\int_\sad^\infty \msf J_{\rm g}(u)\dd u\right|
\leq M \ee^{-\sad}\int_{a+\delta}^\infty x^\alpha \ee^{-2n\phi(x)}\dd x.
$$

The function $\phi$ is given in \eqref{deff:phifction}, and as observed from \eqref{eq:asymptphi} it grows polynomially at $\infty$. A simple estimate then shows that this remaining integral is $\Boh(\ee^{-\eta n})$, for some $\eta>0$. 

We just concluded that
\begin{equation}\label{eq:estJgfinal}
\int_{\sad}^\infty \msf J_{\rm g}(u)\dd u=\Boh(\ee^{-\sad -\eta n}),\quad n\to\infty,
\end{equation}
uniformly for $\sad\geq -\sad_0$.

\subsubsection{Contributions near the soft edge}\hfill

We now estimate
\begin{equation}\label{deff:Js}
\int_{\sad}^\infty \msf J_{\rm s}(u)\dd u,\quad \text{with}\quad 
\msf J_{\rm s}(\sad)\deff 
\int_{a-\delta}^{a+\delta} \left[\bm \Delta_x \bm Y(x\mid \sad)\right]_{21,+}\omega_n(x\mid \sad)\partial \log\sigma_n(x\mid \sad) \dd x.
\end{equation}

Let $\chi_0=\chi_{(0,a)}$ be the characteristic function of the interval $(0,a)$. Recalling \eqref{eq:unfoldYR}, near the soft edge at $z=a$, the unwrap of the transformations gives the relation
$$
\bm Y_+(x)=\ee^{n\ell_V\sp_3}\bm R(x)\bm P_+(x)\left(\bm I+\frac{\ee^{2n\phi_+(x)}}{x^\alpha \sigma_n(x)}\chi_0(x)\bm E_{21}\right)\ee^{-n(2\phi_+(x)-V(x))\sp_3/2},\quad a-\delta\leq x\leq a+\delta,
$$
where $\bm P=\bm P^{(a)}$ for the rest of this section. Therefore, in the same interval,
\begin{multline}\label{eq:YinverseYsoftedge}
\left[\bm\Delta_x\bm Y(x)\right]_{21,+}=
\ee^{-n(2\phi_+(x)-V(x))}\left(\frac{\ee^{2n\phi_+(x)}}{x^\alpha \sigma_n(x)}\right)'\chi_0(x)+
\ee^{-n(2\phi_+(x)-V(x))}
\Bigg[
\left(\bm I-\frac{\chi_0(x)\ee^{2n\phi_+(x)}}{x^\alpha\sigma_n(x)}\bm E_{21}\right) \\
\times
\left(\bm\Delta_x\bm P_+(x)+\bm P_+(x)^{-1}\bm\Delta_x\bm R(x)\bm P_+(x)\right)
\left(\bm I+\frac{\chi_0(x)\ee^{2n\phi_+(x)}}{x^\alpha\sigma_n(x)}\bm E_{21}\right)
\Bigg]_{21}.
\end{multline}

Using \eqref{eq:localparamAiry} and \eqref{eq:matchingfactorAiryparam}, as well as Theorem~\ref{thm:estRforOP} and the analyticity of $\bm E^{(a)}_1$, we write
\begin{multline*}
\bm\Delta_x\bm P_+(x)+\bm P_+(x)^{-1}\bm\Delta_x\bm R(x)\bm P_+(x)=
\bm\Delta_x\left[ (\sigma_n(x)x^\alpha)^{-\sp_3/2}\ee^{n\wt\phi_+(x)\sp_3}\right]
+
(\sigma_n(x)x^\alpha)^{\sp_3/2}\ee^{-n\wt\phi_+(x)\sp_3} \\ 
\times \left[\bm\Delta_x\left(\bm A_+(n^{2/3}\wt\psi(x))\right)+
\bm A_+(n^{2/3}\wt\psi(x))^{-1}\Boh(n^{1/3})\bm A_+(n^{2/3}\wt\psi(x))\right]\ee^{n\wt\phi_+(x)\sp_3}
(\sigma_n(x)x^\alpha)^{-\sp_3/2}.
\end{multline*}

To bound the terms involving $\bm A$ we proceed in a way similar to what we did near $z=0$. For a number $R>0$ to be taken sufficiently large but fixed, define $y_n=y_n(R),x_n=x_n(R)\in (a-\delta,a+\delta)$ through the relation
$$
[y_n,x_n]=\{x\in [a-\delta,a+\delta]\mid n^{2/3}|\wt\psi(x)|\leq R \}.
$$
Choose $R>0$ in such a way that the expansion \eqref{eq:expAiryparaminfinity} holds for $|\zeta|\geq R$. This is done in such a way that, on $[a-\delta,a+\delta]\setminus [y_n,x_n]$ we may use this expansion \eqref{eq:expAiryparaminfinity} to obtain bounds for $\bm A$, whereas along $[y_n,x_n]$ we may simply use that $\bm A$ remains bounded. 

Following this idea, we obtain the same bound in both intervals, namely
$$
\bm\Delta_x\bm P_+(x)+\bm P_+(x)^{-1}\bm\Delta_x\bm R(x)\bm P_+(x)=\Boh(n),\quad x\in [a-\delta,a+\delta],
$$
where we remark that we used that $\sigma_n(x)x^\alpha$ remains bounded in any small neighborhood of $x=a$ and the identity $n\wt\phi=\frac{2}{3}(n^{2/3}\wt\psi)^{3/2}$.

We use this last estimate to bound the term in the right-most side of \eqref{eq:YinverseYsoftedge}. In addition, we bound the derivative term multiplying $\xi_0$ by observing that $\phi_+$ is purely imaginary on $(0,a)$. As a result, we obtain the overall bound
\begin{equation}\label{eq:boundDeltaYsoftedge}
\left[\bm\Delta_x\bm Y(x)\right]_{21,+}=
\Boh\left(n\ee^{-n(2\phi_+(x)-V(x))}\right),\quad n\to\infty, 
\end{equation}
valid uniformly for $|x-a|\leq \delta$ and $\sad\geq \sad_0$.

Using this estimate in \eqref{deff:Js}, we obtain that $\msf J_{\rm s}(\sad)$ is bounded by a uniform constant times
$$
n\int_{a-\delta}^{a+\delta} x^\alpha \frac{\ee^{-\sad-n^{2m}Q(x)}}{(1+\ee^{-\sad-n^{2m}Q(x)})^2}\ee^{-n\phi_+(x)}\dd x\leq n\ee^{-\sad}\int_{a-\delta}^{a+\delta} x^\alpha \ee^{-\sad-n^{2m}Q(x)}\ee^{-n\phi_+(x)}\dd x.
$$
In the remaining integral, both functions $Q(x)$ and $\phi_+(x)$ are continuous on the interval of integration, with $Q$ being strictly positive there. In particular, this observation implies that the whole integral is $\Boh(\ee^{-\eta n^{2m}})$, for some $\eta>0$. Integrating this result now in the $\sad$-variable, we obtain the estimate
\begin{equation}\label{eq:estJsfinal}
\int_\sad^\infty \msf J_{\rm s}(u)\dd u=\Boh\left(\ee^{-\sad -\eta n^{2m}}\right),\quad n\to \infty,
\end{equation}
for some (new) value $\eta>0$, uniformly for $\sad\geq \sad_0$.

\subsubsection{Contributions from the bulk}\hfill 

To finalize the analysis of the terms in \eqref{eq:defformulamain2}, we now analyze
\begin{equation}\label{deff:Jb}
\int_{\sad}^\infty \msf J_{\rm b}(u)\dd u,\quad \text{with}\quad 
\msf J_{\rm b}(\sad)\deff 
\int_{\delta}^{a-\delta} \left[\bm \Delta_x \bm Y(x\mid \sad)\right]_{21,+}\omega_n(x\mid \sad)\partial \log\sigma_n(x\mid \sad) \dd x.
\end{equation}

For $\delta<x<a-\delta$, the unwrap of the transformations now unravels as 
$$
\bm Y_+(x)=
\ee^{n\ell_V\sp_3}\bm R_+(x)\bm G_+(z)\left(\bm I+\frac{x^{-\alpha}\ee^{2n\phi_+(x)}}{\sigma_n(z)}\bm E_{21}\right)\ee^{-n(\phi_+(z)-V(x)/2)\sp_3},
$$
see \eqref{eq:unfoldYR}. 

The function $\phi_+$ is purely imaginary along $(0,a)$. The function $\bm R$, as well as its derivative, remains bounded as $n\to \infty$. As it can be seen from Proposition~\ref{prop:boundglobalparam}, the global parametrix (as well as its derivative) remains bounded as $n\to \infty$. This means that, when applying the identity above to estimate $\bm\Delta_x \bm Y_+$, possibly growing terms come only from the exponential part $\ee^{-n(\pi_+(x)-V(x)/2)\sp_3}$ and its derivative. With this observation in mind, we see immediately that
$$
\left[ \bm\Delta_x\bm Y(x)\right]_{21,+}=\Boh(n\ee^{nV(x)}),\quad n\to\infty,
$$
uniformly for $x\in (\delta,a-\delta)$ and uniformly for $\sad \geq \sad_0$. This estimate is not necessarily optimal, but sufficient for our purposes.

Proceeding in exactly the same way as we did from \eqref{eq:boundDeltaYsoftedge} onwards, we obtain
\begin{equation}\label{eq:estJbfinal}
\int_\sad^\infty \msf J_{\rm b}(u)\dd u=\Boh\left(\ee^{-\sad -\eta n^{2m}}\right),\quad n\to \infty,
\end{equation}
for some $\eta>0$, uniformly for $\sad\geq \sad_0$.

\subsubsection{Summary of contributions and conclusion of the proof of Theorem~\ref{thm:asympLn}} \hfill

Summarizing, we combine \eqref{eq:estJhfinal}, \eqref{eq:estJgfinal}, \eqref{eq:estJsfinal} and \eqref{eq:estJbfinal} into \eqref{eq:defformulamain2}, obtaining
$$
\log\msf L_n^Q(\sad)=
-\int_\sad^\infty \int_{-\infty}^0 \msf K_\alpha(-\zeta,-\zeta\mid u) \frac{\ee^{-\msf h_\infty(\zeta\mid u)}}{1+\ee^{-\msf h_\infty(\zeta\mid u)}}\dd \zeta\dd u+\Boh\left(\frac{\ee^{-\sad}}{n}\right),\quad n\to\infty,
$$
where the error term is uniform for $\sad\geq \sad_0$ with any fixed $\sad_0\in \R$. The proof of Theorem~\ref{thm:asympLn} is completed by making the change of variables $\zeta\mapsto -\zeta$ and using that $(1+\ee^{-\msf h_\infty(\zeta\mid \sad)})^{-1}=\sigma_{\bm\Phi}\left(\frac{4}{\xad^2}\zeta\mid \sad\right)$.

\appendix

\section{The Bessel parametrix}\label{sec:BesselParametrix}

The Bessel parametrix $\bm \Phi_{\alpha}^\bes$ is a solution to a canonical RHP which depends parametrically on $\alpha>-1$, and which was first introduced in \cite{KuijlaarsMcLaughlinVanAsscheVanlessen2004} and since then has been used widely as a model problem in the asymptotic analysis of RHPs. 

Set $\Gamma_\pm \deff (\ee^{\pm 2\pi\ii/3}\infty,0]$, $\Gamma_0\deff (-\infty,0]$ and $\Gamma\deff \Gamma_0\cup\Gamma_+\cup\Gamma_-$, which is consistent with  \eqref{deff:thetambasiccontours} for the choice $m=1$. Introduce

\begin{equation}\label{deff:BessParamExplicit}
\bm \Phi_{\alpha}^{(\text{Bes})}(z)\deff
\begin{pmatrix}
	I_{\alpha}(2z^{1/2})& \dfrac{\ii}{\pi} K_{\alpha}(2z^{1/2})\\
	2\pi \ii z^{1/2}I_{\alpha}'(2z^{1/2}) & -2z^{1/2}K_{\alpha}'(2z^{1/2})
	\end{pmatrix}
    \begin{dcases}
        \bm I, & |\arg z|<\frac{2\pi}{3}, \\ 
        (\bm I \mp \ee^{\pm \pi\ii\alpha}\bm E_{21}), & \frac{2\pi}{3}<\pm \arg z<\pi,
    \end{dcases}
\end{equation}
where $I_{\alpha}(z)$ and $K_{\alpha}(z)$ denote the modified Bessel functions and the principal branch is taken for $z^{1/2}$.

Following \cite[Pages~365--368, in particular Theorem~6.3, and also Equation~(8.4)]{KuijlaarsMcLaughlinVanAsscheVanlessen2004}, the matrix $\bm\Phi_\alpha^\bes$ is the solution to the following RHP.
\begin{rhp}\label{rhp:Bessel}
\hfill
\begin{itemize}
	\item [\rm(i)] $\bm \Phi_{\alpha}^\bes(z)$ is defined and analytic in $\mathbb{C} \setminus (-\infty, 0]$.
	\item [\rm(ii)] For $z\in \Gamma\setminus \{0\}$, we have
	\begin{equation}\label{eq:Besjump}
		\bm \Phi_{\alpha,+}^\bes(z) = \bm \Phi_{\alpha, -}^\bes(z)\bm J^\bes(z),\quad \text{with}\quad 
        \bm J^\bes(z)\deff 
        \begin{dcases}
            \bm I+\ee^{\pm \pi\ii\alpha}\bm E_{21}, & z\in \Gamma^\pm,\\
            \bm E_{12}-\bm E_{21}, & z\in \Gamma_0.
        \end{dcases}
     \end{equation}
	\item [\rm(iii)] With $(\alpha,0):=1$,
    $$
    (\alpha,k)\deff \frac{(4\alpha^2-1)(4\alpha^2-3^2)\cdots (4\alpha^2-(2k-1)^2)}{2^{2k}k!}, \quad k \in \mathbb{Z}_{>0},
    $$
    and
    \begin{equation}\label{def:akbk}
        a_k(\alpha)\deff \frac{(\alpha,k-1)}{4^k k }\left(\alpha^2+\frac{k}{2}-\frac{1}{4}\right),\quad b_k(\alpha)\deff \frac{(\alpha,k-1)}{4^k}\left(k-\frac{1}{2}\right), \quad k\in \mathbb{Z}_{\geq 0}, 
    \end{equation}
    the function $\bm\Phi_\alpha^\bes$ admits the following asymptotic expansion as $z \to \infty$,
\begin{equation}\label{eq:asyBesPhi}
\bm \Phi_\alpha^\bes(z) \sim 
(2\pi)^{-\sp_3/2}z^{-\sp_3/4}\bm U_0\left(\bm I+\sum_{k=1}^\infty \frac{(-1)^k}{z^{k/2}}\left(a_k(\alpha)\bm I+b_k(\alpha)\sp_2\right)\sp_3^k\right)\ee^{2z^{1/2}\sp_3},
\end{equation}
where $\sp_2=\begin{pmatrix}
    0 & -\ii
    \\
    \ii & 0
\end{pmatrix}$ and we choose the principal branches of $z^{1/4}$ and $z^{1/2}$.
	
\item [\rm(iv)] As $z \to 0$ we have
\begin{equation}\label{eq:asy0}
\bm \Phi_{\alpha}^{(\text{Bes})}(z)  =  \bm \Phi_{\alpha,0}^{\bes}(z)z^{\alpha\sp_3/2}\left(\bm I+\msf a(z)\bm E_{12}\right)(\bm I-(\chi_{\mcal S_+}(\zeta)\ee^{\pi \ii \alpha}-\chi_{\mcal S_-}(\zeta)\ee^{-\pi \ii \alpha})\bm E_{21}) ,
\end{equation}
where $ \bm \Phi_{\alpha,0}^{(\text{Bes})}(z)$ is analytic in a neighborhood of the origin, we choose principal branch for $z^{\alpha/2}$, and $\msf a(z)=\msf a_\alpha(z)$ is as in \eqref{eq:hatUpsilonasyfactor}.
\end{itemize}
\end{rhp}

The local behavior in \eqref{eq:asy0} is slightly more precise than the one described in \cite[Equations~(6.19)--(6.21)]{KuijlaarsMcLaughlinVanAsscheVanlessen2004}. Nevertheless, applying Proposition~\ref{Prop:behmodelproblemorigin} with the choice $\sigma\equiv 1$, we get that the behavior from \cite{KuijlaarsMcLaughlinVanAsscheVanlessen2004} implies \eqref{eq:asy0}. Since we also need the value $\bm \Phi_{\alpha,0}^\bes(0)$ at the origin, we now verify the behavior \eqref{eq:asy0} directly from \eqref{deff:BessParamExplicit}.

The modified Bessel function $I_\alpha(w)$ is defined through \cite[(10.25.2)]{NIST:DLMF}
$$
I_\alpha(w)\deff \left(\frac{w}{2}\right)^\alpha F_\alpha(w^2/4),\quad \text{where} \quad F_\alpha(w)\deff \sum_{k=0}^\infty \frac{w^k}{k!\Gamma(\alpha+k+1)}.
$$
Observe that $F_\alpha$ is an entire function for any $\alpha\in \R$. From this expression for $I_\alpha$, we compute
\begin{equation}\label{eq:identityIalphaFalpha}
I_\alpha(2z^{1/2})=z^{\alpha/2}F_\alpha(z)\qquad \text{and}\qquad z^{1/2}I_\alpha'(2z^{1/2})=z^{\alpha/2}\left(\frac{\alpha}{2} F_\alpha(z)+zF_\alpha'(z)\right).
\end{equation}

Using there relations and \eqref{deff:BessParamExplicit} and \eqref{eq:asy0}, we are able to compute the matrix function $\bm \Phi_{\alpha,0}^\bes$ explicitly.

For any $\alpha>-1$, we compute that for $z>0$ the first column of $\bm \Phi_{\alpha,0}^\bes$ is given by
\begin{equation}\label{eq:PsibesPsi0bes11}
\left(\bm \Phi_{\alpha,0}^\bes(z)\right)_{11}=z^{-\alpha/2}\left(\bm \Phi_{\alpha}^\bes(z)\right)_{11}=F_\alpha(z)=z^{-\alpha/2}I_\alpha(2z^{1/2})
\end{equation}
and
\begin{equation}\label{eq:PsibesPsi0bes21}
\left(\bm \Phi_{\alpha,0}^\bes(z)\right)_{21}= z^{-\alpha/2}\left(\bm \Phi_{\alpha}^\bes(z)\right)_{21}=\pi\ii \alpha F_\alpha(z)+2\pi\ii z F'_\alpha(z)=2\pi\ii z^{(1-\alpha)/2}I_\alpha'(2z^{1/2}).
\end{equation}
By analytic continuation, these identities extend from $z>0$ to $z\in \C$.

The second column of $\bm \Phi_{\alpha,0}^\bes$ could be described similarly, but we will not need it so we skip the details.

As a consequence, it is readily seen from \eqref{deff:BessParamExplicit} and \eqref{eq:asy0} that
\begin{equation}\label{eq:BesParaZero}
     \bm \Phi_{\alpha,0}^{\bes}(0)= 
 \begin{dcases}
            \begin{pmatrix}
            \frac{1}{\Gamma(\alpha+1)} & \ii \frac{\Gamma(\alpha)}{2\pi}
            \\
            \frac{\pi\ii }{\Gamma(\alpha)} & \frac{\Gamma(\alpha+1)}{2}
            \end{pmatrix}, &  \alpha \neq 0,\\
            \begin{pmatrix}
            1 & -\frac{\ii}{\pi}\log 2
            \\
            0 & 1
            \end{pmatrix}, &  \alpha = 0.
        \end{dcases}
\end{equation}

For us, we need $\bm\Phi_\alpha^\bes$ to construct the particular solution $\bm\Psi(\cdot\mid \msf h=+\infty)=\bm\Psi_\alpha^\bes$ of the model problem RHP~\ref{RHP:model}. Setting
$$
\bm \Psi_\alpha^\bes(z)\deff (\bm I+\ii(a_1(\alpha)+b_1(\alpha)) \bm E_{21}) (2\pi)^{\sp_3/2}\bm\Phi_\alpha^\bes(z),
$$
where $a_1(\alpha)$ and $b_1(\alpha)$ are given in \eqref{def:akbk}, the conditions in RHP~\ref{rhp:Bessel} imply that indeed $\bm\Psi_\alpha^\bes=\bm\Psi(\cdot\mid \msf h=+\infty)$. For the record, notice that
\begin{equation}\label{eq:1stcolumnPsiBesPhiBes}
\left(\bm\Psi_\alpha^\bes(z)\right)_{j1}=(2\pi)^{(-1)^{j+1}/2}\left(\bm\Phi_{\alpha}^\bes(z)\right)_{j1},\quad j=1,2,
\end{equation}
and we also state that \eqref{eq:asyBesPhi} yields
\begin{equation}\label{eq:asympexpPsibes}
\bm \Psi_\alpha^\bes(z)\sim \left(\bm I+\sum_{k=1}^\infty \frac{1}{z^k}\bm \Psi_{\infty,k}^\bes(\alpha) \right)z^{-\sp_3/4}\bm U_0\ee^{2z^{1/2}\sp_3},\quad z\to \infty,
\end{equation}
with
\begin{equation}\label{eq:Psiinftykgen}
\bm\Psi_{\infty,k}^\bes(\alpha)\deff (\bm I+\ii(a_1(\alpha)+b_1(\alpha)) \bm E_{21})
\begin{pmatrix}
    a_{2k}(\alpha)-b_{2k}(\alpha) & \ii (a_{2k-1}(\alpha)-b_{2k-1}(\alpha)) \\ 
    -\ii (a_{2k+1}(\alpha)+b_{2k+1}(\alpha)) & a_{2k}(\alpha)+b_{2k}(\alpha)
\end{pmatrix},\quad k\geq 1.
\end{equation}
In particular,
\begin{equation}\label{eq:asymptPsiBes1112}
\left(\bm\Psi_\alpha^\bes(z)\right)_{11}=\frac{z^{-1/4}}{\sqrt{2}}\left(1+\Boh(z^{-1})\right)\ee^{2z^{1/2}},\quad \left(\bm\Psi_\alpha^\bes(z)\right)_{21}=\frac{\ii z^{1/4}}{\sqrt{2}}\left(1+\Boh(z^{-1})\right)\ee^{2z^{1/2}},\quad z\to \infty,
\end{equation}
and
\begin{equation}\label{eq:Psiinftyk1}
\left(\bm\Psi_{\infty,1}^\bes(\alpha)\right)_{12}=\frac{\ii (4\alpha^2-1)}{16}.
\end{equation}

Furthermore, the local expansion \eqref{eq:behPsiorigin} holds with, as said, $\sigma(z)\equiv 1$ and
$$
\bm\Psi_0(z)=\bm \Psi_{0,\alpha}^\bes(z)\deff (\bm I +\ii(a_1(\alpha)+b_1(\alpha) )\bm E_{21})(2\pi)^{\sp_3/2}\bm\Phi_{\alpha,0}^\bes(z).
$$


\section{Asymptotic analysis for a class of integrals}

Consider a function $q$ holomorphic on a neighborhood of the origin, with expansion
$$
q(z)=q_0z^m(1+\Boh(z)),\quad z\to 0,
$$
for some $m\geq 1$ and $q_0>0$. Assume in addition that $q>0$ on an interval of the form $[0,\delta]$, $\delta>0$. For a number $\beta>-1$ and a smooth function $f$ on $[0,\delta]$ with $f(0)\neq 0$, introduce
$$
\msf I_t(\sad)\deff \int_0^\delta x^\beta f(x) \log(1+\ee^{-\sad-t q(x)})\dd x,\quad t>0,\quad \sad\in \R.
$$
We assume that $x^\beta f(x)\in L^1[0,\delta]$. We are interested in the behavior of $\msf I_t(\sad)$ when $t\to \infty$ simultaneously with $\sad\to +\infty$. For the next result, set
$$
\msf F_\beta(\sad)\deff \int_0^\infty x^\beta \log(1+\ee^{-\sad -x})\dd x,\quad \sad \in \R,\quad \beta>-1.
$$

\begin{prop}\label{prop:intLogestimate}\hfill
\begin{enumerate}[(i)]
\item [\rm{(i)}] For $\sad>0$ and any $\beta>-1$, the identity
$$
\msf F_\beta(\sad)=-\beta\Gamma(\beta)\Li_{\beta+2}(-\ee^{-\sad})
$$
holds, where $\Li_{\beta}$ stands for the polylogarithms. Furthermore, for any $\sad_0>0$, the estimate
$$
\msf F_\beta(\sad)=\beta\Gamma(\beta)\ee^{-\sad} +\Boh(\ee^{-2\sad}),\quad \sad \to +\infty,
$$
is valid uniformly for $\sad\geq \sad_0$.


\item [\rm{(ii)}] For any $\sad_0\in \R$, the estimate
$$
\msf I_t(\sad)=\frac{f(0)}{mq_0^{\frac{\beta+1}{m}}}\frac{1}{t^{\frac{\beta+1}{m}}}\left(\msf F_{\frac{\beta+1}{m}-1}(\sad)+\Boh\left(\frac{1}{t}\right)\right),\quad t\to \infty,
$$
holds uniformly for $\sad \geq \sad_0 $.
\end{enumerate}
\end{prop}

Notice that the function $\sad\mapsto \msf F_\beta(\sad)$ is bounded on any interval of the form $[M,\infty)$, so the estimate (ii) indeed singles out the leading term in $\msf I_t(\sad)$ on the range $\sad \in (\sad_0,+\infty)$.

\begin{proof}
For the identity between $\msf F_\beta$ and the polylog, we use the assumption $\sad>0$ to expand the log in series and interchange the order of summation and integration, obtaining
$$
\msf F_\beta(\sad)=-\sum_{k=1}^\infty \frac{\left(-\ee^{-\sad}\right)^k}{k}\int_0^\infty x^\beta \ee^{-kx}\dd x=-\Gamma(\beta+1)\sum_{k=1}^\infty \frac{\left(-\ee^{-\sad}\right)^k}{k^{\beta+2}},
$$
where the last identity follows performing the change of variables $kx=u$ in the integration. This last series is precisely the series expansion of $\Li_{\beta+2}(-\ee^{-\sad})$, and it also yields the claimed asymptotic estimate for $\msf F_\beta$.

To prove (ii), we start with the simpler case when $f(x)\equiv 1$ and $q(x)\equiv x$. In this case, we change variables $tx=v$ and obtain
$$
\int_0^\delta x^\beta \log(1+\ee^{-\sad-tx})\dd x=\frac{1}{t^{\beta+1}}\msf F_\beta(\sad)-\frac{1}{t^{\beta+1}}\int_{\delta t}^\infty v^\beta \log(1+\ee^{-\sad-v})\dd v.
$$
To estimate the last integral, we use the inequality $\log(1+u)\leq u$, which is valid for any $u\geq 0$, and obtain
$$
\int_{\delta t}^\infty v^\beta \log(1+\ee^{-\sad-v})\dd v\leq \ee^{-\sad_0}\int_{t\delta}^\infty v^\beta \ee^{-v}\dd v\leq \ee^{-\sad_0}\int_{0}^\infty v^\beta \ee^{-v}\dd v.
$$
The last integral above is independent of $t$, and shows that
\begin{equation}\label{eq:estLogint1}
\int_0^\delta x^\beta \log(1+\ee^{-\sad-tx})\dd x=\frac{1}{t^{\beta+1}}\msf F_\beta(\sad)+\Boh(\ee^{-\eta t}).
\end{equation}

For the general case, we fix a positive value $\delta'<\delta$ for which $q$ is a smooth bijection on $[0,\delta']$ and $f$ is smooth on $[0,\delta']$, and split the integral
$$
\msf I_t(\sad)=\int_0^{\delta'} x^\beta f(x)\log(1+\ee^{-\sad -tq(x)})\dd x+ \int_{\delta'}^\delta x^\beta f(x)\log(1+\ee^{-\sad -tq(x)})\dd x.
$$
Since $q$ is strictly positive on $[\delta',\delta]$, it has a minimum therein, say $q(x)\geq q_1>0$. Proceeding as before, it is immediate to see that the last integral decays exponentially fast in $t$, uniformly for $\sad\geq \sad_0$, and therefore
$$
\msf I_t(\sad)=\int_0^{\delta'} x^\beta f(x)\log(1+\ee^{-\sad -tq(x)})\dd x+\Boh(\ee^{-\eta t}),\quad t\to \infty,
$$
where $\eta>0$ is uniform for $\sad\geq \sad_0$. For the last integral, we change variables $tq(x)=u$, and with $\wt q\deff q^{-1}$ being the functional inverse of $q$ on $[0,\delta']$ we obtain
$$
\int_0^{\delta'} x^\beta f(x)\log(1+\ee^{-\sad -tq(x)})\dd x=
\int_0^{q(\delta')}\wt q(u)^\beta f\left(\wt q(u)\right)(\wt q)'(u)\log(1+\ee^{-\sad - tu})\dd u.
$$
A local analysis near the origin shows that
$$
\wt q(u)^\beta f\left(\wt q(u)\right)(\wt q)'(u)=\frac{f(0)}{m q_0^{(\beta+1)/m}}u^{(\beta+1-m)/m}+R(u)
$$
where the error term $R$ satisfies $|R(u)|\leq R_0 u^{(\beta+1)/m}$, $0\leq u\leq \delta'$, for some positive constant $R_0$ that depends on $f, \wt q,\beta$ but it is independent of $t,\sad$. Therefore,
\begin{multline*}
\int_0^{\delta'} x^\beta f(x)\log(1+\ee^{-\sad -tq(x)})\dd x=\\ \frac{f(0)}{m q_0^{(\beta+1)/m}} \int_0^{q(\delta')} u^{(\beta+1-m)/m} \log(1+\ee^{-\sad-tu})\dd u+\int_0^{q(\delta')} R(u)\log(1+\ee^{-\sad-tu})\dd u.
\end{multline*}
From \eqref{eq:estLogint1}, we obtain the estimates
$$
\int_0^{q(\delta')} u^{(\beta+1-m)/m} \log(1+\ee^{-\sad-tu})\dd u=\frac{1}{t^{\frac{\beta+1}{m}}}\msf F_{\frac{\beta+1-m}{m}}(\sad)+\Boh(\ee^{-\eta t})
$$
and
$$
\int_0^{q(\delta')} R(u)\log(1+\ee^{-\sad-tu})\dd u=\Boh\left(\frac{1}{t^{\frac{\beta+1}{m}+1}}\msf F_{\frac{\beta+1}{m}}(\sad)+\ee^{-\eta t}\right),
$$
as $t\to\infty$, completing the proof.
\end{proof}


\bibliographystyle{abbrv} 
\bibliography{bibliography.bib}

\end{document}